\documentclass[a4paper]{article}
\usepackage{amsfonts,amsmath,amssymb,amsthm}
\usepackage[T1]{fontenc}

\usepackage{calc}               

\usepackage{graphicx}           

\usepackage{amsmath, amsthm, amssymb}

\def\={=&\:}
\def\+{+&\:}
\def\-{-&\:}

\newcommand{\nn}{\nonumber}
\newcommand{\noi}{\noindent}

\usepackage{calc,amsmath}
\newlength\dlf
\newcommand\alignedbox[2]{
  &
  \begingroup
  \settowidth\dlf{$\displaystyle #1$}
  \addtolength\dlf{\fboxsep+\fboxrule}
  \hspace{-\dlf}
  \boxed{#1 #2}
  \endgroup
}

\def\sheaf{\mathcal}
\def\O{\mathcal{O}}

\def\IM{\text{Im}\:}
\usepackage{txfonts} 

\newcommand{\A}{\mathbf{A}_1}
\newcommand{\T}{\langle\mathbf{T}\rangle}
\newcommand{\1}{\langle\mathbf{1}\rangle}

\def\oneflat{\langle\mathbf{1}\rangle_{\text{flat}}}
\def\onesing{\langle\mathbf{1}\rangle_{\text{sing.}}}

\def\ranglesing{\rangle_{\text{sing}}}

\def\C{\mathbb{C}}
\def\CP{\mathbb{P}_{\C}}

\def\N{\mathbb{N}}

\def\Q{\mathbb{Q}}

\def\R{\mathbb{R}}

\def\Z{\mathbb{Z}}

\newcommand\vol{\text{vol}}

\newcommand{\Graph}{\text{Graph}}
\newcommand{\ess}{\text{ess}}
\newcommand{\Is}{\text{Is}}
\newcommand{\Gammared}{\Gamma^{\text{red}}}

\newcommand{\cdim}{\text{dim}_{\C}}

\newcommand{\ord}{\text{ord}\:}

\newcommand{\Rspan}{\text{span}_{\R}}
\newcommand{\cspan}{\text{span}_{\C}}
\newcommand{\End}{\text{End}}

\newcommand{\rechts}{\:\rightarrow\:}

\usepackage{graphicx} 

\def\id{{\rm id}}

\def\i{i\:}

\def\const{\text{const.}}
\newcommand{\eps}{\varepsilon}

\newcommand{\ta}{\tilde{a}}

\newcommand{\tI}{\tilde{I}}
\newcommand{\tnu}{\tilde{\nu}}
\newcommand{\tp}{\tilde{p}}
\newcommand{\tT}{\tilde{T}}

\def\tw{\tilde{w}}
 
\newcommand{\reg}{\text{reg.}}

\newtheorem{theorem}{Theorem}
\newtheorem{'theorem'}{''Theorem``}
\newtheorem{lemma}[theorem]{Lemma}
\newtheorem{proposition}{Propos.}
\newtheorem{cor}[proposition]{Corollary}
\newtheorem{definition}[proposition]{Definition}
\newtheorem{remark}[proposition]{Remark}
\newtheorem{example}[proposition]{Example}

\newtheorem{claim}{Claim}
\newtheorem{definition and theorem}{Definition an Theorem}
\theoremstyle{plain}
\newtheorem*{remark*}{Remark}
\newtheorem*{definition*}{Definition}
\newtheorem*{example*}{Example}

\newcommand{\D}{\mathcal{D}}

\begin{document}

\title{Rational CFTs on Riemann surfaces}
\author{Marianne Leitner$^*$, Werner Nahm$^{\diamond}$\\
Dublin Institute for Advanced Studies,\\
School of Theoretical Physics,\\
10 Burlington Road, Dublin 4, Ireland\\
$^*$\textit{leitner@stp.dias.ie},
$^{\diamond}$\textit{wnahm@stp.dias.ie}
}

\maketitle

\begin{abstract}
The partition function of rational conformal field theories (CFTs) on Riemann surfaces is expected to satisfy ODEs of Gauss-Manin type.
We investigate the case of hyperelliptic surfaces and derive the ODE system for the $(2,5)$ minimal model.
\end{abstract}

\tableofcontents

\pagebreak

\section{Introduction}

The present paper gives an ab initio mathematical introduction to rational conformal field theories (RCFT) on arbitrary genus $g\geq 1$ Riemann surfaces.
Our approach requires only three relatively simple and neat axioms.
The central objects are holomorphic fields and their $N$-point functions $\langle\phi_1\ldots\phi_N\rangle$. 
In order to actually compute these functions and more specifically the partition function $\1$ for $N=0$,
one has to study their behaviour under changes of the conformal structure.
This is done conveniently by first considering arbitrary changes of the metric.
Such a change of $\langle\phi_1\ldots\phi_N\rangle$ is described 
by the corresponding $(N+1)$-point function containing a copy of the Virasosoro field $T$. 
For this reason we have previously investigated the $N$-point functions of $T$ (rather than of more general fields) \cite{L:2013}.
In the present paper we study functions on the \textit{moduli space} $\mathcal{M}_g$, 
which is the space of all possible conformal structures on the genus $g$ surface. 
For the RCFTs one obtains functions which are meromorphic on a compactification of $\mathcal{M}_g$ or of a finite cover. 
We shall use that conformal structures occur as equivalence classes of metrics,
with equivalent metrics being related by Weyl transformations. The $N$-point functions of a CFT
do depend on the Weyl transformation, but only in a way which can be described by a universal automorphy factor.

For $g=1$ this has been made explicit in \cite{L:2017}. 
Less is known about automorphic functions for $g>1$. 
Our work develops methods in this direction. The basic idea is that many of the relevant functions
are algebraic.
In order to proceed step by step, 
we will restrict our investigations to the locus of hyperelliptic curves,
though the methods work in more general context as well.

For an important class of CFTs (the minimal models), 
the zero-point functions $\1$ will turn out to solve a linear differential equation
so that $\1$ can be computed for arbitrary hyperelliptic Riemann surfaces.
Since $\1$ is algebraic (namely a meromorphic function on a finite covering of the moduli space),
it is clear a priori that the equation can not be solved numerically only, but actually analytically.

\pagebreak 

\section{Notations and conventions}

In this paper, $0^0=1$.

For any category $\mathfrak{Cat}$, we denote by $|\mathfrak{Cat}|$ the set of objects of $\mathfrak{Cat}$.
For any pair of objects $O_1,O_2\in|\mathfrak{Cat}|$, we denote by $\text{Mor}_{\mathfrak{Cat}}(O_1,O_2)$ the set of morphisms $O_1\rechts O_2$ of $\mathfrak{Cat}$.

Let $\mathfrak{Diff}$ be the category of differentiable manifolds, and of smooth maps.
Here by smooth we mean $C^{\infty}$.

By a \textsl{Riemann surface} we mean a one-dimensional complex manifold. 
If $U\subseteq\C$ is an open subset, we say that a map $f:U\rechts\C$ is \textsl{conformal} if $f$ is biholomorphic on its image. 
Let $\mathfrak{Riem}$ be the category of (not necessarily compact) Riemann surfaces without boundary, and with conformal maps.

By a \textsl{Riemannian manifold} we mean a real smooth manifold equipped with a Riemannian metric,
i.e.\ a smooth positive section in the symmetric square of the cotangent bundle of the manifold.

Our surfaces are non-singular, i.e.\ they have no multiple ramification points.\\
\vspace{0.3cm}
We shall use the convention \cite{Z:1-2-3}
\begin{displaymath}
G_{2k}(z)
=\frac{1}{2}\sum_{n\not=0}\frac{1}{n^{2k}}+\frac{1}{2}\sum_{m\not=0}\sum_{n\in\Z}\frac{1}{(mz+n)^{2k}}\:, 
\end{displaymath}
and define $E_{2k}$ by $G_k(z)=\zeta(k)E_k(z)$ for $\zeta(k)=\sum_{n\geq 1}\frac{1}{n^k}$, so e.g.
\begin{align*}
G_2(z)\=\frac{\pi^2}{6}E_2(z)\:,\\
G_4(z)\=\frac{\pi^4}{90}E_4(z)\:,\\
G_6(z)\=\frac{\pi^6}{945}E_6(z)\:.
\end{align*}
Let $(q)_{n}:=\prod_{k=1}^n(1-q^k)$ be the $q$-Pochhammer symbol.
The \textsl{Dedekind $\eta$ function} is 
\begin{align*}
\eta(z)
:=q^{\frac{1}{24}}(q)_{\infty}
=q^{\frac{1}{24}}\left(1-q+q^2+q^5+q^7+\ldots\right)\:,\quad q=e^{2\pi\i z}\:.                                
\end{align*}

\vspace{0.3cm}

For $q=e^{2\pi i\tau}$, the theta functions $\vartheta_i(z,q)=\vartheta_i(z)$ at $z=0$ are given by
\cite[which however uses the convention $q=e^{\pi i\tau}$]{A-S},
\begin{align*}
\vartheta_1(0)
\=0\\
\vartheta_2(0)
\=2q^{1/8}\sum_{n=0}^{\infty}q^{\frac{1}{2}n(n+1)}
=2q^{1/8}(1+q+q^3+q^6+q^{10}+\ldots)\\
\vartheta_3(0)
\=1+2\sum_{n=1}^{\infty}q^{\frac{1}{2}n^2}
=1+2q^{\frac{1}{2}}+2q^2+2q^{\frac{9}{2}}+\ldots\\
\vartheta_4(0)
\=1+2\sum_{n=1}^{\infty}(-1)^nq^{\frac{1}{2}n^2}
=1-2q^{\frac{1}{2}}+2q^2-2q^{\frac{9}{2}}+\ldots
\end{align*}
We have the Jacobi identity:
\begin{align}\label{Jacobi id}
\vartheta_2^4+\vartheta_4^4
=\vartheta_3^4\:. 
\end{align}

\pagebreak

%

\section{Preliminaries}

We specify what we mean by a smooth category and introduce $F$-bundle functors.
Subsequently we remind the reader of the definition of primary fields and of $N$-point functions.  

\subsection{Categories with a differentiable structure}

Let $\mathfrak{Diff}$ be the category of differentiable manifolds.

\begin{definition}
A category $\mathfrak{Cat}$ \textbf{has a differentiable structure} 
if
\begin{enumerate}
\item 
$\forall\:O_1,O_2\in|\mathfrak{Cat}|$, $\forall\:\Sigma_1,\Sigma_2\in|\mathfrak{Diff}|$ 
and for any smooth map 
\begin{displaymath}
f:\Sigma_2\rechts\text{Mor}_{\mathfrak{Cat}}(O_1,O_2)\:,
\end{displaymath}
the composition 
\begin{displaymath}
f\circ\varphi:\Sigma_1\rechts \text{Mor}_{\mathfrak{Cat}}(O_1,O_2)
\end{displaymath}
is smooth,
$\forall\:\varphi\in \text{Mor}_{\mathfrak{Diff}}(\Sigma_1,\Sigma_2)$;
\item
$\forall\:O_1,O_2,O_3\in|\mathfrak{Cat}|$, $\forall\:\Sigma_1,\Sigma_2\in|\mathfrak{Diff}|$
and
for any pair of smooth maps 
\begin{displaymath}
f_i:\Sigma_i\rechts \text{Mor}_{\mathfrak{Cat}}(O_i,O_{i+1})\:,\quad i=1,2, 
\end{displaymath}
the induced map
\begin{displaymath}
\Sigma_1\times\Sigma_2\:\rechts\:\text{Mor}_{\mathfrak{Cat}}(O_1,O_3)
\end{displaymath}
defined by $(z_1,z_2)\:\mapsto\:f_2(z_2)\circ f_1(z_1)$ is smooth.
\end{enumerate}
\end{definition}

\begin{definition}
Let $\mathfrak{Cat}$ be a category with a differentiable structure.
A functor $F:\mathfrak{Cat}\rechts F(\mathfrak{Cat})$ is \textbf{smooth} 
if
\begin{enumerate}
\item
$F(\mathfrak{Cat})$ has a differentiable structure,
\item
$\forall\:O_1,O_2\in|\mathfrak{Cat}|$,
\begin{displaymath}
\text{Mor}_{\mathfrak{Cat}}(O_1,O_2)\rechts \text{Mor}_{F(\mathfrak{Cat})}(F(O_1),F(O_2)) 
\end{displaymath}
is smooth.
\end{enumerate}
\end{definition}

\subsection{$F$-bundle functors}

Let $F$ be an infinite dimensional $\C$-vector space endowed with an ascending filtration 
by finite-dimensional subvector spaces
\begin{displaymath}
F_0\subset F_1\subset F_2\subset\ldots\:,
\quad
F=\cup_{i\in\N_0}F_i\:. 
\end{displaymath}
Equip $F$ with 
the finest topology 
for which the inclusions $F_i\subset F$ for $i\geq 0$ are continuous.
Equivalently, a series $(\mathbf{x}^i)_{i\in\N}$ with $\mathbf{x}^i\in F$ for $i\in\N$ converges to $\mathbf{x}\in F$ iff 
\begin{enumerate}
\item 
$\exists\:m_0\in\N$ such that $\mathbf{x}\in F_{m_0}$ and $\mathbf{x}^i\in F_{m_0}$, $\forall\:i\in\N$,
\item
$(\mathbf{x}^i)_{i\in\N}$ converges to $\mathbf{x}$ in $F_{m_0}$.
\end{enumerate}

\begin{definition}
We refer to $F$ as a \textbf{quasi-finite $\C$-vector space}.
\end{definition}

\noi
The filtration induces a grading
\begin{displaymath}
F
=\oplus_{i\in\N_0}F_i/F_{i-1}
\end{displaymath}
of $F$ into finite-dimensional complex subvector spaces.
Let
\begin{displaymath}
\End_{\text{grad}}(F)\cong\oplus_{i}\End(F_i/F_{i-1})\:
\end{displaymath}
be the ring of endomorphisms of $F$ that respect this grading.
These are the only endomorphisms of $F$ we will consider. 
In a basis of $F$, 
an element $A\in\End_{\text{grad}}(F)$ can be written as a block diagonal matrix 
in which the $i$'th block defines an element $A_i\in\End(F_i/F_{i-1})$.
$A$ is smooth (we mean $C^{\infty}$) 
if for every $i\in\N_0$, $A_i$ is smooth on the real vector space underlying $F_i/F_{i-1}$. 

\begin{definition}
Let $F$ be a quasi-finite $\C$-vector space.
\begin{enumerate}
\item 
By a \textbf{vector bundle} $\sheaf{E}$ \textbf{with fiber $F$}
we mean a family of pairs $(\sheaf{E}_i,\imath_i)$ for $i\in\N_0$,
where $\sheaf{E}_i$ is a vector bundle with standard fiber $F_i$, 
and $\imath_i:\sheaf{E}_i\subset\sheaf{E}_{i+1}$ is an inclusion of vector bundles.  
\item
For any two vector bundles $\sheaf{E},\sheaf{E}'$ with fiber $F$,
a \textbf{morphism} $f:\sheaf{E}\rechts\sheaf{E}'$ \textbf{of vector bundles with fiber $F$}
is a family of vector bundle morphisms $f_i:\sheaf{E}_i\rechts\sheaf{E}'_i$ with
\begin{displaymath}
f_{i+1}|_{\sheaf{E}_i}=f_i 
\end{displaymath}
for $i\in\N_0$, 
where $\sheaf{E}_i$ and $\sheaf{E}'_i$ are the vector bundles with standard fiber $F_i$ 
defined by $\sheaf{E}$ and $\sheaf{E}'$, respectively.
\item
In particular, if $\sheaf{E}$ and $\sheaf{E}'$, respectively, is a vector bundle over a smooth manifold, 
then $f$ is \textbf{smooth} if $f_i$ is smooth for every $i\in\N_0$.
\end{enumerate}
We define $\mathfrak{Vec}(F)$ to be the category of vector bundles with fiber $F$, and with smooth morphisms.
The objects in $|\mathfrak{Vec}(F)|$ are referred to as \textbf{$F$-bundles}.
\end{definition}

The morphism set of the category $\mathfrak{Riem}$ and $\mathfrak{Vec}(F)$, respectively, has a natural manifold structure: 

\begin{proposition}
For $\Sigma,\Sigma'\in|\mathfrak{Riem}|$,
the set $\text{Mor}_{\mathfrak{Riem}}(\Sigma,\Sigma')$ is naturally an infinite dimensional complex manifold. 
For $\sheaf{E},\sheaf{E}'\in|\mathfrak{Vec}(F)|$,
we have $\text{Mor}_{\mathfrak{Vec}(F)}(\sheaf{E},\sheaf{E}')\in|\mathfrak{Diff}|$ in a natural way. 

\end{proposition}

\begin{proof}
Let $\Sigma,\Sigma'\in|\mathfrak{Riem}|$, and let $M$ be any complex manifold.
We say that $\varphi:M\rechts \text{Mor}_{\mathfrak{Riem}}(\Sigma,\Sigma')$ is \textsl{holomorphic}
if the induced map $\varphi_1:M\times\Sigma\rechts\Sigma'$ defined by 
$\varphi_1(p,q):=\left(\varphi(p)\right)(q)$ for $p\in M,\:q\in\Sigma$
is holomorphic. 
The proof of the statement for $\mathfrak{Vec}(F)$ is analogous.
\end{proof}

In the following, we shall treat $\text{Mor}_{\mathfrak{Riem}}(\Sigma_1,\Sigma_2)$ as a smooth manifold (by forgetting about its complex structure).

\begin{definition}
For any quasi-finite $\C$-vector space $F$, an \textbf{$F$-bundle functor} is a covariant functor 
\begin{displaymath}
\Phi_F:\mathfrak{Riem}\rechts\mathfrak{Vec}(F)\:
\end{displaymath}
with the following properties: 
\begin{itemize}
\item
$\forall\:\Sigma\in|\mathfrak{Riem}|$, $\Phi_F(\Sigma)=:\sheaf{F}_{\Sigma}$ is a vector bundle over $\Sigma$,
\item
$\Phi_F$ is compatible with restrictions: if $U\subset\Sigma$ then $\sheaf{F}_U=\sheaf{F}_{\Sigma}|_U$,
\item 
$\forall\:\Sigma_1,\Sigma_2\in|\mathfrak{Riem}|$,
$\Phi_F$ defines an element in 
\begin{displaymath}
\text{Mor}_{\mathfrak{Diff}}(\text{Mor}_{\mathfrak{Riem}}(\Sigma_1,\Sigma_2),\text{Mor}_{\mathfrak{Vec}(F)}(\sheaf{F}_{\Sigma_1},\sheaf{F}_{\Sigma_2}))\:. 
\end{displaymath}
\end{itemize}
\end{definition}

\begin{example}\label{example: T functor}
The tangent functor $T:\mathfrak{Diff}\rechts\mathfrak{Diff}$ has precisely the above listed properties:
For $M\in|\mathfrak{Diff}|$, $TM$ defines the the tangent bundle over $M$, and if $f\in \text{Mor}_{\mathfrak{Diff}}(M,N)$, 
we have $Tf=df\in\text{Mor}_{\mathfrak{Diff}}(TM,TN)$.
Moreover,
if $(U,z)$ is a chart on $M$, 
$Tz=dz$ defines a nowhere vanishing section in the cotangent bundle $T^*U$, 
and thus a trivialisation $TU\cong U\times\C$.
\end{example}

The latter observation is actually a general feature.

\begin{proposition}
$\Phi_F$ defines a canonical trivialisation of $\sheaf{F}_{\C}=\Phi_F(\C)$ with fiber $\sheaf{F}_{\C,0}=F$.
\end{proposition}

\begin{proof}
All conformal self-maps of $\C$ are affine linear.
For $z\in\C$, let $t_z:\C\rechts\C$ be the translation by $z$.
The induced morphism $\Phi_F(t_z)$ maps $F=\sheaf{F}_{\C,0}$ isomorphically to $\sheaf{F}_{\C,z}$.
The map $\C\times F\rechts\sheaf{F}_{\C}$ defined by $(z,\varphi)\mapsto\left(\Phi_F(t_z)\right)(\varphi)\in\sheaf{F}_{\C,z}$ 
is invertible.
\end{proof}

If $U\in|\mathfrak{Riem}|$ has coordinate $z:U\rechts\C$, 
$\Phi_F(U)$ trivialises in a way determined by $\Phi_F(z)$.
For $(p,\varphi)\in\C\times F$, the corresponding element in $\sheaf{F}_U$ is
\begin{displaymath}
\varphi_z(p)=(\Phi_F(z))^{-1}(p,\varphi)\:.
\end{displaymath}
Abusing notations, we shall simply write $\varphi(z)$ where we actually mean $\varphi_z(p)$.
(This will entail notations like $\hat{\varphi}(\hat{z})$ instead of $\varphi_{\hat{z}}(p)$ etc.)

We shall only consider bundles that lie in $|\mathfrak{Vec}(F)|$. 


\subsection{Primary fields}

Let $\O_{\C}$ be the sheaf of germs $\langle U,f\rangle$
which are represented by pairs $(U,f)$ for some open set $U\subseteq\C$ and some conformal map $f:U\rechts\C$.
Let $\O_{\C,0}$ be the fiber of germs in $\O_{\C}$ which are defined at the origin in $\C$,
and let
\begin{displaymath}
G:=\{\langle U,f\rangle\in\O_{\C,0}\:|\:f(0)=0\}
\:. 
\end{displaymath}
It is easy to see that $G$ is a group under pointwise composition, with identity element $\langle\C,\text{id}\rangle$.                                                                                                
$G$ is actually a Lie group \cite[p.\ 267]{Neeb-Pianzola}. 
$G$ is a real manifold that admits no complexification.

The Lie algebra $\mathfrak{g}$ of $G$ can be identified 
with the Lie algebra of germs of holomorphic vector fields on $\C$ which vanish at the origin
\cite{BKS:1989}, 
\begin{displaymath}
\mathfrak{g}
=\Rspan\{\langle U,\ell_n\rangle\}_{n\geq 0}\:,
\end{displaymath}
where $\ell_n=-z^{n+1}\partial_z$.
These polynomial vector fields define diffeomorphisms of $S^1$ that extend to the unit disc $\{z\in\C|\:|z|\leq 1\}$.
Over $\C$, the vector fields $\ell_n$ for $n\in\Z$ generate the Witt algebra. 
\cite[p.\ 34]{Schottenloher:2008}.
The infinite-dimensional Lie group $\text{Diff}(S^1)$ of orientation preserving diffeomorphisms of $S^1$ 
has no complexification.


\begin{proposition}\label{propos: F-functor defines rep of G on F}
$\Phi_F$ defines a representation of $G$ on $F$.
\end{proposition}

\begin{proof}
For any pair of representatives $(U,f)$ and $(V,g)$ of a germ $\langle U,f\rangle\in G$, 
the corresponding bundle maps $\Phi_F(f)$ and $\Phi_F(g)$ induce the same automorphism of $F=\sheaf{F}_{\C,0}$.
 
\end{proof}

By assumption, the representation decomposes into finite-dimensional subrepresentations, 
corresponding to the grading of $F$.
The corresponding representation of the Lie algebra 
\begin{displaymath}
\mathfrak{g}\:\rechts\:\End_{\text{grad}}(F)
\end{displaymath}
extends to an $\R$-linear representation $L+\bar{L}$ 
of the complexified Lie algebra $\mathfrak{g}_{\C}=\mathfrak{g}\otimes\C$,
where $L$ and $\bar{L}$ are complex linear and a complex antilinear Lie algebra homomorphisms, respectively.
For $n\geq 0$, let $L_n$ and $\bar{L}_n$
be the image of $\langle U,\ell_n\rangle$ under $L$ and $\bar{L}$, respectively,
in $\End_{\text{grad}}(F)$.
$\{L_n\}_{n\geq 0}$ satisfy a Lie subalgebra of the Witt algebra
\begin{align}\label{Witt type algebra, but for n geq 0 only}
[L_n,L_m]
=(n-m)L_{n+m}\:.
\end{align}
$\{\bar{L}_n\}_{n\geq 0}$ define an isomorphic Lie algebra, and $[\bar{L}_n,L_m]=0$ for $n,m\geq 0$.

For $n\geq 0$, $L_n+\bar{L}_n$ and $i(L_n-\bar{L}_n)$ 
represent the generator of the infinitesimal transformations
\begin{align*}
\quad z\:\mapsto&\:\exp(-\eps z^{n+1}\partial_z)\approx z(1-\eps z^n)\:,\\
\quad z\:\mapsto&\:\exp(-i\eps z^{n+1}\partial_z)\approx z(1-i\eps z^n)\:,\quad\eps>0\:,\quad z\in\C\:,
\end{align*}
respectively. ($\bar{z}$ is treated as an independent variable and will be disregarded.)
In particular, $L_0+\bar{L}_0$ and $i(L_0-\bar{L}_0)$ 
represent the generator of the infinitesimal dilation and rotation, respectively, in a one-dimensional complex vector space.

\begin{proposition}
Let $V$ be a complex representation of $G$, $\cdim V=1$, such that 
\begin{align}\label{eqs: eigenspace of L0 and bar(L)0}
L_0|_V=h\cdot\id_V\:,\quad\bar{L}_0|_V=\bar{h}\cdot\id_V\:,
\end{align}
for some pair of numbers $h,\bar{h}\in\R$. 
Then $h-\bar{h}\in\Z$, and   
\begin{align}\label{Ln is zero on one-dimensional subspace} 
L_n|_V=0\quad\text{for $n>0$}\:. 
\end{align}
\end{proposition}

\begin{proof}
By eq.\ (\ref{eqs: eigenspace of L0 and bar(L)0}), $L_0-\bar{L}_0=h-\bar{h}$ in $V$.
Now $\exp(i\eps(L_0-\bar{L}_0))$ defines a rotation by $\eps$ in $V$, 
so taking $\eps=2\pi$ shows that $h-\bar{h}\in\Z$.
Now let $V=\cspan\{v\}$ for some simultaneous eigenvector $v\not=0$ of $L_0$ and $\bar{L}_0$.
Since $[L_n,L_0]\not=0$ for $n>0$, we have $L_nv=0$ in this case.
\end{proof}

\begin{definition}
An element $\varphi\in F$ has the property of being \textbf{primary} 
if $\text{span}_{\C}\langle\varphi\rangle$ defines a one-dimensional representation of $G$. 
\end{definition}

\noi
We give a converse to Proposition \ref{propos: F-functor defines rep of G on F}.

\begin{proposition}\label{propos: F bundle characterised by rep of G}
$F$-bundle functors $\Phi_F$ are characterised, up to bundle isomorphisms, by representations of $G$.
\end{proposition}

\begin{proof}
Let $V$ be a complex one-dimensional representation of $G$ with property (\ref{eqs: eigenspace of L0 and bar(L)0}).
Suppose $V=\cspan\{v\}$ for some vector $v\in F$.
By definition
\begin{align}\label{definition of filtration}
v\in F_n
\end{align}
if $h-\bar{h}\leq n$. This defines a grading $F=\oplus_{i\in\N_0}F_i/F_{i-1}$.
Since $\Phi_F$ is compatible with restrictions, 
it suffices to define the functor locally. 
For two contractible sets $U,V\in|\mathfrak{Riem}|$ and for $f\in\text{Mor}_{\mathfrak{Riem}}(U,V)$, 
$\Phi_F(U)\cong U\times F$,
so $\Phi_F(f)$ is determined by $f$ and the representations $F_i/F_{i-1}\rechts F_i/F_{i-1}$ of $G$, for $i\in\N_0$.
\end{proof}

\noi
We shall come back to our standard example and consider tensorial powers of tangent line bundle $T\C$ and its complex conjugate $\overline{T\C}$.

\begin{proposition}
Every rank-one subbundle of $\Phi_F(\C)$ is isomorphic to a bundle of the form
\begin{displaymath}
(T\C)^{h-\bar{h}}\otimes(T\C\otimes\overline{T\C})^{\bar{h}}\:
\end{displaymath}
with $\bar{h}\in\R^+_0$ and $h-\bar{h}\in\Z$.
We refer to this bundle as the \textbf{$(h,\bar{h})$-bundle} and write
\begin{displaymath}
(T\C)^h\otimes(\overline{T\C})^{\bar{h}}\:. 
\end{displaymath}
\end{proposition}

Note that the latter should be taken as a notation only. 
Since under coordinate change $z\mapsto w$, $T\C$ has the holomorphic transition function $\frac{dw}{dz}$,
the transition function of $T\C\otimes\overline{T\C}$ is $\left|\frac{dw}{dz}\right|^2$,
which is real and positive. Thus it has a well-defined logarithm.

\begin{proof}
Example \ref{example: T functor} shows that $T\C$ and thus every $(h,\bar{h})$-bundle defines a rank-one subbundle of an $F$-bundle.
To prove the converse, it suffices by Proposition \ref{propos: F bundle characterised by rep of G}
to show that every $(h,\bar{h})$-bundle, or in fact the differential functor $T$
defines a one-dimensional representation of $G$ which isomorphic to (\ref{eqs: eigenspace of L0 and bar(L)0}) 
and (\ref{Ln is zero on one-dimensional subspace}).

For infinitesimal $\eps>0$ and for $n\geq 0$, define $F_n:\C\rechts\C$ by $F_n(z)=z(1+\eps z^n)$.
$F_n$ defines an element in $G$. 
Since
\begin{displaymath}
T_0(F_m\circ F_n)
=d(F_m\circ F_n)_0
=(F_m'\circ F_n)(0)F_n'(0)\:dz_0
=F_m'(0)F_n'(0)\:dz_0\:,
\end{displaymath}
$T$ defines a one-dimensional representation of $G$
by 
\begin{displaymath}
F_n\mapsto F_n'(0)\:. 
\end{displaymath}
Since $F_n'(0)=\exp(\eps\delta_{n,0})$, the representation is isomorphic to that generated by $L_n$ for $n\geq 0$, in $V$.
We have a similar description for $\overline{T}$ and anti-holomorphic functions,
and obtain a representation isomorphic to that generated by $\bar{L}_n$ for $n\geq 0$. 
\end{proof}


\subsection{States and $N$-point functions}

For $N\geq 0$, define
\begin{align*}
\sheaf{M}_{g,N}\quad
&\text{the moduli space of compact Riemann surfaces $\Sigma$ of genus $g$}\\
&\text{with $N$ different distinguished points $p_1,\ldots, p_N\in\Sigma$;}\\
\sheaf{M}_{g,N}^F\quad
&\text{the moduli space of Riemann surfaces $\Sigma\in\sheaf{M}_{g,N}$, on which for $1\leq i\leq N$,}\\
&\text{there is a copy of $F$ attached to $p_i$ with one marked point $\varphi_i\in F$.}
\end{align*}
Let $o_N:\sheaf{M}_{g,N}^F\rechts\sheaf{M}_{g,0}$ be the forgetful map for $N>0$ and the identity otherwise.
Conversely, from $\Sigma\in\sheaf{M}_{g,0}$ we recover an element in $\sheaf{M}_{g,1}^F$ ($N=1$)
by choosing a point $p\in\Sigma$ and marking an element $\varphi(p)$ in the fiber $\sheaf{F}_{\Sigma,p}$ of $\sheaf{F}_{\Sigma}=\Phi_F(\Sigma)$.  
We may view $\sheaf{F}_{\Sigma}$ as the set of elements in $\sheaf{M}_{g,1}^F$ 
that corresponds to all possible markings,
\begin{displaymath}
\sheaf{F}_{\Sigma}
\cong\:o_1^{-1}\Sigma\:. 
\end{displaymath}
This description allows to vary the markings $p\in\Sigma$ and $\varphi(p)\in\sheaf{F}_{\Sigma,p}$ in a continuous way.

\noi
We will also have to have to discuss Riemannian metrics which are compatible with a given complex structure.
Let $S$ be a compact oriented genus $g$ surface with a differentiable structure, (determined up to diffeomorphism).
Let $\text{Met}_g(S)$ be the additive semi-group of Riemannian metrics $G$ on $S$ or equivalently,
the set of Riemannian surfaces $\tilde{S}$ diffeomorphic to $S$.
(We shall use the two descriptions interchangeably.)
We have the well-known isomorphism \cite{B-M:1986}
\begin{displaymath}
\sheaf{M}_{g,0}\cong\:\text{Met}_g(S)/\text{Weyl}\ltimes\text{Diffeo}\:. 
\end{displaymath}
The map $\tilde{o}:\text{Met}_g(S)\rechts\sheaf{M}_{g,0}$ is given 
by forgetting about the specific Riemannian metric $G$ on $\tilde{S}\in\text{Met}_g(S)$ and keeping only its conformal class $[G]$. 

\begin{definition}
For any $N\geq 0$, we define
\begin{align*}
M_{g,N}^F
:\=\{(\tilde{S},\Sigma)\in\text{Met}_g(S)\times\sheaf{M}_{g,N}^F|\:
\text{$\tilde{o}\tilde{S}=o_N\Sigma$ in $\sheaf{M}_{g,0}$}\}\:. 
\end{align*}
For $N=0$, we write $M_{g,0}$.
An \textbf{$N$-point function} is a map
\begin{displaymath}
\langle\quad\rangle:\quad M_{g,N}^F\rechts\C\:                                  
\end{displaymath}
which is 
\begin{itemize}
\item
continuous as a function of $\tilde{S}\in\text{Met}_g(S)$, or of the metric $G$ on $S$,
\item
smooth as a function on $o_N\Sigma\in\sheaf{M}_{g,0}$, 
and $N$-linear on the fibers $F$ of $\Sigma\in\sheaf{M}_{g,N}^F$.
\end{itemize}
A \textbf{state} is a family of $N$-point functions for $N\in\N_0$.
\end{definition}

Since the marked points $p_1,\ldots, p_N$ on the Riemann surface are all distinct, 
for the purpose of local variations we replace an element $\Sigma\in\sheaf{M}_{g,N}^F$ with markings at $(p_1,\varphi_1(p_1)),\ldots,(p_N,\varphi_N(p_N))$
with the $N$-fold symmetric fiber product of elements in $\sheaf{M}_{g,1}^F$ defined on $o_N\Sigma$, 
where for  $1\leq i\leq N$ , the $i$th factor is marked at $(p_i,\varphi_i(p_i))$.

More specifically, suppose $\Sigma\in\sheaf{M}_{g,0}$. 
We restrict the $N$-fold Cartesian product $\text{sym}^{\times N}(\Sigma)$ of $\Sigma$
to the locus 
\begin{displaymath}
\text{sym}_{\text{restr}}^{\times N}(\Sigma)
:=\text{sym}^{\times N}(\Sigma)\setminus\{(z_1\ldots,z_N)|\:z_i=z_j\:\text{for some $i\not=j$}\}\: 
\end{displaymath}
off partial diagonals.
Moreover, let $\sheaf{F}_{\Sigma}=\Phi_F(\Sigma)$ with fiber $\sheaf{F}_{\Sigma,p}$ at $p\in\Sigma$,
and let $\text{sym}^{\boxtimes N}(\sheaf{F}_{\Sigma})$ be its $N$-fold symmetric fiber product.
We define $\text{sym}^{\boxtimes N}_{\text{restr}}(\sheaf{F}_{\Sigma})$ to be the set
obtained by restricting $\text{sym}^{\boxtimes N}(\sheaf{F}_{\Sigma})$ to the set of tensor products $\sheaf{F}_{\Sigma,p_1}\otimes\ldots\otimes\sheaf{F}_{\Sigma,p_N}$
with $(p_1,\ldots,p_N)\in\text{sym}_{\text{restr}}^{\times N}(\Sigma)$. 
Thus
\begin{displaymath}
\text{sym}^{\boxtimes N}_{\text{restr}}(\sheaf{F}_{\Sigma})
\cong\:
o_N^{-1}\Sigma\:. 
\end{displaymath}
To conclude, let $(G,\Sigma)\in M_{g,0}$ and let
\begin{displaymath}
P:\text{sym}^{\boxtimes N}_{\text{restr}}(\sheaf{F}_{\Sigma})
\:\rechts\:
\text{sym}^{\times N}_{\text{restr}}(\Sigma)\:
\end{displaymath}
be the projection onto the base points. 
An $N$-point function on a Riemann surface $\Sigma$ takes values $\langle\varphi\rangle_G$, 
where $\varphi\in\text{sym}^{\boxtimes N}_{\text{restr}}(\sheaf{F}_{\Sigma})$
and $P(\varphi)\in\text{sym}_{\text{restr}}^{\times N}(\Sigma)$.


\section{Definition of a rational conformal field theory}

Three axioms are required to define the notion of a rational conformal field theory.

\subsection{Axiom 1: Invariance under diffeomorphisms 
that preserve the conformal structure close to the respective base points}\label{Axiom 1}

Using the previous notations,
suppose $\Sigma\in\sheaf{M}_{g,0}$ and $S$ is the oriented surface underlying $\Sigma$.
Let $f$ be an infinitesimal automorphism on
\begin{displaymath}
\text{Met}_g(S)\times\text{sym}^{\boxtimes N}_{\text{restr}}(\sheaf{F}_{\Sigma})\:.
\end{displaymath}
On the first factor, $f$ defines a diffeomorphic automorphism on $\text{Met}_g(S)$ given by $G\mapsto G+\delta G$.
Call this automorphism $\chi$. 
On the second factor, $f$ acts by $\varphi\mapsto\varphi+\delta_{\hat{f}}\:\varphi$,
for some map $\hat{f}$.
Our approach to CFT is through $N$-point functions $\langle\varphi\rangle_G$ for $\varphi\in\text{sym}^{\boxtimes N}_{\text{restr}}(\sheaf{F}_{\Sigma})$
which restricts $\text{Met}_g(S)$ to metrics $G$ on $S$ with $(G,\Sigma)\in\text{M}_{g,0}$.
Moreover, as we want to understand the change of $\langle\varphi\rangle_G$ under smooth variations of $G$,
we only admit a specific class of diffeomorphisms which depends on the tuple 
$P(\varphi)=(p_1,\ldots,p_N)\in\text{sym}^{\times N}_{\text{restr}}(\Sigma)$:
We require that for $i=1,\ldots,N$
there exists a neighbourhood $U_i$ of $p_i$ in $S$ such that after restriction to $U_i$, 
$\chi(G)|_{U_i}$ defines a metric in the conformal class defined by $G$ or $\Sigma$.
This allows to define the derivative of $\langle\varphi\rangle_G$ w.r.t.\ the metric $G$:
\begin{align}\label{eq: expansion variation of N-pt function at G+delta G in powers of delta G}
\langle\varphi\rangle_{G+\delta G}
=:\langle\varphi\rangle_G+\delta_G\:\langle\varphi\rangle_G+O((\delta G)^2)
\:,
\end{align}
It is easy to check that the map on the $N$-point function induced by $f$ is given by
\begin{align*}
\langle\varphi\rangle_G
\mapsto
\langle\varphi+\delta_{\hat{f}}\:\varphi\rangle_{G+\delta G}
\=\langle\varphi\rangle_{G+\delta G}+\langle\varphi+\delta_{\hat{f}}\:\varphi\rangle_{G}-\langle\varphi\rangle_{G}
+O(\delta_{\hat{f}}\varphi\cdot\delta G)\nn
\:,
\end{align*}
using the defining properties of the state.
The additive change to $\langle\varphi\rangle_G$ induced by $f$ is
\begin{align}\label{diffeom derivative}
\Delta_f\langle\varphi\rangle_G
:\=\langle\varphi+\delta_{\hat{f}}\:\varphi\rangle_{G+\delta G}-\langle\varphi\rangle_G
\:.
\end{align}
Given a diffeomorphic automorphism $f$ of $S$, let $\chi_f:\text{Met}_g(S)\rechts\text{Met}_g(S)$ be the natural induced diffeomorphism.
By assumption, for $i=1,\ldots,N$, $\chi_f$ preserves the conformal structure on $U_i$. 
Thus $f$ gives rise to germs $\langle U_i,f_i\rangle$ of conformal maps close to $p_i$,
and thus by Proposition \ref{propos: F-functor defines rep of G on F}, to an automorphism $\Phi_F(f_i)$ of $\sheaf{F}_{\Sigma,p_i}$.
We postulate that we have in eq.\ (\ref{diffeom derivative}),
\begin{displaymath}
\Delta_{f}\langle\varphi\rangle_G=0\:. 
\end{displaymath}
This means that for $\varphi=\varphi_1(p_1)\otimes\ldots\otimes\varphi_N(p_N)$,
\begin{displaymath}
\langle\Phi_F(f_1)\varphi_1(p_1)\otimes\ldots\otimes\Phi_F(f_N)\varphi_N(p_N)\rangle_{\chi_f(G)}
=\langle\varphi_1(p_1)\otimes\ldots\otimes\varphi_N(p_N)\rangle_G\:. 
\end{displaymath}

\subsection{Axiom 2: Einstein derivative}\label{Axiom 2}

Let $\Sigma\in\sheaf{M}_{g,0}$ and $\sheaf{F}_{\Sigma}=\Phi_F(\Sigma)$.
Let $S$ be the oriented surface underlying $\Sigma$, with tangent bundle $TS$.
Denote by $\text{sym}^{\otimes 2}(T_{\R}\Sigma)$ the symmetric $2$-fold tensor product of $TS$.
We postulate that to every metric $G\in\text{Met}(S)$,
there exists an element $T\in\Gamma(\Sigma,\sheaf{F}_{\Sigma}\otimes\text{sym}^{\otimes 2}(TS))$
such that for $\varphi\in\text{sym}^{\boxtimes N}_{\text{restr}}(\sheaf{F}_{U})$,
the derivative $\delta_G$ defined by (\ref{eq: expansion variation of N-pt function at G+delta G in powers of delta G})
is given by
\begin{align}\label{eq: variation formula}
\delta_G\langle\varphi\rangle_G
=\iint\langle(T,\delta G)\:\varphi\rangle_G\:d\text{vol}_G\:.
\end{align}
Here $(\:,\:)$ is the dual pairing,
and $d\vol_G=\sqrt{|\det G_{\mu\nu}|}\:dx^0dx^1$ is the coordinate independent volume form.

\subsection{Axiom 3: Trace Anomaly}\label{Axiom 3}

Let $G\in\text{Met}(S)$, and let $T\in\Gamma(\Sigma,\sheaf{F}_{\Sigma}\otimes\text{sym}^{\otimes 2}(TS))$ be the corresponding element from Axiom \ref{Axiom 2}.
Let $\mathcal{R}_G$ be the scalar curvature of the Levi-Civit\` a connection on $S$,
\begin{displaymath}
\mathcal{R}_G=G^{\kappa\lambda}R_{\kappa\lambda}\:.
\end{displaymath}
Let $T$ be the field from Subsection \ref{Axiom 2}, and let $(.,.)$ be the dual pairing.
We postulate that
\begin{displaymath}
(T,G)=-\frac{c}{48\pi}\:\mathcal{R}_G\:, 
\end{displaymath}
where $c\in\R$ is the central charge.

\subsection{Definition of rational Conformal Field Theories}

\begin{definition}
Let $F$ be a quasi-finite vector space.
A (rational) \textbf{conformal field theory (CFT)} is a pair $(\Phi_F,\langle\:\rangle)$ where $\Phi_F$ is an $F$-functor and $\langle\:\rangle$ is a state
such that Axiom \ref{Axiom 1}, Axiom \ref{Axiom 2} and Axiom \ref{Axiom 3} are valid.
\end{definition}

\section{Immediate Consequences of the Axioms}

\subsection{Conservation Law}

According to Noether's theory, 
every continuous symmetry in a field theory gives rise to a conserved quantity.
In a CFT, $N$-point functions are invariant under certain under diffeomorphisms (Axiom \ref{Axiom 1}), 
and the corresponding conserved quantity is the energy momentum tensor.
\begin{displaymath}
\partial_{\mu}T^{\mu\nu}
=0\:.
\end{displaymath}
We shall explain the relationship with the Virasoro field $T(z)$ on $\Sigma\in\sheaf{M}_{g,0}$ and the induced conservation law.
Let $S$ be the oriented surface underlying $\Sigma$.
Let $G\in\text{Met}(S)$ and let $T\in\Gamma(\Sigma,\sheaf{F}_{\Sigma}\otimes\text{sym}^{\otimes 2}(TS))$ be the corresponding Virasoro field.
On any coordinate neighbourhood $U\subset\Sigma$, it is given by the energy momentum tensor
\begin{displaymath}
T|_U
=\sum_{\mu,\nu=0}^1T^{\mu\nu}\frac{\partial}{\partial x^{\mu}}\frac{\partial}{\partial x^{\nu}}
\:.
\end{displaymath}
Changing to complex coordinates $z=x^0+ix^1$ and $\bar{z}=x^0-ix^1$, we have \cite{Blum:2009}
\begin{displaymath}
T_{zz}=\frac{1}{4}(T_{00}-2iT_{10}-T_{11})\:.
\end{displaymath}

\begin{lemma}
$T_{\mu\nu}$ satisfies the conservation law
\begin{displaymath}\label{conservation law for T}
\nabla_{\mu}{T^{\mu}}_z
=0\:.
\end{displaymath}
Here $\nabla$ is the covariant derivative of the Levi-Civit\` a connection on $S$ w.r.t.\ $G_{\mu\nu}$.
\end{lemma}

\begin{proof}
We have
\begin{displaymath}
\nabla_{\mu}{T^{\mu}}_z
=\nabla_z{T^z}_z+\nabla_{\bar{z}}{T^{\bar{z}}}_z
\:.\nn
\end{displaymath}
${T^z}_z$ transforms like a scalar \cite{F:1982}, so $\nabla_z{T^z}_z=\partial_z{T^z}_z$. 
Moreover, $\nabla_{\mu}G^{\mu\nu}=0$ so
\begin{displaymath}
\nabla_{\bar{z}}{T^{\bar{z}}}_z
=G^{z\bar{z}}\partial_{\bar{z}}T_{zz}
\:.
\end{displaymath}
This vanishes, since $T_{zz}$ takes values in a holomorphic line bundle \cite{F:1982}.  
We conclude that
\begin{displaymath}
\nabla_{\mu}{T^{\mu}}_z
=\partial_z{T^z}_z+G^{z\bar{z}}\partial_{\bar{z}}T_{zz}
=0\:.
\end{displaymath}
\end{proof}

The Virasoro field does not depend on the specific metric on $\Sigma$,
but only on the conformal class.
$\langle T(x)\rangle\:(dx)^2$ defines a $0$-cochain in the sheaf cohomology group 
of sheaf of holomorphic sections in $(T^*\Sigma)^{\otimes 2}$ associated to a complex analytic coordinate covering,
but fails to satisfy the cocycle condition (i.e.\ to define a quadratic differential)
when the coordinate changes induce the addition of a Schwarzian derivative term. 
The Schwarzian derivative, however, satisfies the $1$-cocycle condition, 
and $\langle T(x)\rangle(dx)^2$ is known as projective connection.

\begin{lemma}\cite{EO:1987}
Suppose $\Sigma$ has scalar curvature $\mathcal{R}=\const$
Let
\begin{align}\label{definition of T(z)}
\frac{1}{2\pi}T(z)
:=T_{zz}-\frac{c}{24\pi}t_{zz}\:,\
\end{align}
(with the analogous equation for $\bar{T}(\bar{z})$),
where
\begin{displaymath}
t_{zz}
:=\left(\partial_z{\Gamma^z}_{zz}-\frac{1}{2}({\Gamma^z}_{zz})^2\right).{{1}}\:. 
\end{displaymath}
Here ${\Gamma^z}_{zz}=\partial_z\log G_{z\bar{z}}$ is the Christoffel symbol.
We have
\begin{displaymath}
\partial_{\bar{z}}T(z)=0\:. 
\end{displaymath}
\end{lemma}

\begin{proof}
Direct computation shows that
\begin{displaymath}
\partial_{\bar{z}}t_{zz}
=-\frac{1}{2}G_{z\bar{z}}\:\partial_z(\mathcal{R}.{{1}})\:. 
\end{displaymath}
From the conservation law Lemma \ref{conservation law for T} follows
\begin{align*}
\partial_{\bar{z}}T_{zz}
\=-G_{z\bar{z}}\:\partial_z{T^z}_z\\
\=-\frac{c}{48\pi}G_{z\bar{z}}\:\partial_z(\sqrt{G}\:\mathcal{R}.{{1}})
=\frac{c}{24\pi}\:\partial_{\bar{z}}t_{zz}\:.
\end{align*}
\end{proof}

Thus for constant sectional curvature, $T(z)$ is a holomorphic quadratic differential.

\begin{remark}
$t_{zz}$ defines a projective connection: 
Under a holomorphic coordinate change, $z\mapsto w$ such that $w\in\mathcal{D}(S)$, 
\begin{align*}
t_{ww}\;(dw)^2
\=t_{zz}\:(dz)^2-S(w)(z).{{1}}\:(dz)^2\:, 
\end{align*}
where $S(w)$ is the Schwarzian derivative,
\begin{displaymath}
S(w)
=\frac{w'''}{w'}-\frac{3}{2}\left[\frac{w''}{w'}\right]^2\:.
\end{displaymath}
$t_{zz}$ is known as the \textsl{Miura transform} of the affine connection 
given by the differentials ${\Gamma^z}_{zz}dz$.
\end{remark}

$T(z)$ is the holomorphic field introduced in \cite{L:2013},\cite{L:PhD14}.\footnote{
Our notations differ from those used in \cite{EO:1987}.
Thus the standard field $T(z)$ in \cite{EO:1987} equals $-T_{zz}$ in our exposition,
and the field $\tT(z)$ in \cite{EO:1987} equals $-\frac{1}{2\pi}T(z)$ here.}
For later reference, we note that from the transformation formula of $t_{zz}$ 
and invariance of $T_{zz}(dz)^2$, the following transformation rule follows for $T(z)$:
For a coordinate change $z\mapsto w$ with $w\in\mathcal{D}(S)$, we have 
\begin{align}\label{eq: coordinate transformation rule for T}
\hat{T}(w(z))\left[\frac{dw}{dz}\right]^2=T(z)-\frac{c}{12}S(w)(z).1\:.
\end{align}

For infinitesimal $\eps>0$, consider the map $F_n:\Sigma\rechts\Sigma$ given by
\begin{displaymath}
F_n: z
\mapsto 
\left(1+\eps f_n(z)\frac{\partial}{\partial z}\right)z
=z(1+\eps z^n)\quad\text{for}\quad f_n(z):=z^{n+1}\:. 
\end{displaymath}
In particular, $F_n(0)=0$.

\begin{definition}
Suppose $\Sigma$ has scalar curvature $\mathcal{R}=0$.
For $n\geq 0$, we define the map 
\begin{displaymath}
\delta_{F_n}:F\rechts F 
\end{displaymath}
as follows: For $\varphi(0)\in F=\sheaf{F}_{\Sigma,0}$,
\begin{displaymath}
\delta_{F_n}\varphi(0)
:=-\ointctrclockwise_{\gamma}f_n(z)T(z)\varphi(0)\:dz
-\ointctrclockwise_{\gamma}\overline{f}_n(z)\overline{T}(z)\varphi(0)\:dz\:.
\end{displaymath}
Here $\gamma$ is any closed path not containing (but possibly enclosing) the argument of $\varphi$.
\end{definition}

\begin{claim}
\textbf{If $\varphi$ is a holomorphic field, $\varphi\in F_{\text{hol}}$, 
then only the integral involving $T$ contributes}.
\end{claim}

\begin{proof}
The OPE of $\overline{T}(z_1)\otimes\varphi(z_2)$ has no singular part.
Indeed, Laurent expansion of $\overline{T}(z_1)$ yields
\begin{displaymath}
\overline{T}(z_1)\otimes\varphi(z_2)
=\sum_{n\geq n_0}(\bar{z}_1-\bar{z}_2)^nA_n(z_2)
\end{displaymath}
for the fields 
\begin{displaymath}
A_n\left(=\frac{1}{n!}\frac{\partial^nT}{\partial z_2^n}|_{z_2}\varphi(z_2)\right) 
\end{displaymath}
which depend only on $z_2$, and the dependence is holomorphic.
On the other hand, Laurent expansion of $\varphi(z_2)$ yields
\begin{displaymath}
\overline{T}(z_1)\otimes\varphi(z_2)
=\sum_{m\geq m_0}(z_1-z_2)^mB_m(z_1)\:,
\end{displaymath}
where $B_m$ depend holomorphically on $z_1$. 
The two expansions are incompatible unless the powers are non-negative.
\end{proof}

\begin{claim}
\textbf{$T_{\mu\nu}$ does not depend on the specific metric,
but only on the conformal class}. 
Thus $T(z)$ and $\overline{T}(z)$ defined for fixed $z=x^1+ix^2$ by
\begin{displaymath}
T_{\mu\nu}dx^{\mu}dx^{\nu}
=T(z)dz^2+\overline{T}(z)d\bar{z}^2
+\frac{c\mathcal{R}}{24\pi}dzd\bar{z}
\end{displaymath}
define elements of $\Gamma(\Sigma,\sheaf{F}_{\Sigma}\boxtimes (T\Sigma)^{\otimes 2})$.
\end{claim}

\begin{proof}
 
\end{proof}

\pagebreak

\subsection{OPE of the Virasoro field}

Application to the particular field $\varphi=T$ yields:

\begin{claim}
The operator product expansion (OPE) of the Virasoro field reads
\begin{align*}
T(z_1)\otimes T(z_2)
\mapsto
\frac{c/2}{(z_1-z_2)^4}.{{1}}+\frac{2T(z_2)}{(z_1-z_2)^2}+\reg\:. 
\end{align*} 
\end{claim}

\begin{proof}
$\delta_{F_n}$ acts as a diffeomorphism on $F$. 
Under the transformation $z\mapsto F_n(z)$, 
$T$ transforms according to eq.\ (\ref{eq: coordinate transformation rule for T}) as
\begin{align*}
\hat{T}(F_n(z))
\=(1+\eps f_n'(z))^{-2}\left(T(z)-\frac{c}{12}S(F_n)(z).{{1}}\right)\\
\approx&\:(1-2\eps f_n'(z))\left(T(z)-\frac{c}{12}\frac{f_n'''(z)}{f_n'(z)}.{{1}}\right)
\quad\text{for}\quad |z|<1\:,\:n\geq 0\:.
\end{align*}
On the other hand, $\hat{T}(F_n(z))=T(z)+\eps\delta_{F_n} T(z)+O(\eps^2)$ where for $\gamma$ enclosing $z=0$, 
\begin{displaymath}
\delta_{F_n}T(0) 
=-\ointctrclockwise_{\gamma(z=0)}f_n(z)T(z)T(0)\:dz\:.
\end{displaymath}
So
\begin{align*}
-\ointctrclockwise_{\gamma(z=0)}f_n(z)T(z)T(0)\:dz
\=-2f_n'(0)\left(T(0)-\frac{c}{12}\frac{f_n'''(0)}{f_n'(0)}.{{1}}\right)\\
\=-2(n+1)z^n|_{z=0}T(0)-\frac{c}{12}n(n^2-1)z^{n-2}|_{z=0}.{{1}}\\
\=-2\delta_{n,0}T(0)-\frac{c}{2}\delta_{n,2}.{{1}}
\:.
\end{align*}
Now on the l.h.s., $f_n=z^{n+1}$ sorts out the pole in the OPE of $T(z)\otimes T(0)$ of order $n+2$.
We let $n$ run through $n=0,1,2,\ldots$.

\begin{table}[ht]
\centering
\begin{tabular}{|l|l|l|}
\hline
$n$&order of pole&term in OPE\\ 
\hline
$0$&$2$&$2T(0)$\\ 
\hline
$2$&$4$&$\frac{c}{2}.{{1}}$\\ 
\hline
\end{tabular}\\
\end{table}
\end{proof}

\noi
The set $\{L_n\}_{n\in\Z}$ defined by
\begin{displaymath}
L_n
:=\oint\frac{T(z)}{z^{n-1}}\frac{dz}{2\pi i} 
\end{displaymath}
satisfies the Virasoro algebra
\begin{align}\label{eq: Virasoro algebra}
[L_m,L_n]
=(m-n)L_{m+n}+\frac{c}{12}m(m^2-1)\delta_{m+n,0}\:, 
\end{align}
a central extension of the Witt algebra (\ref{Witt type algebra, but for n geq 0 only}).

\pagebreak

\section{The variation formula}\label{Chaper: Heuristic supporting a new variation formula}

\subsection{The variation formula in the literature}\label{Section: The variation formula in the literature}

We cosndier a surface $S$ with metric $G_{\mu\nu}$.
The effect on $\1$ of a change $d G_{\mu\nu}$ in the metric is given by
\begin{align}\label{variation of the zero-point function}
d\1 
\=-\frac{1}{2}\iint d G_{\mu\nu}\:\langle T^{\mu\nu}\rangle\:\sqrt{G}\:dx^0\wedge dx^1\:.
\end{align}
Here $G:=|\det G_{\mu\nu}|$, 
and $dvol_2=\sqrt{G}\:dx^0\wedge dx^1$ is the volume form which is invariant under base change.
Eq.\ (\ref{variation of the zero-point function}) generalises to the variation of $N$-point functions $\langle\varphi_1(x_1)\ldots\varphi_N(x_N)\rangle$ as follows:
Suppose the metric is changed on an open subset $R\subseteq S$ of the surface $S$.
Then
\begin{align}\label{variation of N-point function}
d\langle\varphi_1(x_1)\ldots\varphi_N(x_N)\rangle 
\=-\frac{1}{2}\iint_S(d G_{\mu\nu})\:\langle T^{\mu\nu}\varphi_1(x_1)\ldots\varphi_N(x_N)\rangle\:dvol_2\:,
\end{align}
\cite[eq.\ (12.2.2) on p.\ 360]{Wein:GaC}, see also \cite[eq.\ (11)]{EO:1987}\footnote{
Note that both references introduce the Virasoro field with the opposite sign. 
Our sign convention follows e.g.\ \cite{DiFranc:1997}, cf.\ eq.\ (5.148) on p.\ 140.},
provided that
\begin{align}\label{metric change away from fields' positions}
x_i\not\in R\:,\hspace{1.5cm}\text{for}\:i=1,\ldots,N\:. 
\end{align}

Note that in order for the formula to be well-defined,
$T_{\mu\nu} dx^{\mu}dx^{\nu}$ must be quadratic differential on $S$,
i.e.\ one which transforms homogeneously under coordinate changes.
The antiholomorphic contribution in eq.\  (\ref{variation of N-point function}) is omitted.
It is of course of the same form as the holomorphic one, 
up to complex conjugation.

Due to invariance of $N$-point functions under diffeomorphisms, 
$T_{\mu\nu}$ satisfies the conservation law Lemma \ref{conservation law for T}. 

A Weyl transformation $G_{\mu\nu}\mapsto\mathcal{W}G_{\mu\nu}$ changes the metric only within the respective conformal class.
(In any chart $(U,x)$ on $S$, such transformation is given by $G_{\mu\nu}(x)\mapsto h(x)G_{\mu\nu}(x)$ with $h(x)\not=0$ on all of $U$.)
The effect of a Weyl transformation on $N$-point functions is described by the trace of $T$ (eq.\ (3) on p.\ 310 in \cite{EO:1987}),
which equals
\begin{align}\label{trace anomaly}
{T_{\mu}}^{\mu}
={T_z}^z+{T_{\bar{z}}}^{\bar{z}}
\=2{T_z}^z
=\frac{c}{24\pi}\:\mathcal{R}.{{1}}\:,
\end{align}
(\cite{DiFranc:1997}, eq.\ (5.144) on page 140, which is actually true for the underlying fields).
Here ${{1}}$ is the identity field, and $\mathcal{R}$ is the scalar curvature of the Levi-Civit\` a connection for $\nabla$ on $S$.
The non-vanishing of the trace (\ref{trace anomaly}) is referred to as the \textsl{trace} or \textsl{conformal anomaly}.

Since ${T_{\mu}}^{\mu}$ is a multiple of the unit field, the restriction (\ref{metric change away from fields' positions}) is unnecessary.
Thus under a Weyl transformation $G_{\mu\nu}\mapsto\mathcal{W}G_{\mu\nu}$,
all $N$-point functions change by the same factor $Z$ (equal to $\1$),
given by
\begin{align*}
d\log Z
\=-\frac{c}{24\pi}\iint\mathcal{R}\:d\mathcal{W}\:dvol_2\:. 
\end{align*}

While $T_{zz}$ transforms as a two-form, it is not holomorphic.
redefine the Virasoro field by 
Definition \ref{definition of T(z)}
to obtain a holomorphic field, 
but which as a result of the conformal anomaly, does not transform homogeneously in general. 

\subsection{The concise statement and proof of the variation formula}\label{Section: A new variation formula}

Let $S$ be a Riemann surface.
We introduce
\begin{align*}
\gamma:
&\quad\text{one-dimensional smooth submanifold of}\:S\:,
\:\text{topologically isomorphic to}\:S^1,\\
R:
&\quad\text{a tubular neighbourhood of}\:\gamma\:\text{in}\:S\:,\\
A:
&\quad\text{a vector field which conserves the metric on}\:S\:\text{and is holomorphic on}\:R\:.
\end{align*}
We think of $A\propto\frac{\partial}{\partial z}\in TR$ as an infinitesimal coordinate transformation
\begin{align}\label{coordinate transformation defined by alpha}
z
\quad
\mapsto
\quad 
w(z)
\=\left(1+\epsilon\frac{\partial}{\partial z}\right)z
=z+\alpha(z)\:, 
\end{align}
where $|\epsilon|\ll 1$. 

\begin{theorem}\label{Theorem: main formula for A with alpha=1}
Suppose $S$ has scalar curvature $\mathcal{R}=0$.
Let $\varphi$ be a holomorphic field on $S$. 
The effect of the transformation (\ref{coordinate transformation defined by alpha})
on $\langle\varphi(w)\rangle$ is
\begin{align*}
\left.\frac{d}{d\epsilon}\right\vert_{\epsilon=0}\langle\varphi(w)\rangle
=-i\ointctrclockwise_{\gamma}\:\langle T_{zz}\:\varphi(w)\rangle\:dz
\:,
\end{align*} 
provided that  
\begin{align}\label{condition on the position of the holomorphic field}
w\:\text{does not lie on the curve}\:\gamma\:.
\end{align}
\end{theorem}

In particular, as $w$ is not enclosed by $\gamma$, $\langle\varphi(w)\rangle$ doesn't change.

\begin{proof}
By property (\ref{condition on the position of the holomorphic field}),  
the position of $\varphi$ is not contained in a small tubular neighbourhood $R$ of $\gamma$.
Let 
\begin{displaymath}
R\setminus\gamma=R_{\text{left}}\sqcup R_{\text{right}}\: 
\end{displaymath}
be the decomposition in connected parts left and right of $\gamma$
(we assume $\gamma$ has positive orientation). 
Let $W\subset S$ be an open set s.t.\ 
\begin{displaymath}
\overline{W}\cap\gamma=\emptyset\:,
\quad
W\cup R=S\:. 
\end{displaymath}
We let $F:R\rechts [0,1]$ be a smooth function s.t.
\begin{align*}
F
\=1\quad\text{on}\quad R_{\text{left}}\cap W\:,\\ 
F     
\=0\quad\text{on}\quad R_{\text{right}}\cap W\:.
\end{align*}
Let $\epsilon$ be so small that $z\in W^c=S\setminus W$ implies $\exp(\epsilon F)(z)\in R$.
Define a new metric manifold $(S^{\epsilon},G_{z\bar{z}}^{\epsilon})$ by
\begin{align*}
S^{\epsilon}|_W
:\=S|_W\\
G_{z\bar{z}}^{\epsilon}(z)\:|dz|^2
:\=G_{z\bar{z}}(\exp(\epsilon F)(z))\:|d\exp(\epsilon F)(z)|^2\:, 
\quad z\in W^c\:.
\end{align*}
We have
\begin{displaymath}
d G_{\mu\nu}T^{\mu\nu}
=dG_{\bar{z}\bar{z}}T^{\bar{z}\bar{z}}+\text{antiholomorphic contributions}+\text{Weyl terms}\:, 
\end{displaymath}
where we disregard the antiholomorphic contributions $\sim T_{\bar{z}\bar{z}}$, 
and the Weyl terms are absent since by assumption $\mathcal{R}=0$. 
Alternatively, we can describe the change in the metric by the map
\begin{displaymath}
|dz|^2\mapsto |dz+\mu d\bar{z}|^2=dzd\bar{z}+\mu d\bar{z}d\bar{z}+\ldots\:,
\end{displaymath}
where
\begin{displaymath}
\mu=\epsilon\partial_{\bar{z}}F+O(\epsilon^2)\: 
\end{displaymath}
is the Beltrami differential. Thus
\begin{displaymath}
dG_{\bar{z}\bar{z}}
=2G_{z\bar{z}}\:d\mu(z,\bar{z})\:. 
\end{displaymath}
Eq.\ (\ref{variation of N-point function}) yields
\begin{align*}
\frac{d\langle\varphi\rangle}{d\epsilon}|_{\epsilon=0} 
\=-\frac{1}{2}\iint_S\frac{\partial G_{\mu\nu}}{\partial\epsilon}|_{\epsilon=0} \:\langle T^{\mu\nu}\:\varphi\rangle\:dvol_2\\
\=-\frac{i}{2}\iint_S 2G_{z\bar{z}}\:\frac{\partial\mu(z,\bar{z})}{\partial\epsilon}|_{\epsilon=0}
\:(G^{z\bar{z}})^2\langle T_{zz}\:\varphi\rangle\:G_{z\bar{z}}\:dz\wedge d\bar{z}\\
\=\i\iint_{R}(\partial_{\bar{z}}F)\:\langle T_{zz}\:\varphi\rangle\:d\bar{z}\wedge dz\:,
\end{align*}
since $(G^{z\bar{z}})^{k}=(G_{z\bar{z}})^{-k}$ for $k\in\Z$.
Here 
\begin{displaymath}
\langle T_{zz}\:\varphi\rangle\:dz
=\iota_A(\langle T_{zz}\:\varphi\rangle\:(dz)^2)
\end{displaymath}
is the holomorphic $1$-form 
given by the contraction of the holomorpic vector field $A=\frac{\partial}{\partial z}$
with the quadratic differential $\langle T_{zz}\:\varphi\rangle\:(dz)^2$,
which is holomorphic on $R$.
By Stokes' Theorem,
\begin{align*}
\frac{d\langle\varphi\rangle}{d\epsilon}|_{\epsilon=0} 
\=\i\iint_{R}\partial_{\bar{z}}\left(F\:\langle T_{zz}\:\varphi\rangle\right)\:d\bar{z}\wedge dz\\
\=\i\ointctrclockwise_{W_R}F\:\langle T_{zz}\:\varphi\rangle\:dz
+\i\ointclockwise_{W_L}F\:\langle T_{zz}\:\varphi\rangle\:dz\\
\=-\i\ointctrclockwise_{W_L}F\:\langle T_{zz}\:\varphi\rangle\:dz\:.
\end{align*}
Here $W_R=N_R\cap\partial W$ and $W_L=N_L\cap\partial W$ are the left and right boundary, respectively, of $W$ in $R$.
We conclude that
\begin{displaymath}
\frac{d\langle\varphi\rangle}{d\epsilon}|_{\epsilon=0} 
=-\i\ointctrclockwise_{W_L}\langle T_{zz}\:\varphi\rangle\:dz
=-\i\ointctrclockwise_{\gamma}\langle T_{zz}\:\varphi\rangle\:dz\:,
\end{displaymath}
by holomorphicity on $R_{\text{left}}\cup\gamma$.
\end{proof}
  
\begin{remark}
The construction is independent of $F$.
When $F$ approaches the discontinuous function defined by
\begin{align*}
\begin{cases}
F
=1\quad&\text{on}\quad R_{\text{left}}\:,\\ 
F
=0\quad&\text{on}\quad R_{\text{right}}\:,
\end{cases}
\end{align*}
we obtain a description of $(S^{\epsilon},G_{z\bar{z}}^{\epsilon})$ by cutting along $\gamma$ and pasting back
after a transformation by $\exp(\epsilon)$ on the left.  
\end{remark}

\begin{remark}
The integral formula is similar to the conformal Ward identity in the literature \cite{DiFranc:1997} (in particular the so-called conformal Ward identity (5.46)).
The exposition is not very clear, however, and may refer to global transformations, while we consider local coordinate transformations.
Also, the contour of the integral is required to strictly enclose the position of any field contained in the $N$-point function,
while we just require them not to lie on the contour of integration. 
\end{remark}

There is a way to check the result of Theorem \ref{Theorem: main formula for A with alpha=1}:
Let $\varphi$ be a holomorphic field 
whose position lies in a sufficiently small open set $U\subset S$ with boundary $\partial U=\gamma$. 
We can use a translationally invariant metric in $U$ and corresponding coordinates $z,\bar{z}$. Then 
\begin{displaymath}
T_{zz}=\frac{1}{2\pi}T(z)\: 
\end{displaymath}
in eq.\ (\ref{definition of T(z)}). For $A=\frac{d}{dw}$, we have
\begin{align}\label{primitive example for the main formula, when alpha=1}
\langle A\varphi(w)\ldots\rangle
\=\frac{1}{2\pi\i}
\ointctrclockwise_{\gamma}\langle T(z)\varphi(w)\ldots\rangle\:dz\:,
\end{align}                     
This can be seen in two ways.
\begin{enumerate}
\item 
Eq.\ (\ref{primitive example for the main formula, when alpha=1}) follows from the residue theorem 
for the OPE of $T(z)\otimes\varphi(w)$.
Indeed, the Laurent coefficient of the first order pole at $z=w$ is $N_{-1}(T,\varphi)(w)=\partial_w\varphi$,
which is holomorphic.
\item
Alternatively, by Theorem \ref{Theorem: main formula for A with alpha=1},
\begin{align*}
\left.\frac{d}{d\epsilon}\right\vert_{\epsilon}\langle\varphi(w+\epsilon)\ldots\rangle
\=\frac{1}{2\pi\i}\ointctrclockwise_{\gamma}\langle T(z)\varphi(w)\ldots\rangle\:dz\:.
\end{align*}
\end{enumerate}
The two approaches are compatible!

\subsection{Discussion of the metric}\label{Section: Introduction of singular metric}

Let $\Sigma_g$ be the genus $g$ hyperelliptic Riemann surface
\begin{displaymath}
\Sigma_g:\quad y^2=p(x)\:,
\quad
\deg p=n=2g+1\:. 
\end{displaymath}
Recall that $x$ which varies over the Riemann sphere, defines a complex coordinate on $\Sigma_g$, outside the ramification points 
where we must change to the $y$ coordinate. 
$\CP^1$ does not allow for a constant curvature metric
but we shall define a metric on $\CP^1$ which is flat almost everywhere.

Suppose we consider a genus one surface with $n=3$. By means of the isomorphism $\CP^1\cong\C\cup\{\infty\}$,
we may identify the branch points of $\Sigma_1$ with points $X_1,X_2,X_3\in\C$ and $X_4=\{\infty\}$, respectively.

Let $\theta\gg 1$, but finite, such that in the flat metric of $\C$,
\begin{displaymath}
|X_i|<\theta\:,\quad i=1,2,3\:.
\end{displaymath}
We define $|X_4|:=\infty$. For $\epsilon>0$,  
define a metric 
\begin{align}\label{eps-metric}
(ds(\epsilon))^2=
2G_{z\bar{z}}(\epsilon)\:dz\otimes d\bar{z}
\end{align}
on $\CP^1$ by
\begin{align*}
2G_{z\bar{z}}(\epsilon):=
\begin{cases}
(1+\epsilon \theta^2)^{-2} &\quad\text{for}\quad |z|\leq\theta\:,\\
(1+\epsilon z\bar{z})^{-2} &\quad\text{for}\quad |z|\geq\theta\:.
\end{cases}
\end{align*}
The metric on $\Sigma_1$ is obtained by lifting.

\begin{lemma}
In the disc $|z|\leq\theta$, the metric is flat, 
while in the area $|z|\geq\theta$, it is of Fubini-Study type of Gauss curvature $\mathcal{K}=4\epsilon$.
\end{lemma}

\begin{proof}
For $\rho=2G_{z'\bar{z}'}(\epsilon)$ with
\begin{displaymath}
G_{z'\bar{z}'}(\epsilon) 
:=\frac{1}{2\epsilon}(1+z'\bar{z}')^{-2}\quad\text{for}\quad |z'|\geq\sqrt{\epsilon}\theta\:,
\end{displaymath}
we have \cite{F:1982}
\begin{align*}
\mathcal{R}
=\rho^{-1}(-4\partial_z\partial_{\bar{z}}\log\rho)
=\epsilon(1+z'\bar{z}')^2(8\partial_{z'}\partial_{\bar{z}'}\log(1+z'\bar{z}')^2)
=8\epsilon\:,
\end{align*}
and $\mathcal{R}=2\mathcal{K}$. 
\end{proof}

\begin{definition}
Let $\Sigma$ be a genus $g=1$ Riemann surface 
with conformal structure defined by the position of the ramification points $\{X_i\}_{i=1}^3$
with finite relative distance on $\CP^1$.
Let $G_{z\bar{z}}(\epsilon)$ be the metric defined by eq.\ (\ref{eps-metric}). 
We define $\1_{\{X_i\}_{i=1}^3,\epsilon,\theta}$ to be the zero-point function on $(\Sigma,G_{z\bar{z}}(\epsilon))$. 
\end{definition}

By eq.\ (\ref{trace anomaly}) and the fact that on any surface, $\mathcal{R}=2\mathcal{K}$,
\begin{displaymath}
T_{z\bar{z}}
=\frac{c}{24\pi}G_{z\bar{z}}\mathcal{K}.{{1}}\:, 
\end{displaymath}
where ${{1}}$ is the identity field.
So according to eq.\ (\ref{variation of the zero-point function}) we have for the $2$-sphere $S^2_{\theta}$ of radius $\theta$, 
\begin{align*}
d\log\1_{\{X_i\}_{i=1}^3,\epsilon,\theta} 
\=\frac{c}{48\pi}\iint_{S^2_{\theta}}(d\log G_{z\bar{z}}(\epsilon))\:\mathcal{K}\:dvol_2\:.
\end{align*}
Since $G(\epsilon)=(G_{z\bar{z}}(\epsilon))^2$, 
for $|z|>\theta$, the two-dimensional volume form is
\begin{displaymath}
dvol_2
=G_{z\bar{z}}(\epsilon)\:dz\wedge d\bar{z}
=\frac{1}{2}\frac{\pi d(r^2)}{(1+\epsilon r^2)^2}\:. 
\end{displaymath}
Now
\begin{displaymath}
d\log\1_{\{X_i\}_{i=1}^3,\epsilon,\theta} 
=dI_{|z|<\theta}+dI_{|z|>\theta}\:, 
\end{displaymath}
where for $\varrho_0^2:=\epsilon\theta^2$, the integrals yield 
\begin{align*}
dI_{|z|<\theta}
\=-\frac{c\theta^2}{12}\:d(\epsilon)\:\frac{\varrho_0^2}{(1+\varrho_0^2)^3}\:,\\
dI_{|z|>\theta}
\=-\frac{c}{12}\:(d\log\epsilon)\:\int_{|\varrho|^2>\varrho_0^2}
\frac{\varrho^2\:d(\varrho^2)}{(1+\varrho^2)^3}
=-\frac{c}{24}\:(d\log\epsilon)\:(1+O(\varrho_0^4))
\:.
\end{align*}
So for $|\varrho_0|\ll 1$,
\begin{align}\label{limit of zero-point function as eps in the metric goes to zero}
\1_{\{X_i\}_{i=1}^3,\epsilon,\theta}
\=\epsilon^{-\frac{c}{24}(1+O(\varrho_0^4))}\:Z\:\exp\left(-\frac{c}{12}\:\frac{\varrho_0^4}{(1+\varrho_0^2)^3}\right)\:,
\end{align}
where $Z\in\C$ is an integration constant.

Variation of $\epsilon$ rescales the metric within the conformal class defined by the branch points.
In the limit as $\epsilon\searrow 0$,
\begin{align}\label{definition of the singular metric on CP1}
G_{z\bar{z}}
:=\lim_{\epsilon\searrow 0}\:G_{z\bar{z}}(\epsilon)
=\frac{1}{2}
\quad\text{for}\quad |z|<\infty\:,
\end{align}
(and is undefined for $|z|=\infty$). 
Thus $\CP^1$ becomes an everywhere flat surface except for the point at infinity, which is a singularity for the metric.
 
\begin{definition}
Let $\Sigma_1$ be a genus $g=1$ Riemann surface 
with conformal structure defined by the position of the ramification points $\{X_i\}_{i=1}^3$
with finite relative distance on $\CP^1$.
Let $G_{z\bar{z}}$ be the metric on $\Sigma$ defined by eq.\ (\ref{definition of the singular metric on CP1}).
We define the zero-point function on $(\Sigma_1,G_{z\bar{z}})$ by 
\begin{align*}
\1_{\{X_i\}_{i=1}^3} 
:=\lim_{\rho_0\searrow 0}\:\epsilon^{\frac{c}{24}(1+O(\varrho_0^4))}\1_{\{X_i\}_{i=1}^3,\epsilon,\theta}\:. 
\end{align*}
\end{definition}

Thus $\1_{\{X_i\}_{i=1}^3}=Z$. We shall also write $\onesing$
to emphasise distinction from the $0$-point function on the flat torus $(\Sigma_1,|dz|^2)$, which we denote by $\oneflat$.

\begin{remark}\label{remark: Infinite proportionality factor between the zero function w.r.t. the flat metric and the one w.r.t. the singular metric}
The reason for introducing $\epsilon$ and performing $\lim_{\epsilon\searrow 0}$ is the fact that the logarithm of the Weyl factor $\mathcal{W}$
is not defined for surfaces with a singular metric and infinite volume. 
We have
\begin{displaymath}
d\log\frac{\onesing}{\oneflat}
=d\log\mathcal{W}\:,
\end{displaymath}
so $\mathcal{W}$ is determined only up to a multiplicative constant,  which is infinite for $\epsilon=0$.
\end{remark}

Our method is available for any surface $\Sigma_g: y^2=p(x)$ with $\deg p=n\geq 3$.
When $n$ is odd, the point at infinity is a non-distinguished element in the set of ramification points on $\Sigma_g$.
We shall distribute the curvature of $\Sigma_g$ evenly over these.
Using the Gauss-Bonnet theorem, the total curvature is recovered as 
\begin{align*}
\int_{\Sigma_g}\mathcal{K}\:dvol_2
=2\pi\:\chi(\Sigma_g)
=4\pi(1-g)
=8\pi-2\pi(2g+2)\:. 
\end{align*}
We interpret $8\pi$ as the contribution to the curvature from the $g=0$ double covering 
and $-2\pi$ from any branch point.


The method is now available for arbitrary genus $g\geq 1$ hyperelliptic Riemann surfaces
and will in the following be checked against the case $g=1$.
  
\subsection{The main theorem}\label{Section: Main theorem}

We now get to an algebraic description of the effect on an $N$-point function 
as the position of the ramification points of the surface is changed.

\begin{theorem}\label{Thm: Main theorem}
Let $\Sigma_g$ be the hyperelliptic Riemann surface
\begin{displaymath}
\Sigma_g:\quad y^2=p(x)\:, \quad n=\deg p=2g+1\:,
\end{displaymath}
with roots $X_j$. 
We equip the $\CP^1$ underlying $\Sigma_g$ with the singular metric 
which is equal to 
\begin{displaymath}
|dz|^2
\quad
\text{on}\:\CP^1\setminus\{X_1\,\dots,X_n\}\:.
\end{displaymath}
Let $\langle\:\ranglesing$ be a state on $\Sigma_g$ with the singular metric.
We define a deformation of the conformal structure by 
\begin{displaymath}
\xi_j=dX_j\quad\text{for}\quad j=1,\ldots,n\:.
\end{displaymath}
Let $(U_j,z)$ be a chart on $\Sigma_g$ containing $X_j$ but no field position.
We have
\begin{align}\label{main formula, but for alpha=1 and N=1}
d\langle\varphi\ldots\ranglesing
\=\sum_{j=1}^{n}\left(\frac{1}{2\pi\i}\ointctrclockwise_{\gamma_j}\langle T(z)\varphi\ldots\ranglesing\:dz\right)\:\xi_j\:,
\end{align} 
where $\gamma_j$ is a closed path around $X_j$ in $U_j$. 
\end{theorem}

\begin{proof}
On the chart $(U,z)$, we have $\frac{1}{2\pi}\:T(z)=T_{zz}$ in eq.\ (\ref{definition of T(z)}),
outside the points which project onto one of the $X_j$ for $j=1,\ldots,n$ on $\CP^1$.
Moreover, $\gamma$ does not pick up any curvature for whatever path $\gamma$ we choose.
Since
\begin{displaymath}
d\onesing
=\sum_{i=1}^{n}\xi_i\frac{\partial}{\partial X_i}\onesing\:,
\end{displaymath}
formula (\ref{main formula, but for alpha=1 and N=1}) follows from Theorem \ref{Theorem: main formula for A with alpha=1}.
\end{proof}

\section{Differential equation for $N$-point functions of the Virasoro field, for arbitrary genus}

\subsection{Notations}

In the remainder of the paper, we will deal with very specific fields which will be distinguishable by the letter - $1,T,\vartheta,\psi$ - 
rather than by a lower index. 
\begin{itemize}
\item 
To enhance readibility of the formulae, we shall denote $p(x),\vartheta(x),\ldots$ and $f(x,X_s)$ 
by $p_x,\vartheta_x,\ldots$ and $f_{xX_s}$.  
Instead of $f(x_1,x_2)$ and $p_{x_i},\vartheta_{x_j},\ldots$ we shall write $f_{12}$ and $p_i,\vartheta_j,\ldots$, respectively.
Thus
\begin{displaymath}
f_{12}
=\left(\frac{y_1+y_2}{x_1-x_2}\right)^2\:.
\end{displaymath}
We shall avoid notations like $f_{x,y}$ and write instead $f_x$ since $y=p(x)$.
Subscripts will never denote derivatives.
We we also use lower indices for the coefficents of Laurent series expansions,
however,
like 
\begin{displaymath}
a_k,\:\Theta_k,\:\Psi_k\:.        
\end{displaymath}
These coefficients will not depend on position other than the reference point of the expansion, so the notatiion should be unambiguous.
\item
For a function $f$ of $x$, we denote $f'=\frac{\partial}{\partial x}f$, and for $k\geq 3$, $f^{(k)}=\frac{\partial^k}{\partial x^k}f$.
(However, in the notation $f^{(k)}$ we may include $k=0,1$.)
We also write
\begin{displaymath}
f'_{X_s}:=\frac{d}{dx}|_{x=X_s}f_x\:.
\end{displaymath}
and for $\xi_s=dX_s$,
\begin{displaymath}
d_{X_s}
=\xi_s\frac{\partial}{\partial X_s}\:.
\end{displaymath}
\item
We let
\begin{align*}
\alignedbox{\omega_s}
{:=\sum_{t\not=s}\frac{\xi_s}{X_s-X_t}}
\end{align*}
and
\begin{displaymath}
\omega
:=\sum_{s=1}^n\omega_s
=\sum_{s=1}^n\sum_{t>s}\frac{\xi_s-\xi_t}{X_s-X_t}
=\frac{1}{2}\underset{t\not=s}{\sum_{t,s=1}^n}\frac{\xi_s-\xi_t}{X_s-X_t}\:.
\end{displaymath}
\item
By a pole at $x=0$ we mean a $\frac{1}{x^m}$ singularity with $-m\in\N\setminus\{0\}$.
\item
For any rational function $R$ of $x,y$ with $y^2=p(x)$, 
let $[R(x,y)]_{\text{no pole}}$ denote the projection of $R(x,y)$ onto those terms of $R(x,y)$ that have no pole at $x=X_s$
(but may have a squarte root singularity), 
where $X_s$ is the image of a ramification point $(X_s,0)$ ,on $\CP^1$ (a simple zero of $p=y^2$) specified in the context.
Thus
\begin{align*}
\left[\vartheta(x)\vartheta_{X_s}\right]_{\underset{x=X_s}{\text{no pole}}}
\=\lim_{x\rechts X_s}\left[\vartheta_x\vartheta_{X_s}\right]_{\text{no pole at $x=X_s$}}\:.
\end{align*}
\item
The Schwarzian derivative of $f$ w.r.t.\ $x$ at $x_0$ is (assuming it is defined)
\begin{align*}
S(f_x)(x_0)
:=\frac{f^{(3)}_{x_0}}{f'_{x_0}}-\frac{3}{2}\left[\frac{f''_{x_0}}{f'_{x_0}}\right]^2\:,
\end{align*}
where $f'=\frac{d}{dx}f$, etc.
\item
When using contour integrals, when $P$ is a point on a surface $S$, 
we shall denote by ${\gamma_P}$ a closed path in $S$ that encloses the point $P$ but does not pass through it. 
\end{itemize}

\subsection{Introduction of the auxiliary fields $\vartheta$ and $\psi$}

We recall resp.\ generalise, a few definitions from \cite{L:2013} and \cite{L:PhD14}.
Let $\vartheta$ be the field defined by
\begin{align}\label{definition of vartheta}
T_x\:p_x
=\vartheta_x+\frac{c}{32}\frac{[p'_x]^2}{p_x}.1\:.
\end{align}

\begin{lemma}
Let $g\geq 1$. In the $(2,5)$ minimal model, the OPE for the field $\vartheta$ reads
\begin{align}\label{map: OPE of vartheta}
\vartheta_1\otimes\vartheta_2
\:\mapsto&\:\frac{c}{32}f_{12}^2
+\frac{1}{4}f_{12}(\vartheta_1+\vartheta_2)
+\psi_x
+O((x_1-x_2))\:,
\end{align}
where
\begin{align}\label{eq: definition of psi}
\psi_x
:=-\frac{c}{480}[p'_x]^2S(p_x).1
+\frac{1}{5}(p''_x\vartheta_x-\frac{1}{2}p'_x\vartheta'_x-p_x\vartheta''_x).1\:.
\end{align}
\end{lemma}

\begin{proof}
From eqs (\ref{definition of vartheta}) and (\ref{transformation of Virasoro field from x to y coordinate}),
\begin{align}
\label{eq: writing vartheta in terms of T(x) or in terms of T(y)}
\vartheta_x
=\frac{[p_x']^2}{4}\hat{T}_y+\frac{c}{12}p_xS(p_x)\:.
\end{align}
For brevity, we introduce the notation $S=S(p_x)(x)$ and for $i=1,2$, $S_i=S(p_x)(x_i)$.
From the OPE for $\hat{T}_y$, using that in the $(2,5)$ minimal model, $\Phi_y=-\frac{1}{5}\partial_y^2\hat{T}_y$,
we have
\begin{align}
\vartheta_1\otimes\vartheta_2
\mapsto&\:\frac{[p_1'p_2']^2}{16}
\left(\frac{c}{2}\frac{1}{(y_1-y_2)^4}.1
+\frac{\hat{T}_1+\hat{T}_2}{(y_1-y_2)^2}
-\frac{1}{5}\partial_y^2\hat{T}_y\right)\label{line with OPE of T in y coordinates}\\
\+\frac{c}{6}p_xS\vartheta_x
-\left(\frac{c}{12}p_xS\right)^2.1
+O(y_1-y_2)\:,\nn
\end{align}
where the expression on the r.h.s.\ of the arrow in line (\ref{line with OPE of T in y coordinates}) reads
\begin{align*}
&\:\frac{c}{32}\frac{[p_1'p_2']^2}{(p_1-p_2)^4}(y_1+y_2)^4.1\\
\+\frac{1}{4}\frac{p'_1p'_2}{(p_1-p_2)^2}(y_1+y_2)^2
\left\{
(\vartheta_1-\frac{c}{12}p_1S_1)
+(\vartheta_2-\frac{c}{12}p_2S_2)
\right\}\\
\-\frac{[p'_x]^4}{10}\left[\frac{1}{p'_x}\partial_x+\frac{2p_x}{p'_x}\partial_x\frac{1}{p'_x}\partial_x\right]
\left(\frac{\vartheta_x-\frac{c}{12}p_xS.1}{[p_x']^2}\right)\:,
\end{align*}
We use
\begin{align*}
&\frac{(p_1-p_2)^2}{p_1'p_2'}\\
&\quad=(x_1-x_2)^{2}
\left(1
-\frac{(x_1-x_2)^2}{12}(S_1+S_2)
+\frac{(x_1-x_2)^4}{30}\left(\frac{S''(p_1)+S''(p_2)}{4}+\frac{S_1S_2}{3}\right)
+O((x_1-x_2)^6)\right)\,
\end{align*}
(indeed, the l.h.s.\ is invariant under linear fractional transformations),
and
\begin{displaymath}
(y_1+y_2)^4
=2(y_1+y_2)^2(p_1+p_2)-(p_1-p_2)^2\:.
\end{displaymath}
The expression on the r.h.s.\ of the arrow in line (\ref{line with OPE of T in y coordinates}) becomes
\begin{align}
&\:\frac{c}{32}f_{12}^2.1+\frac{1}{4}f_{12}(\vartheta_1+\vartheta_2)
-\frac{c}{32}\left(\frac{p_1-p_2}{x_1-x_2}\right)^2\frac{S}{3}\label{line 1 of OPE of vartheta}\\
\+\frac{c}{96}\left(\frac{y_1+y_2}{x_1-x_2}\right)^{2}(p_1-p_2)(S_1-S_2).1\nn\\
\+\frac{c}{96}(y_1+y_2)^4\left(\left(\frac{S}{3}\right)^2-\frac{1}{5}\left(\frac{S''}{2}+\frac{S^2}{3}\right)\right).1\nn\\
\+\frac{1}{4}(y_1+y_2)^2\frac{S}{3}\left(\vartheta_x-\frac{c}{12}p_xS.1\right)\nn\\
\+\frac{1}{4}\left(\frac{y_1+y_2}{x_1-x_2}\right)^2(p'_2-p'_1)
\left(\frac{\vartheta_1-\frac{c}{12}p_1S.1}{[p_1']^2}-\frac{\vartheta_2-\frac{c}{12}p_2S.1}{[p_2']^2}\right)\nn\\
\-\frac{[p'_x]^3}{10}\partial_x\frac{\vartheta_x-\frac{c}{12}p_xS.1}{[p_x']^2}
-\frac{p_x[p'_x]^3}{5}\partial_x\frac{1}{p'_x}\partial_x\frac{\vartheta_x-\frac{c}{12}p_xS.1}{[p_x']^2}\label{line 6 of OPE of vartheta}\\
\+O((x_1-x_2)^2)\:.\nn
\end{align}
Any term in the linear span of 
\begin{align*}
&p_x^2S'',\:p_xp'_xS',\:p_xp''_xS,\:\frac{p_x^2p''_x}{p'_x}S',\:\frac{p_x^2[p''_x]^2}{[p'_x]^2}S',\\
&\frac{p_xp^{(3)}_x}{p'_x}\vartheta_x,\:\frac{p_x[p''_x]^2}{[p'_x]^2}\vartheta_x,\:\frac{p_xp''_x}{p'_x}\vartheta'_x
\end{align*}
must drop out from the OPE for $\vartheta$ by eq.\ (\ref{definition of vartheta}),
as $T(x)$ is regular at $p'_x=0$.
Note that $[p'_x]^2S$ is allowed.
A combinatorial argument dealing with the number of factors of $p_x$ and its derivatives,
and the overall number of derivatives, (counted with sign), shows that every term must have a factor of $p_x$.
The only term excluded from the list is $p_x\vartheta''_x$, which is allowed.
We find that only lines (\ref{line 1 of OPE of vartheta}) and (\ref{line 6 of OPE of vartheta}) contribute to the OPE,
where
\begin{align*}
-\frac{[p'_x]^3}{10}\left[\partial_x+2p_x\partial_x\frac{1}{p'_x}\partial_x\right]
\frac{\vartheta_x}{[p_x']^2} 
\=-\frac{1}{10}p'_x\vartheta'_x
+\frac{1}{5}p''_x\vartheta_x
-\frac{1}{5}p_x\vartheta''_x+O(1/p'_x)\:,\\
\frac{c}{12}\frac{[p'_x]^3}{10}\left[\partial_x+2p_x\partial_x\frac{1}{p'_x}\partial_x\right]
\frac{p_xS}{[p_x']^2}
\=\frac{c}{120}[p'_x]^2S+\frac{c}{24}p_xp^{(4)}+O(1/p'_x)\:.
\end{align*}
Collecting terms yields the claimed OPE.
\end{proof}

\begin{claim}
For $k\in\Z$ and $1\leq s\leq n$ fixed, 
we define an operator $\Theta_k$ on holomorphic fields by
\begin{align}\label{def of Laurent coefficients of vartheta} 
\Theta_k
\=\ointctrclockwise_{\gamma_{X_s}}\frac{\vartheta_x}{(x-X_s)^{k+1}}\:\frac{dx}{2\pi\i}\:,
\end{align}
The operators $\Theta_k$ generate a non-commutative algebra which is equivalent to the Virasoro algebra
(i.e., their respective commutation relations can be deduced from one another).
\end{claim}

\begin{remark}
We have
\begin{align}\label{eq: Theta(k) operator gives zero for negative k} 
\Theta_k=0\quad\text{for $k<0$}\:,
\end{align}
i.e.\ all $N$-point functions of $\Theta_k$ and $N-1$ holomorphic fields vanishes.
So in the following, we shall always assume $k\in\N_0$.
\end{remark}
 
\begin{proof}
We define a local coordinate $\check{y}_x$ for $x$ near some ramification point in close distance to $X_s$, by
\begin{displaymath}
\check{y}_x^2
:=(x-X_s)
\:.
\end{displaymath}
By the transformation formula eq.\ (\ref{eq: coordinate transformation rule for T}),
\begin{displaymath}
\frac{1}{4\check{y}^2}\check{T}_{\check{y}}
=T_x-\frac{c}{32}\frac{1}{\check{y}^4}.1
\:.
\end{displaymath}
It is convenient to introduce $p_x=:(x-X_s)\hat{p}_x$.
Thus by eq.\ (\ref{definition of vartheta}),
\begin{displaymath}
\vartheta_x
=\frac{1}{4}\check{T}_{\check{y}}\hat{p}_x
-\frac{c}{32}\left(2\hat{p}'_x+\check{y}^2\frac{[\hat{p}'_x]^2}{\hat{p}_x}\right).1\:,
\end{displaymath}
(note that $\langle\vartheta_x\rangle$ is regular $x=X_s$ since $\hat{p}_{X_s}\not=0$).
This shows that
\begin{displaymath}
\ointctrclockwise_{X_s}\frac{\vartheta_x}{(x-X_s)^{k+1}}\frac{dx}{2\pi\i}
=\frac{1}{4}\ointctrclockwise_{X_s}\frac{\check{T}_{\check{y}}\hat{p}_x}{(x-X_s)^{k+1}}\:\frac{dx}{2\pi\i}
+\{\text{terms}\propto.1\}
\:.
\end{displaymath}
Set
\begin{displaymath}
\hat{p}_x
=\sum_{\ell=0}^{n-1}\hat{a}_{\ell}(x-X_s)^{\ell}
=\sum_{\ell=0}^{n-1}\hat{a}_{\ell}\check{y}^{2\ell}\:, 
\end{displaymath}
where $\hat{a}_{\ell}$ are constant in $\check{y}$.
Then
\begin{displaymath}
\frac{1}{4}\ointctrclockwise_{\underset{\text{twice}}{X_s}}\frac{\check{T}_{\check{y}}\hat{p}_x}{\check{y}^{2k+2}}\:\frac{dx}{2\pi\i}
=\sum_{\ell=0}^{n-1}\hat{a}_{\ell}\ointctrclockwise_{0}\frac{\check{T}_{\check{y}}}{\check{y}^{2(k-\ell)+1}}\:\frac{d\check{y}}{2\pi\i}
=\sum_{\ell=0}^{n-1}\hat{a}_{\ell}\check{L}_{2(k-\ell+1)}\:,
\end{displaymath}
where the $\check{L}_m$ satisfy the Virasoro algebra (\ref{eq: Virasoro algebra}) (with $\check{L}_m$ in place of $L_m$).
So the $\Theta_k$ satify the commutation relation
\begin{align}\label{commutation relation of the Theta's}
\left[\Theta_{k_1},\Theta_{k_2}\right]
=\frac{1}{4}\sum_{\ell_1,\ell_2=0}^{n-1}\hat{a}_{\ell_1}\hat{a}_{\ell_2}\left[\check{L}_{2(k_1-\ell_1+1)},\check{L}_{2(k_2-\ell_2+1)}\right]
\:.
\end{align}
(Note the factor of $1/2$ which accounts for circling $X_s$ twice.)
Inversely, for $|\sum_{\ell=1}^{n-1}\frac{\hat{a}_{\ell}}{a_0}\check{y}^{2\ell}|<1,$ that is, $|x-X_s|\ll1$,
\begin{align*}
\frac{1}{\hat{p}_x}
=\frac{1}{\hat{a}_0}\sum_{m=0}^{\infty}(-1)^m\left(\sum_{\ell=1}^{n-1}\frac{\hat{a}_{\ell}}{\hat{a}_0}\check{y}^{2\ell}\right)^m
=\frac{1}{\hat{a}_0}\sum_{m=0}^{\infty}(-1)^m
\sum_{m_1+m_2+\cdots +m_{n-1} = m}\frac{m!}{m_1!\ldots m_{n-1}!}\prod_{\ell=1}^{n-1}\left(\frac{\hat{a}_{\ell}}{\hat{a}_0}\check{y}^{2\ell}\right)^{m_{\ell}}
\end{align*}
so
\begin{align*}
\check{L}_{2(k+1)}
\=\ointctrclockwise_{0}\frac{\check{T}_{\check{y}}}{\check{y}^{2k+1}}\:\frac{d\check{y}}{2\pi\i}\\
\=2\ointctrclockwise_{X_s}\frac{\vartheta_x}{(x-X_s)^{k+1}\hat{p}_x}\frac{dx}{2\pi\i}+\{\text{terms}\propto.1\}\\
\=\frac{2}{\hat{a}_0}
\sum_{m=0}^{\infty}(-1)^m
\sum_{m_1+m_2+\cdots +m_{n-1} = m}\frac{m!}{m_1!\ldots m_{n-1}!}
\left(\frac{\hat{a}_{\ell}}{\hat{a}_0}\right)^{\sum_{\ell=1}^{n-1}m_{\ell}}
\Theta_{k-(\sum_{\ell=1}^{n-1}\ell\cdot m_{\ell})}+\{\text{terms}\propto.1\}
\:,
\end{align*}
and the commutation relation for the $\check{L}_n$ follows from that of the $\Theta_k$. 
The sum is finite in practice, by eq.\ (\ref{eq: Theta(k) operator gives zero for negative k}).
\end{proof}


\begin{claim}
For $\ell\in\Z$ and $1\leq s\leq n$ fixed, 
we define an operator $\Psi_{\ell}$ on holomorphic fields by
\begin{displaymath}
\Psi_{\ell}
=\ointctrclockwise_{\rho_{X_s}}\frac{\psi_x}{(x-X_s)^{\ell+1}}\frac{dx}{2\pi\i}\:
\end{displaymath}
where $\psi$ is the field defined in eq.\ (\ref{eq: definition of psi}).
We have
\begin{align}
\Psi_{k}
\=\sum_{m=0}^{k}\Theta_{k-m}\Theta_m
+\text{known correction terms}\:,
\label{eq: Laurent coefficient Psi in terms of Laurent coefficients Theta}
\end{align}
where $\Theta_k$ is given by eq.\ (\ref{def of Laurent coefficients of vartheta}).
\end{claim}

\begin{proof}
We have
\begin{align}
\Psi_{k}
\=\ointctrclockwise_{\rho_{1,X_s}}\frac{1}{(x_1-X_s)^{k+1}}
\ointctrclockwise_{\rho_{2,x_1}}\frac{\vartheta_1\vartheta_2}{x_2-x_1}\:\frac{dx_2dx_1}{(2\pi\i)^2}
\label{first line for Psi}
\\
\+\frac{c}{32}\ointctrclockwise_{\rho_{1,X_s}}\frac{1}{(x_1-X_s)^{k+1}}
\ointctrclockwise_{\rho_{2,x_1}}
\frac{f_{12}^2.1}{x_1-x_2}\:\frac{dx_2dx_1}{(2\pi\i)^2}\nn\\
\+\frac{1}{4}\ointctrclockwise_{\rho_{1,X_s}}\frac{1}{(x_1-X_s)^{k+1}}
\ointctrclockwise_{\rho_{2,x_1}}
f_{12}\frac{\vartheta_1+\vartheta_2}{x_1-x_2}\:\frac{dx_2dx_1}{(2\pi\i)^2}\:.
\label{third line for Psi}
\end{align}
We address line (\ref{first line for Psi}).
For $|x_1-X_s|<|x_2-X_s|$,
\begin{displaymath}
\frac{1}{x_2-x_1}
=\sum_{m=0}^{\infty}\frac{(x_1-X_s)^m}{(x_2-X_s)^{m+1}}\:, 
\end{displaymath}
so by choosing a contour enclosing both $x_1$ and $X_s$,
\begin{displaymath}
\ointctrclockwise_{\rho_{1,X_s}}\frac{1}{(x_1-X_s)^{k+1}}
\ointctrclockwise_{\rho_{2,x_1}}\frac{\vartheta_1\vartheta_2}{x_2-x_1}\:\frac{dx_2dx_1}{(2\pi\i)^2}
=\sum_{m=0}^{k}\Theta_{k-m}\Theta_m\:.
\end{displaymath}
In line (\ref{third line for Psi}), we replace accordingly
\begin{displaymath}
\frac{f_{12}}{x_2-x_1}
=(p_1+p_2+2y_1y_2)\sum_{m=0}^{\infty}\frac{m(m+1)^2}{2}\frac{(x_1-X_s)^{m}}{(x_2-X_s)^{m+3}}
\:.  
\end{displaymath}
Here for $x=x_1,x_2$, $p_x=(x-X_s)\hat{p}_x$. Taylor expansion of $\hat{p}_x$ about $x=X_s$ involves finitely many terms only.
All occuring terms in line (\ref{third line for Psi}) are either known by reference to the Laurent coefficients $\check{L}_k$ of $\check{T}_{\check{y}}$,
or they involve a square root of one of $x_1-X_s$ and $x_2-X_s$ and do not contribute.
Eq.\ (\ref{eq: Laurent coefficient Psi in terms of Laurent coefficients Theta}) follows. 
\end{proof}

$\vartheta$ admits a Galois splitting
\begin{align}\label{eq: Galois splitting of vartheta}
\vartheta_x
=\vartheta^{[1]}_x+y\vartheta^{[y]}_x\:. 
\end{align}
Note that $\vartheta^{[1]}$ and $\vartheta^{[y]}$ do in general not themselves define fields
(except when one of the two equals $\vartheta_x$).
We define
\begin{displaymath}
\langle\vartheta_x\ldots\rangle
=:\langle\vartheta^{[1]}_x\ldots\rangle+y\langle\vartheta^{[y]}_x\ldots\rangle\:.
\end{displaymath}

\begin{theorem}\label{theorem: graph rep for N-pt fct of vartheta}
Let $S(x_1,\ldots,x_N)$, $N\in\N$, be the set of oriented graphs with vertices $x_1,\ldots,x_N$, (not necessarily connected), 
subject to the following condition:
\begin{align*}
&\forall\:i=1,\ldots,N\:,\:\text{$x_i$ has at most one ingoing and at most one outgoing line,} \\
&\text{and if $(x_i,x_j)$ is an oriented line connecting $x_i$ and $x_j$ then $i\not=j$.}
\end{align*}
We have
\begin{align}\label{graph rep}
\langle\vartheta_1\ldots\vartheta_N\rangle
=\sum_{\Gamma\in S(x_1,\ldots,x_N)}G(\Gamma)\:,
\end{align}
where for $\Gamma\in S(x_1,\ldots,x_N)$,
\begin{align*}
G(\Gamma)
:=\left(\frac{c}{2}\right)^{\sharp\text{loops}}
\prod_{(x_i,x_j)\in\Gamma}\left(\frac{1}{4}\:f_{ij}\right)
\left\langle\bigotimes_{k\in {E_N}^c}\vartheta_k\right\rangle_r\:, 
\end{align*} 
where $E_N$ are the endpoints.
\end{theorem}

\begin{proof}
Cf.\ Appendix, Section \ref{proof: theorem: graph rep for N-pt fct of vartheta}.
\end{proof}

According to the graphical representation theorem, for $x_1$ close to $x_2$,
\begin{align}\label{eq.: graphical representation of vartheta's}
\alignedbox{\langle\vartheta_1\vartheta_2\rangle}
{=\frac{c}{32}f_{12}^2\1
+\frac{1}{4}f_{12}\left(\langle\vartheta_1\rangle+\langle\vartheta_2\rangle\right)
+\langle\vartheta_1\vartheta_2\rangle_r}\:.
\end{align}
We will use the splitting
\begin{align}\label{eq: Galois-splitting of two-pt function of vartheta}
\langle\vartheta_1\vartheta_2\rangle
=\langle\vartheta_1\vartheta_2\rangle^{[1]}
+y_1\langle\vartheta_1\vartheta_2\rangle^{[y_1]}
+y_2\langle\vartheta_1\vartheta_2\rangle^{[y_2]}
+y_1y_2\langle\vartheta_1\vartheta_2\rangle^{[y_1y_2]}\:,
\end{align}
where e.g.
\begin{align}
\Big[\langle\vartheta_{X_s}\vartheta^{[1]}_x\rangle\Big]_{\reg}
\=\left[\frac{c}{32}f_{X_sx}^2\1+\frac{1}{4}f_{X_sx}\left\{\langle\vartheta_{X_s}\rangle+\langle\vartheta^{[1]}_x\rangle\right\}\right]_{\reg}
+\langle\vartheta_{X_s}\vartheta^{[1]}_x\rangle_r\:,
\label{eq: regular part of <vartheta(Xs) 1-part of vartheta>}\\
\Big[\langle\vartheta_{X_s}\vartheta^{[y]}_x\rangle\Big]_{\reg}
\=\left[\frac{1}{4}f_{X_sx}\langle\vartheta^{[y]}_x\rangle\right]_{\reg}
+\langle\vartheta_{X_s}\vartheta^{[y]}_x\rangle_r
\:.\label{eq: regular part of <vartheta(Xs) y-part of vartheta>}
\end{align}

\subsection{The differential equation for $N$-point functions of $T$}
\label{Subsection: differential eq. for $N$-pt function of T}

\begin{lemma}\label{lemma: differential eq. for N-point function of T}
\begin{align}\label{differential eq. for N-point function of T}
\left(d-\frac{c}{8}\:\omega\right)\langle T(x_1)\ldots T(x_N)\rangle
\=2\sum_{s=1}^n
\frac{\xi_s}{p'_{X_s}}\:\langle\vartheta_{X_s}T(x_1)\ldots T(x_N)\rangle
\:.
\end{align} 
\end{lemma}

\begin{proof}
We change to the $y$ coordinate at $x=X_i$: 
We have  $\frac{dy}{dx}=\frac{1}{2}y\frac{p'}{p}$, so
\begin{align}\label{transformation of Virasoro field from x to y coordinate}
\hat{T}(y)\frac{[p']^2}{4p}
=T(x)-\frac{c}{12}\left(S(p)+\frac{3}{8}\left(\frac{p'}{p}\right)^2\right).1\:.
\end{align}
Here $y(X_i)=0$ and $S(p)$ is regular at $x=X_i$ (i.e.\ $p=0$), 
so can be omitted from the contour integral.
\begin{align*}
d&\langle T(x_1)\ldots T(x_N)\rangle\\
\=\frac{1}{2\pi\i}\sum_{s=1}^n
\left(\oint_{X_s}\langle T(x)T(x_1)\ldots T(x_N)\rangle\:dx\right)\:dX_s\\
\=\frac{1}{8\pi\i}\sum_{s=1}^n
\left(\ointctrclockwise_{\underset{\text{twice}}{X_s}}\frac{[p'_x]^2}{p_x}\langle\hat{T}(y)T(x_1)\ldots T(x_N)\rangle\:dx\right)\:dX_s\\
\+\frac{1}{2\pi\i}\frac{c}{32}\langle T(x_1)\ldots T(x_N)\rangle\sum_{s=1}^n
\left(\oint_{X_s}\left(\frac{p'}{p}\right)^2\:dx\right)\:dX_s\:,
\end{align*}
In the first integral on the r.h.s.\ of the last identity, we wind around $X_s$ twice.

\begin{remark}
Note that the variation formula is compatible with the OPE,
since $d$ commutes with $\frac{c/2}{(x_1-x_2)^4}$ and $\frac{1}{(x_1-x_2)^2}$ in the (ordinary) Virasoro OPE.
By induction, the singularities at $x_i=x_j$ for $1\leq i<j\leq N$ are the same on both sides of the equation.
\end{remark}

We obtain, 
by eqs (\ref{transformation of Virasoro field from x to y coordinate}) and (\ref{definition of vartheta}),
\begin{align*}
\frac{1}{8\pi\i}\ointctrclockwise_{\underset{\text{twice}}{X_s}}\frac{[p']^2}{p}\langle\hat{T}(y)T(x_1)\ldots T(x_N)\rangle\:dx\: 
\=\frac{1}{2}p'_{X_s}\:\langle\hat{T}(0)T(x_1)\ldots T(x_N)\rangle\\
\=\frac{2}{p'_{X_s}}\:\langle\vartheta_{X_s}T(x_1)\ldots T(x_N)\rangle\:
\end{align*}
Moreover,
\begin{align*}
\frac{1}{2\pi\i}\oint_{X_s}\left(\frac{p'}{p}\right)^2\:dx
\=\frac{1}{2\pi\i}\oint_{X_s}
\left(\frac{1}{(x-X_s)^2}
+\frac{2}{(x-X_s)}\sum_{j\not=s}\frac{1}{(x-X_j)}\right)\:dx
=4\sum_{j\not=s}\frac{1}{(X_s-X_j)}\:,
\end{align*}
so
\begin{align*}
\frac{1}{2\pi\i}\sum_{s=1}^n\xi_s\oint_{X_s}\left(\frac{p'}{p}\right)^2\:dx
=4\omega\:.
\end{align*}
From this follows eq.\ (\ref{differential eq. for N-point function of T}).
\end{proof}

\subsection{The differential equation for $N$-point functions of $\vartheta$}\label{Subsection: DE for N-point fct of vartheta}

\begin{lemma}\label{lemma: differential eq. for the N-pt function of vartheta}
\begin{align}
\left(d-\frac{c}{8}\:\omega\right)\langle\vartheta_1\ldots\vartheta_N\rangle
\=2\sum_{s=1}^n
\frac{\xi_s}{p'_{X_s}}\langle\vartheta_{X_s}\vartheta_1\ldots\vartheta_N\rangle\label{first line}\\
\+\langle\vartheta_1\ldots\vartheta_N\rangle\sum_{i=1}^N\frac{dp_i}{p_i}\label{second line}\\
\-\frac{c}{16}\sum_{i=1}^Np_i'\:d\left(\frac{p_i'}{p_i}\right)
\langle\vartheta_1\ldots\widehat{\vartheta_i}\ldots\vartheta_N\rangle\:,\label{third line}
\end{align}
\end{lemma}

Here 
\begin{align}
\frac{dp_x}{p_x}
\=-\sum_{s=1}^n\frac{\xi_s}{x-X_s}\:,\nn\\ 
d\left(\frac{p'}{p}\right)
\=\sum_{s=1}^n\frac{\xi_s}{(x-X_s)^2}\:.\label{d of p prime over p}
\end{align}

\begin{proof}
By induction, cf.\ Appendix, Section \ref{proof: lemma: differential eq. for the N-pt function of vartheta}.
\end{proof}

\begin{remark}\label{remark: In the differential eq. for varthetas, singular terms at x=Xs drop out}
We show that the singularities on both sides of the differential equation in Lemma \ref{lemma: differential eq. for the N-pt function of vartheta} are the same.
By eq.\ (\ref{eq.: graphical representation of vartheta's}), we have in line (\ref{first line}),
\begin{align*}
f_{xX_s}
\=\frac{p_x}{(x-X_s)^2}\\
\=\frac{p'_{X_s}}{x-X_s}+\frac{1}{2}p''_{X_s}+\frac{1}{6}p^{(3)}_{X_s}(x-X_s)+\frac{1}{24}p^{(4)}_{X_s}(x-X_s)^2+O((x-X_s)^3)\:. 
\end{align*}
So
\begin{displaymath}
\frac{1}{p'_{X_s}}f_{xX_s}^2
=\frac{p'_{X_s}}{(x-X_s)^2}
+\frac{p''_{X_s}}{x-X_s}
+\frac{1}{4}\frac{[p''_{X_s}]^2}{p'_{X_s}}+\frac{1}{3}p^{(3)}_{X_s}
+
O(x-X_s)\:,
\end{displaymath}
and the two singular terms cancel against corresponding terms of the sum in line (\ref{third line}), 
upon expansion of $p'_x$ about $x_i=X_s$.
Moreover, Taylor expansion in $y$ about $x=X_s$ yields, in line (\ref{second line}),
\begin{align*}
-\frac{\vartheta_x}{x-X_s}
=-\frac{\vartheta_{X_s}}{x-X_s}
-\frac{y}{x-X_s}\vartheta^{[y]}_{X_s}
-\frac{p_x}{x-X_s}\left(\frac{(\vartheta^{[1]})'_{X_s}}{p'_{X_s}}
+y\frac{(\vartheta^{[y]})'_{X_s}}{p'_{X_s}}\right)
+O(x-X_s)\:,
\end{align*}
and in line (\ref{first line}),
\begin{align}\label{fxXs singularity the first line}
\frac{1}{2}\frac{1}{p'_{X_s}}f_{xX_s}\left\{\vartheta_x+\vartheta_{X_s}\right\}
\=\frac{\vartheta_{X_s}}{x-X_s}
+\frac{1}{2}\frac{y}{x-X_s}\vartheta^{[y]}_{X_s}\nn\\
\+\frac{1}{2}\frac{p_x}{x-X_s}\frac{(\vartheta^{[1]})'_{X_s}}{p'_{X_s}}
+\frac{1}{2}\frac{p''_{X_s}}{p'_{X_s}}\vartheta^{[1]}_{X_s}\nn\\
\+\frac{1}{2}y\frac{p_x}{x-X_s}\frac{(\vartheta^{[y]})'_{X_s}}{p'_{X_s}}
+O(x-X_s)\:.
\end{align}
So the first term on the r.h.s.\ of eq.\ (\ref{fxXs singularity the first line})
cancels against the corresponding summand in line (\ref{second line}).
The second term on the r.h.s.\ of eq.\ (\ref{fxXs singularity the first line}) and in line (\ref{second line}), respectively, 
match the singularity on the l.h.s.\ of the differential equation,
since
\begin{align}\label{dXs applied to y}
d_{X_s}y
=\frac{\xi_s}{2}y\frac{\partial}{\partial X_s}\log p
=-\frac{\xi_s}{2}\frac{y}{x-X_s}\:.
\end{align}
and
\begin{align}
d_{X_s}\langle\vartheta_x\rangle
=d_{X_s}\left(\langle\vartheta^{[1]}_x\rangle+y\langle\vartheta^{[y]}_x\rangle\right)
\=d_{X_s}\langle\vartheta^{[1]}_x\rangle+(d_{X_s}y+yd_{X_s})\langle\vartheta^{[y]}_x\rangle\nn\\
\=d_{X_s}\langle\vartheta^{[1]}_x\rangle+y\left(d_{X_s}-\frac{\xi_s}{2}\frac{1}{x-X_s}\right)\langle\vartheta^{[y]}_x\rangle\:,
\label{eq: ODE for 1-pt fct of vartheta splits into Galois even and Galois odd part}
\end{align}
upon expansion of $\langle\vartheta^{[y]}_x\rangle$ about $x=X_s$.
We conclude that the singularities on both sides of the differential equation are the same.


\end{remark}

\begin{cor}\label{cor: diff for one-point function of vartheta}
For $N=1$, $1\leq s\leq n$ and $\xi_i=\delta_{is}$,
\begin{align}\label{eq: formula for ODE of vartheta}
\left(d_{X_s}-\frac{c}{8}\:\omega_s\right)\langle\vartheta_x\rangle
\=
-\xi_s\frac{c}{96}
p'_{X_s}S(p_x)(X_s)\1\nn\\
\+\frac{1}{2}\xi_s\frac{p''_{X_s}}{p'_{X_s}}\langle\vartheta_{X_s}\rangle\nn\\
\-\frac{1}{2}\xi_s\frac{p_x}{x-X_s}
\left(
\frac{\langle(\vartheta^{[1]})'_{X_s}\rangle}{p'_{X_s}}
+y\frac{\langle(\vartheta^{[y]})'_{X_s}\rangle}{p'_{X_s}}
\right)\nn\\
\+\frac{2\xi_s}{p'_{X_s}}\langle\vartheta_{X_s}\vartheta_x\rangle_r
+O(x-X_s)
\:.
\end{align}
Here $S(p_x)(X_s)$ is the Schwarzian derivative of $p_x$ w.r.t.\ $x$ at $x=X_s$. 
\end{cor}

\begin{proof}
The coefficient of $\1$ in lines (\ref{first line}) and (\ref{third line}) for $N=1$,
to order $O(1)$ term at $x=X_s$, equals 
\begin{align*}
\left[\frac{1}{p'_{X_s}}f_{xX_s}^2
-p'_x\:\frac{\partial}{\partial X_s}\left(\frac{p'_x}{p_x}\right)\right]_{O(1)|_{x=X_s}}
\=-\frac{1}{6}\:p'_{X_s}S(p_x)(X_s)\:,
\end{align*}
(cf.\ Remark \ref{remark: In the differential eq. for varthetas, singular terms at x=Xs drop out}).
\end{proof}

Higher genus requires more terms that are subsumed in $O(x-X_s)$.

\section{Exact results for the $(2,5)$ minimal model and arbitrary genus}

\subsection{Computation of $\psi$ and $\langle\vartheta_{X_s}\vartheta_{X_s}\rangle_r$ for arbitrary genus}

We consider the hyperelliptic genus $g\geq 1$ Riemann surface
\begin{displaymath}
\Sigma_g:\quad y^2=p_x\:,\quad g\geq 1\:,
\end{displaymath}
$\deg p=n=2g+1$, with a distinguished ramification point $x=X_s$ which is a simple zero of $p$, 
\begin{displaymath}
p_{X_s}=0\:,\quad p'_{X_s}\not=0\:. 
\end{displaymath}
Let
\begin{displaymath}
\mathcal{D}
:=-p\:\partial^2-\frac{1}{2}p'\:\partial+p''\:,
\end{displaymath}
with $\partial=\frac{\partial}{\partial x}$.
We have
\begin{align*}
\mathcal{D}\vartheta
\=p''\vartheta^{[1]}
-\frac{1}{2}p'\:(\vartheta^{[1]})'
-p\:(\vartheta^{[1]})''
+y\left\{
\frac{1}{2}p''\vartheta^{[y]}
-\frac{3}{2}p'(\vartheta^{[y]})'
-p(\vartheta^{[y]})''
\right\}\\
\=p''\left(\vartheta^{[1]}+\frac{1}{2}y\vartheta^{[y]}\right)
-\frac{1}{2}p'\left((\vartheta^{[1]})'+3y(\vartheta^{[y]})'\right)
-p\left((\vartheta^{[1]})''+y(\vartheta^{[y]})''\right)
\:.
\end{align*}
Thus in the $(2,5)$ minimal model, 
the Galois splitting of $\vartheta_x$ induces a Galois splitting of $\psi_x$ by eq.\ (\ref{eq: definition of psi}).
By means of the decomposition
\begin{align}\label{eq: decomposition of vartheta 1 vartheta 2 r}
\langle\vartheta_1\vartheta_2\rangle_r
\=\langle\vartheta^{[1]}_1\vartheta^{[1]}_2\rangle_r
+y_1y_2\langle\vartheta^{[y]}_1\vartheta^{[y]}_2\rangle_r
+y_1\langle\vartheta^{[y]}_1\vartheta^{[1]}_2\rangle_r
+y_2\langle\vartheta^{[1]}_1\vartheta^{[y]}_2\rangle_r
\:,
\end{align}
and by $\langle\vartheta_x\vartheta_x\rangle_r=\langle\psi_x\rangle$,
the Galois splitting of $\vartheta$ induces a Galois splitting of $\Psi$,
\begin{align}\label{eq: Galois splitting of Psi}
\langle\psi_x\rangle
=\langle\psi^{[1]}_x\rangle+y\langle\psi^{[y]}_x\rangle\:,
\end{align}
with
\begin{align*}
\langle\psi^{[1]}_x\rangle 
\=\langle\vartheta^{[1]}_x\vartheta^{[1]}_x\rangle_r
+p_x\langle\vartheta^{[y]}_x\vartheta^{[y]}_x\rangle_r\\
\langle\psi^{[y]}_x\rangle
\=2\langle\vartheta^{[1]}_x\vartheta^{[y]}_x\rangle_r\:.
\end{align*}

\begin{lemma}\label{Lemma: The first derivative of the Galois components of Psi at Xs}
For the Galois splitting eq.\ (\ref{eq: Galois splitting of Psi}) of $\Psi$, 
we have
\begin{align*}
\alignedbox{\langle(\psi^{[1]})'_{X_s}\rangle} 
{=\langle\vartheta_{X_s}(\vartheta^{[1]})'_{X_s}\rangle_r
+\frac{1}{2}p'_{X_s}\langle\vartheta^{[y]}_{X_s}\vartheta^{[y]}_{X_s}\rangle_r}\:.
\end{align*}
and
\begin{align*}
\alignedbox{\langle(\psi^{[y]})'_{X_s}\rangle} 
{=\langle\vartheta_{X_s}(\vartheta^{[y]})'_{X_s}\rangle_r}\:.
\end{align*}
In the $(2,5)$ minimal model, these are known.
\end{lemma}

\begin{proof}
Cf.\ Appendix, Section \ref{proof: lemma: The first derivative of the Galois components of Psi at Xs}. 
\end{proof}

\begin{claim}\label{claim: regular part of <vartheta Xs vartheta Xs>} 
We assume the $(2,5)$ minimal model.
We have
\begin{align*}
\left[\langle\vartheta_{X_s}\vartheta^{[1]}_x\rangle\right]_{\underset{x=X_s}{\reg}}
\=
\frac{c}{40}\left(\frac{1}{3}p'_{X_s}p^{(3)}_{X_s}
+\frac{7}{16}[p''_{X_s}]^2\right)\1
+\frac{9}{20}p''_{X_s}\langle\vartheta_{X_s}\rangle
+\frac{3}{20}p'_{X_s}\langle(\vartheta^{[1]})'_{X_s}\rangle
\:,\\
\left[\langle\vartheta_{X_s}\vartheta^{[y]}_x\rangle\right]_{\underset{x=X_s}{\reg}}
\=
\frac{1}{4}
\left(p''_{X_s}\langle\vartheta^{[y]}_{X_s}\rangle
+p'_{X_s}\langle(\vartheta^{[y]})'_{X_s}\rangle\right)
+\langle\vartheta_{X_s}\vartheta^{[y]}_{X_s}\rangle_r\:,
\end{align*}
where 
\begin{displaymath}
\langle\vartheta_{X_s}\vartheta^{[y]}_{X_s}\rangle_r=\frac{1}{2}\langle\psi^{[y]}_{X_s}\rangle 
\end{displaymath}
is known.
\end{claim}

\begin{proof}
Cf.\ Appendix, Section \ref{proof: claim: no pole part of <vartheta Xs vartheta x>, to leading order}.
\end{proof}

\subsection{The system of ODEs for $\1$ and $\langle\vartheta_{X_s}\rangle$}

\begin{cor}\label{cor: system of exact ODEs for one -point functions of 1 and vartheta} 
Assume the $(2,5)$ minimal model.
For $g\geq 1$, we have the system of ODEs
\begin{align}
\left(d_{X_s}-\frac{c}{8}\omega_s\right)\1
\=\frac{2\xi_s}{p'_{X_s}}\langle\vartheta_{X_s}\rangle\:,\label{ODE for 0-pt function}\\
%
\left(d_{X_s}-\frac{c}{8}\:\omega_s\right)\frac{\langle\vartheta_{X_s}\rangle}{p'_{X_s}}
\=
\frac{77}{1200}\xi_s\:
S(p_x)(X_s)\1
+\frac{2}{5}\xi_s\:\frac{p''_{X_s}}{p'_{X_s}}\frac{\langle\vartheta_{X_s}\rangle}{p'_{X_s}}
+\frac{3}{10}\xi_s\:\frac{\langle(\vartheta^{[1]})'_{X_s}\rangle}{p'_{X_s}}
\:,
\label{ODE for 1-pt function of vartheta at Xs}
\end{align}
where $S(p_x)(X_s)$ is the Schwarzian derivative w.r.t.\ $x$ evaluated at the position $x=X_s$.
Moreover,
\begin{align*}
\left(d_{X_s}-\frac{c}{8}\:\omega_s\right)\langle\vartheta_{X_s}^{[y]}\rangle
\=\frac{2\xi_s}{p'_{X_s}}\langle\vartheta_{X_s}\vartheta^{[y]}_{X_s}\rangle_r
+\frac{1}{2}\xi_s\:\langle(\vartheta^{[y]})'_{X_s}\rangle
\:.
\end{align*}
where 
\begin{displaymath}
\langle\vartheta_{X_s}\vartheta^{[y]}_{X_s}\rangle_r=\frac{1}{2}\langle\psi^{[y]}_{X_s}\rangle 
\end{displaymath}
is known.
\end{cor}

\begin{proof}
The ODEs follow from Lemma \ref{lemma: differential eq. for the N-pt function of vartheta} for $N=0$ and $N=1$, respectively,
under the assumption $\xi_i=0$ for $i\not=s$.
For eq.\ (\ref{ODE for 1-pt function of vartheta at Xs}),
the ODE is given by Corollary \ref{cor: diff for one-point function of vartheta}.
On the other hand, the l.h.s.\ is given by eq.\ (\ref{eq: ODE for 1-pt fct of vartheta splits into Galois even and Galois odd part}).
So the differential equations for the Galois even respectively the Galois odd part can be treated separately.
For the two-point function in line (\ref{eq: formula for ODE of vartheta}),
we use eq.\ (\ref{eq.: graphical representation of vartheta's}) 
and the Galois splitting (\ref{eq: decomposition of vartheta 1 vartheta 2 r}).
\begin{enumerate}
\item 
For the Galois-even part, we replace every copy of $\vartheta$ by $\vartheta^{[1]}$.
We have seen that all singularities on the r.h.s.\ drop out in Remark \ref{remark: In the differential eq. for varthetas, singular terms at x=Xs drop out}.
We argue that
\begin{align}
\left(d_{X_s}-\frac{c}{8}\omega_s\right)|_{x=X_s}\langle\vartheta^{[1]}_x\rangle
\=\left(d_{X_s}-\frac{c}{8}\omega_s\right)\langle\vartheta_{X_s}\rangle
-\langle(\vartheta^{[1]})'_{X_s}\rangle\xi_s\label{eq: exchange differentiation and evaluation}
\:,
\end{align}
where $\langle\vartheta^{[1]}_{X_s}\rangle=\langle\vartheta_{X_s}\rangle$.
Indeed, since both $\langle\quad\rangle$ and $\vartheta_{X_s}$ depend on $X_s$ (and $\vartheta^{[1]}_x$ does not), 
set $\langle\vartheta_{X_s}\rangle=f(X_s,\vartheta_{X_s})$ for some function $f$. 
Then
\begin{align*}
d_{X_s}\langle\vartheta_{X_s}\rangle
\=\xi_s\frac{\partial}{\partial X_s}f(X_s,\vartheta_{X_s})\\
\=\xi_s\frac{dx}{dX_s}\frac{\partial}{\partial x}|_{(x,y)=(X_s,\vartheta_{X_s})}f(x,y)
+\xi_s\frac{dy}{dX_s}\frac{\partial}{\partial y}|_{(x,y)=(X_s,\vartheta_{X_s})}f(x,y)
\\
\=d_{X_s}|_{x=X_s}\langle\vartheta^{[1]}_x\rangle+\xi_s\langle(\vartheta^{[1]})'_{X_s}\rangle\:.
\end{align*}
From this and from eq.\ (\ref{eq: definition of psi}) follows 
\begin{align*}
\left(d_{X_s}-\frac{c}{8}\:\omega_s\right)\langle\vartheta_{X_s}\rangle
\=
\xi_s\frac{c}{32}
\left(
\frac{1}{2}\frac{[p''_{X_s}]^2}{p'_{X_s}}-\frac{1}{3}p^{(3)}_{X_s}
\right)\1\nn\\
\+\frac{1}{2}\xi_s\frac{p''_{X_s}}{p'_{X_s}}\langle\vartheta_{X_s}\rangle\nn\\
\-\frac{1}{2}\xi_s\langle(\vartheta^{[1]})'_{X_s}\rangle
+\xi_s\langle(\vartheta^{[1]})'_{X_s}\rangle\nn\\
\-\frac{2\xi_s}{p'_{X_s}}\left(
\frac{c}{480}\left(p_{X_s}'p_{X_s}^{(3)}-\frac{3}{2}[p_{X_s}'']^2\right)\1
+\frac{1}{10}\:p'_{X_s}\:\langle(\vartheta^{[1]})'_{X_s}\rangle
-\frac{1}{5}\:p''_{X_s}\langle\vartheta_{X_s}\rangle
\right)
\:,
\end{align*}
or
\begin{align*}
\left(d_{X_s}-\frac{c}{8}\:\omega_s\right)\langle\vartheta_{X_s}\rangle
\=-\xi_s\frac{7c}{480}\left(p^{(3)}_{X_s}-\frac{3}{2}\frac{[p''_{X_s}]^2}{p'_{X_s}}\right)\1
+\xi_s\frac{9}{10}\langle\vartheta_{X_s}\rangle
+\xi_s\frac{3}{10}\langle(\vartheta^{[1]})'_{X_s}\rangle\:,
\end{align*}
and thus eq.\ (\ref{ODE for 1-pt function of vartheta at Xs}).
\item
According to eqs (\ref{eq: ODE for 1-pt fct of vartheta splits into Galois even and Galois odd part})
and (\ref{eq: formula for ODE of vartheta}), 
the differential equation for $\vartheta^{[y]}$ is given by 
\begin{align*}
\left(d_{X_s}-\frac{c}{8}\:\omega_s\right)\langle\vartheta_x^{[y]}\rangle
\=
\frac{2\xi_s}{p'_{X_s}}\langle\vartheta_{X_s}\vartheta^{[y]}_x\rangle_r
-\frac{1}{2}\xi_s\frac{p_x}{x-X_s}\frac{\langle(\vartheta^{[y]})'\rangle}{p'_{X_s}}
+O(x-X_s)
\:.
\end{align*}
Evaluating at $x=X_s$ and using the argument (\ref{eq: exchange differentiation and evaluation})
yields the claimed formula.
\end{enumerate}
This completes the proof.
\end{proof}

\subsection{The LHS of the ODEs for $\langle\vartheta^{(k)}_{X_s}\ldots\rangle$, for arbitrary genus}

\begin{claim}
We consider the Galois-even part only.
The l.h.s.\ of the differential equation for $N=1$ reads
\begin{displaymath}
\left(d_{X_s}-\frac{c}{8}\:\omega_s\right)\langle\vartheta_x\rangle
=\sum_{k=0}^{n-2}\frac{1}{k!}(x-X_s)^k\:
\left(d_{X_s}|_{x=X_s}\langle\vartheta^{(k)}_x\rangle-\frac{c}{8}\:\omega_s\langle\vartheta^{(k)}_{X_s}\rangle\right)\:,
\end{displaymath}
For $N=2$, $\langle\vartheta_1\vartheta_2\rangle$ is not differentiable at $x_2=X_s$, but we have  
\begin{displaymath}
\left(d_{X_s}-\frac{c}{8}\:\omega_s\right)\langle\vartheta_{x_1}\vartheta_{x_2}\rangle^{[1]}
=\sum_{k=0}^{n-2}\frac{1}{k!\ell!}(x_1-X_s)^k(x_2-X_s)^{\ell}\:
\left(\frac{\partial}{\partial X_s}|_{x_1,x_2=X_s}\langle\vartheta^{(k)}_{x_1}\vartheta^{(\ell)}_{x_2}\rangle^{[1]}
-\frac{c}{8}\:\omega_s\langle\vartheta^{(k)}_{X_s}\vartheta^{(\ell)}_{X_s}\rangle^{[1]}\right)\:.
\end{displaymath}
We have
\begin{align}\label{eq: taking the derivative of the two-pt function of vartheta w.r.t. Xs before resp. after evaluation at Xs}
\frac{\partial}{\partial X_s}\langle\vartheta^{(k)}_{x_1}\vartheta^{(\ell)}_{X_s}\rangle^{[1]}
=\frac{\partial}{\partial X_s}|_{x_2=X_s}\langle\vartheta^{(k)}_{x_1}\vartheta^{(\ell)}_{x_2}\rangle^{[1]}
+\langle\vartheta^{(k)}_{x_1}\vartheta^{(\ell+1)}_{X_s}\rangle^{[1]}\:.
\end{align}
\end{claim}

It is clear that this generalises to arbitrary finite $N$.

\begin{proof}
($N=1$)
We consider the Galois-even part only.
For $f(X_s,\vartheta^{(k)}_x)=\langle\vartheta^{(k)}_x\rangle_{\ldots,X_s,\ldots}$ and $k\geq0$, 
we have
\begin{displaymath}
f(X_s,\vartheta^{(k)}_x)
=f(X_s,\vartheta^{(k)}_{X_s}+\vartheta^{(k+1)}_{X_s}(x-X_s)+\frac{1}{2}\vartheta^{(k+2)}_{X_s}(x-X_s)^2+\ldots)\:.
\end{displaymath}
Since $f$ is linear in its second argument, we have
\begin{displaymath}
f(X_s,\vartheta^{(k)}_x)
=f(X_s,\vartheta^{(k)}_{X_s})
+f(X_s,\vartheta^{(k+1)}_{X_s})(x-X_s)
+\frac{1}{2}f(X_s,\vartheta^{(k+2)}_{X_s})(x-X_s)^2+\ldots\:.
\end{displaymath}
Thus 
\begin{displaymath}
\frac{\partial}{\partial X_s}|_{x=X_s}f(X_s,\vartheta^{(k)}_x)
=\frac{\partial}{\partial X_s}f(X_s,\vartheta^{(k)}_{X_s})
-f(X_s,\vartheta^{(k+1)}_{X_s})\:. 
\end{displaymath}
i.e.
\begin{displaymath}
\frac{\partial}{\partial X_s}\langle\vartheta^{(k)}_{X_s}\rangle
=\frac{\partial}{\partial X_s}|_{x=X_s}\langle\vartheta^{(k)}_x\rangle
+\langle\vartheta^{(k+1)}_{X_s}\rangle
\end{displaymath}
We apply this to the $d_{X_s}$ derivative of
\begin{displaymath}
\langle\vartheta_x\rangle
=\langle\vartheta_{X_s}\rangle+\langle(\vartheta')_{X_s}\rangle(x-X_s)+\frac{1}{2}\langle(\vartheta'')_{X_s}\rangle(x-X_s)^2+\ldots
\end{displaymath}
We have
\begin{align*}
\frac{\partial}{\partial X_s}\langle\vartheta_{X_s}\rangle
\=\frac{\partial}{\partial X_s}|_{x=X_s}\langle\vartheta_x\rangle
+\langle\vartheta'_{X_s}\rangle\\
\frac{\partial}{\partial X_s}\left\{\frac{1}{k!}\langle\vartheta^{(k)}_{X_s}\rangle(x-X_s)^k\right\}
\=\frac{1}{k!}\left(\frac{\partial}{\partial X_s}|_{x=X_s}\langle\vartheta^{(k)}_x\rangle
+\langle\vartheta^{(k+1)}_{X_s}\rangle\right)(x-X_s)^k\\
&\hspace{2.5cm}-\frac{1}{(k-1)!}\langle\vartheta^{(k)}_{X_s}\rangle(x-X_s)^{k-1}\:,\quad k\geq 1\:.
\end{align*}
Thus in the expression for $d_{X_s}\langle\vartheta_x\rangle$,
the terms $\xi_s\langle\vartheta'_{X_s}\rangle$ and $\frac{\xi_s}{k!}\langle\vartheta^{(k+1)}_{X_s}\rangle(x-X_s)^k$ drop out for $0\leq k\leq\deg\langle\vartheta_x\rangle=n-2$.

($N=2$) For $f(X_s,\vartheta^{(k)}_x\ldots)=\langle\vartheta^{(k)}_x\ldots\rangle_{\ldots,X_s,\ldots}$ and $k\geq0$, 
we have
\begin{displaymath}
f(X_s,\vartheta^{(k)}_{x_1}\vartheta^{(\ell)}_{x_2})
=\sum_{i,j}\frac{1}{j!}f(X_s,\vartheta^{(k)}_{x_1}\vartheta^{(\ell+j)}_{X_s})\:(x_2-X_s)^j\:.
\end{displaymath}
so
\begin{displaymath}
\frac{\partial}{\partial X_s}f(X_s,\vartheta^{(k)}_{x_1}\vartheta^{(\ell)}_{X_s})
=
\frac{\partial}{\partial X_s}|_{x_2=X_s}f(X_s,\vartheta^{(k)}_{x_1}\vartheta^{(\ell)}_{x_2})
+f(X_s,\vartheta^{(k)}_{x_1}\vartheta^{(\ell+1)}_{X_s})\:.
\end{displaymath}
or eq.\ (\ref{eq: taking the derivative of the two-pt function of vartheta w.r.t. Xs before resp. after evaluation at Xs}).
Also,
\begin{displaymath}
f(X_s,\vartheta^{(k)}_{x_1}\vartheta^{(\ell)}_{x_2})
=\sum_{i,j}\frac{1}{i!j!}f(X_s,\vartheta^{(k+i)}_{X_s}\vartheta^{(\ell+j)}_{X_s})\:(x_1-X_s)^i(x_2-X_s)^j\:.
\end{displaymath}
Thus 
\begin{displaymath}
\frac{\partial}{\partial X_s}f(X_s,\vartheta^{(k)}_{X_s}\vartheta^{(\ell)}_{X_s})
=\frac{\partial}{\partial X_s}|_{x_1=x_2=X_s}f(X_s,\vartheta^{(k)}_{x_1}\vartheta^{(\ell)}_{x_2})
+f(X_s,\vartheta^{(k+1)}_{X_s}\vartheta^{(\ell)}_{X_s}+\vartheta^{(k)}_{X_s}\vartheta^{(\ell+1)}_{X_s})
\end{displaymath}
i.e.
\begin{displaymath}
\frac{\partial}{\partial X_s}\langle\vartheta^{(k)}_{X_s}\vartheta^{(\ell)}_{X_s}\rangle^{[1]}
=\frac{\partial}{\partial X_s}|_{x_1=x_2=X_s}\langle\vartheta^{(k)}_{x_1}\vartheta^{(\ell)}_{x_2}\rangle^{[1]}
+\langle\vartheta^{(k+1)}_{X_s}\vartheta^{(\ell)}_{X_s}\rangle^{[1]}
+\langle\vartheta^{(k)}_{X_s}\vartheta^{(\ell+1)}_{X_s}\rangle^{[1]}
\:.
\end{displaymath}
We apply this to the $d_{X_s}$ derivative of
\begin{displaymath}
\langle\vartheta^{(k)}_{x_1}\vartheta^{(\ell)}_{x_2}\rangle^{[1]}
=\sum_{i,j}\frac{1}{i!j!}\langle\vartheta^{(k+i)}_{X_s}\vartheta^{(\ell+j)}_{X_s}\rangle^{[1]}\:(x_1-X_s)^i(x_2-X_s)^j\:.
\end{displaymath}
We have for $k=0$ 
\begin{align*}
\frac{\partial}{\partial X_s}&\langle\vartheta_{X_s}\vartheta_{X_s}\rangle\:\\
\=\frac{\partial}{\partial X_s}|_{x_1,x_2=X_s}\langle\vartheta_{x_1}\vartheta_{x_2}\rangle
+\langle\vartheta'_{X_s}\vartheta_{X_s}\rangle
+\langle\vartheta_{X_s}\vartheta'_{X_s}\rangle\\
\frac{\partial}{\partial X_s}&\left\{\frac{1}{\ell!}\langle\vartheta_{X_s}\vartheta^{(\ell)}_{X_s}\rangle\:(x_2-X_s)^{\ell}\right\}\\
\=\frac{1}{\ell!}\left(\frac{\partial}{\partial X_s}|_{x_1,x_2=X_s}\langle\vartheta_{x_1}\vartheta^{(\ell)}_{x_2}\rangle
+\langle\vartheta'_{X_s}\vartheta^{(\ell)}_{X_s}\rangle
+\langle\vartheta_{X_s}\vartheta^{(\ell+1)}_{X_s}\rangle\right)(x_2-X_s)^{\ell}\\
&\hspace{2.5cm}-\frac{1}{(\ell-1)!}\langle\vartheta_{X_s}\vartheta^{(\ell)}_{X_s}\rangle\:(x_2-X_s)^{\ell-1}\:,\quad\ell\geq 1
\end{align*}
and
\begin{align*}
\frac{\partial}{\partial X_s}&\left\{\frac{1}{k!\ell!}\langle\vartheta^{(k)}_{X_s}\vartheta^{(\ell)}_{X_s}\rangle\:(x_1-X_s)^k(x_2-X_s)^{\ell}\right\}\\
\=\frac{1}{k!\ell!}\left(\frac{\partial}{\partial X_s}|_{x_1,x_2=X_s}\langle\vartheta^{(k)}_{x_1}\vartheta^{(\ell)}_{x_2}\rangle
+\langle\vartheta^{(k+1)}_{X_s}\vartheta^{(\ell)}_{X_s}\rangle
+\langle\vartheta^{(k)}_{X_s}\vartheta^{(\ell+1)}_{X_s}\rangle\right)(x_1-X_s)^k(x_2-X_s)^{\ell}\\
&\hspace{2.5cm}-\frac{1}{(k-1)!\ell!}\langle\vartheta^{(k)}_{X_s}\vartheta^{(\ell)}_{X_s}\rangle\:(x_1-X_s)^{k-1}(x_2-X_s)^{\ell}\\
&\hspace{2.5cm}-\frac{1}{k!(\ell-1)!}\langle\vartheta^{(k)}_{X_s}\vartheta^{(\ell)}_{X_s}\rangle\:(x_1-X_s)^{k}(x_2-X_s)^{\ell-1}\:,\quad k,\ell\geq 1
\end{align*}
Thus in the expression $d_{X_s}\langle\vartheta_{x_1}\vartheta_{x_2}\rangle$,
the terms $\frac{\xi_s}{k!}\langle\vartheta^{(k+1)}_{X_s}\vartheta_{X_s}\rangle(x-X_s)^k$ 
and $\frac{1}{k)!\ell!}\langle\vartheta^{(k+1)}_{X_s}\vartheta^{(\ell)}_{X_s}\rangle\:(x_1-X_s)^{k}(x_2-X_s)^{\ell}$ drop out.
\end{proof}

\subsection{The actual number of equations}

\begin{lemma}
We assume the $(2,5)$ minimal model.
Let $\Sigma_g$ have genus $g\geq 1$ and be defined by $y^2=p_x$ where $\deg p=n$.
Suppose $\vartheta^{[y]}_x=0$.
The number of differential equations required to specify $\1$ equals a Fibonacci number.
\end{lemma}

\begin{proof}
\begin{enumerate}
\item 
Let $P_n$ be the set of ascending chains, including the empty chain, of non-negative integer numbers $\leq n-3$,
\begin{align}\label{condition on the sequence of integers}
i_1<\ldots<i_k\:,\quad |i_j-i_{j+1}|\geq 2\:,\quad 1\leq j\leq k-1\:.
\end{align}
Let $F_n=\sharp P_n$.
By considering partitions that do resp.\ do not contain the number $n$ itself, we find
\begin{displaymath}
F_n=F_{n-1}+F_{n-2}.
\end{displaymath}
Moreover, $F_1=F_2=1$ (corresponding to $P_1=P_2=\{\emptyset\}$).
Thus the $F_n$ are the Fibonacci numbers. 
It remains to show that for $n=2g+1$, $F_n$ is the number of ODEs required.
\item
For $g\geq 1$, $\1$ is obtained by integrating the ODE
\begin{displaymath}
\mathcal{D}_s\1
=\frac{2\xi_s}{p'_{X_s}}\langle\vartheta_{X_s}\rangle\:,
\end{displaymath}
(Lemma \ref{lemma: differential eq. for the N-pt function of vartheta} for $N=0$).
$\langle\vartheta_x\rangle$ is a polynomial of degree $n-2$ whose leading coefficient only is known as a function of $\1$ \cite{L:PhD14}.
Indeed, for large $x_1$,
\begin{align}\label{eq: N-pt fct of vartheta for large x}
\langle\vartheta_x\ldots\rangle
=-\frac{c}{32}(n^2-1)a_0x^{n-2}\langle\ldots\rangle+O(x^{n-3})
\:,
\end{align}
where the dots stand for holomorphic fields.
Thus $\langle\vartheta_x\rangle$ for $x$ close to $X_s$ is determined by $\1$ and $\langle\vartheta_{X_s}^{(k)}\rangle$ for $k=0,\ldots,n-3$.
Assume now $\langle\vartheta_x\rangle$ for $x$ close to $X_s$ is given.
The differential equation
\begin{align*}
\left(d_{X_s}-\frac{c}{8}\:\omega_s\right)\langle\vartheta_x\rangle
\=2\frac{\xi_s}{p'_{X_s}}\langle\vartheta_{X_s}\vartheta_x\rangle\\
\+\langle\vartheta_x\rangle\frac{dp_x}{p_x}\\
\-\frac{c}{16}p_x'\:d\left(\frac{p_x'}{p_x}\right)
\1\:,
\end{align*}
involves the two-point function $\langle\vartheta_{X_s}\vartheta_x\rangle$.
By eq.\ (\ref{eq: N-pt fct of vartheta for large x}),
$\langle\vartheta_1\vartheta_1\rangle$ for $x_1,x_2$ close to $X_s$ is determined by $\langle\vartheta_x\rangle$ 
(and thus by $\1$ and $\langle\vartheta^{(k)}_x\rangle$)
and by the derivatives $\langle\vartheta^{(k_1)}_1\vartheta^{(k_2)}_2\rangle$ for $0\leq k_1,k_2\leq n-3$.
In the $(2,5)$ minimal model, the singular terms of $\langle\vartheta_1\vartheta_2\rangle$ and their derivatives
are given by the OPE (using our previous knowledge of $\langle\vartheta^{(k)}_x\rangle$). Moreover, $\psi_{X_s}$ is given by eq.\ (\ref{eq: definition of psi})
but while all Laurent coefficients $\Psi_k$ from eq.\ (\ref{eq: Laurent coefficient Psi in terms of Laurent coefficients Theta}) are known,
the individual summands $\Theta_{k-m}\Theta_k$ of $\Psi_k$ are not.
By the commutation relations (\ref{commutation relation of the Theta's}) for the $\Theta_k$,
an exchange of the factors in $\Theta_{i_k}\Theta_{i_{k+1}}$ within an $N$-point function gives rise to additional $M$-point functions with $M<N$,
which has been dealt with before.
Thus it is sufficient to consider pairs $\Theta_{k_i}\Theta_{k_{i+1}}$ with
\begin{displaymath}
k_{i+1}\geq k_i+1\:.
\end{displaymath}
which by knowledge of $\Psi_k$ can be further restricted to 
\begin{align}\label{condition: difference in orders of derivative}
k_{i+1}\geq k_i+2\:. 
\end{align}
Proceeding inductively, 
the differential equation for the $N$-point function of the field $\vartheta_x$ involves an $(N+1)$-point function,
and only the nonsingular terms of $\langle\vartheta^{(i_1)}_1\vartheta^{(i_2)}_2\ldots\vartheta^{(i_k)}_k\rangle$ for $1\leq k\leq N+1$ are required at $X_s$.
We can write them as $\langle\Theta_{i_1}\Theta_{i_2}\ldots\Theta_{i_k}\rangle$.
By the commutation relations (\ref{commutation relation of the Theta's}),
we may assume condition (\ref{condition on the sequence of integers}) to hold.

The strictly monotonously increasing sequence $(i_j)$ is bounded from above by $n-3$, which is the highest required order of derivative of $\vartheta_x$.
The procedure using the differential equation from Lemma \ref{lemma: differential eq. for the N-pt function of vartheta} terminates
and the number of $N$-point functions for which an equation is required is $F_n$.
\end{enumerate}
\end{proof}

For $k\geq 0$ and $N\geq 1$, we have
\begin{align*}
\frac{1}{k!}\left(d_{X_s}\langle\vartheta^{(k)}_x\ldots\rangle|_{x=X_s}
-\frac{c}{8}\:\omega_s\langle\vartheta^{(k)}_{X_s}\ldots\rangle\right)
\=2
\frac{\xi_s}{p'_{X_s}}\oint_{\gamma_{X_s}}\frac{\langle\vartheta_{X_s}\vartheta_{x}\ldots\rangle}{(x-X_s)^{k+1}}\:\frac{dx}{2\pi i}\\
\-\xi_s\oint_{\gamma_{X_s}}\frac{\langle\vartheta_{x}\ldots\rangle}{(x-X_s)^{k+2}}\:\frac{dx}{2\pi i}\\
\-\frac{c}{16}\:\xi_s\langle\ldots 1\rangle\oint_{\gamma_{X_s}}\frac{p'_{x}}{(x-X_s)^{k+3}}\:\frac{dx}{2\pi i}
\:,
\end{align*}
where the dots stand for $N-1$ copies of $\vartheta_x$. 
Using that
\begin{displaymath}
\oint_{\gamma_{X_s}}\frac{\langle\vartheta_{X_s}\vartheta_{x}\ldots\rangle}{(x-X_s)^{k+1}}\:\frac{dx}{2\pi i}
=\oint_{\gamma_{X_s}}\frac{1}{(x_1-X_s)^{k+1}}\oint_{\gamma_{x_1}}\frac{\langle\vartheta_1\vartheta_2\ldots\rangle}{(x_2-x_1)}\:\frac{dx_2}{2\pi i}\frac{dx_1}{2\pi i}\\
\end{displaymath}
and the OPE (\ref{map: OPE of vartheta}) for $\vartheta_1\otimes\vartheta_2$, 
the $(N+1)$-point function $\langle\vartheta_{X_s}\vartheta_{x}\ldots\rangle$ maps back to $M$-point functions with $M\leq N$.
Thus
\begin{align*}
\oint_{\gamma_{X_s}}\frac{\langle\vartheta_{X_s}\vartheta_{x}\ldots\rangle}{(x-X_s)^{k+1}}\:\frac{dx}{2\pi i}
\=\langle\Psi_k\ldots\rangle
-\frac{c}{32}\langle\ldots\rangle\ointctrclockwise_{\rho_{X_s}}\frac{1}{(x_1-X_s)^{k+1}}
\ointctrclockwise_{\rho_{x_1}}
\frac{f_{12}^2}{x_1-x_2}\:\frac{dx_2dx_1}{(2\pi\i)^2}\nn\\
\-\frac{1}{4}\ointctrclockwise_{\rho_{X_s}}\frac{1}{(x_1-X_s)^{k+1}}
\ointctrclockwise_{\rho_{x_1}}
f_{12}\frac{\langle\vartheta_1\ldots\rangle+\langle\vartheta_2\ldots\rangle}{x_1-x_2}\:\frac{dx_2dx_1}{(2\pi\i)^2}\:.
\end{align*}
Here the dots stand for $N-1$ copies of $\vartheta$, or of their derivatives.

Using the OPE (\ref{map: OPE of vartheta}) for $\vartheta_1\otimes\vartheta_2$, 
the $(N+1)$-point function $\langle\vartheta_{X_s}\vartheta_{x}\rangle$ maps back to $M$-point functions with $M\leq N$,
and to $\langle\psi_{X_s}\rangle$. The singuar terms are known in terms of $\1$ and $\langle\vartheta_x\rangle$,
which by our counting argument are supposed to be known.
$\langle\psi_x\rangle$ for $x$ close to $X_s$ is determined by its Laurent coefficients $\Psi_k$.

For $N\geq 1$ and for $k\geq 0$, we obtain from the differential equation in Lemma \ref{lemma: differential eq. for the N-pt function of vartheta}
\begin{align*}
\frac{1}{k!}\Big(d_{X_s}\langle\vartheta^{(k)}_1\vartheta_2\ldots\vartheta_N\rangle|_{x_1=X_s}
\-\frac{c}{8}\:\omega_s\langle\vartheta^{(k)}_{X_s}\vartheta_2\ldots\vartheta_N\rangle\Big)\\
\=2
\frac{\xi_s}{p'_{X_s}}\oint_{\gamma_A}\frac{\langle\vartheta_{X_s}\vartheta_1\ldots\vartheta_N\rangle}{(x_1-X_s)^{k+1}}\:\frac{dx_1}{2\pi i}\\
\-\xi_s\sum_{i=1}^N\oint_{\gamma}\frac{1}{(x_i-X_s)}\frac{\langle\vartheta_1\ldots\vartheta_N\rangle}{(x_1-X_s)^{k+1}}\:\frac{dx_1}{2\pi i}\\
\-\frac{c}{16}\:\xi_s\sum_{i=1}^N\oint_{\gamma}\frac{p'_i}{(x_i-X_s)^2}\frac{\langle\vartheta_1\ldots\widehat{\vartheta_i}\ldots\vartheta_N\rangle}{(x_1-X_s)^{k+1}}\:\frac{dx_1}{2\pi i}
\:.
\end{align*}
For $N\geq 1$ and for $k\geq 0$, we obtain from the differential equation in Lemma \ref{lemma: differential eq. for the N-pt function of vartheta}
\begin{align*}
\frac{1}{k!}\Big(d_{X_s}\langle\vartheta^{(k_1)}_1\ldots\vartheta^{(k_N)}_N\rangle|_{x_i=X_s}
\-\frac{c}{8}\:\omega_s\langle\vartheta^{(k)}_{X_s}\vartheta_2\ldots\vartheta_N\rangle\Big)\\
\=2
\frac{\xi_s}{p'_{X_s}}\oint_{\gamma_A}\frac{\langle\vartheta_{X_s}\vartheta_1\ldots\vartheta_N\rangle}{(x_1-X_s)^{k+1}}\:\frac{dx_1}{2\pi i}\\
\-\xi_s\sum_{i=1}^N\oint_{\gamma}\frac{1}{(x_i-X_s)}\frac{\langle\vartheta_1\ldots\vartheta_N\rangle}{(x_1-X_s)^{k+1}}\:\frac{dx_1}{2\pi i}\\
\-\frac{c}{16}\:\xi_s\sum_{i=1}^N\oint_{\gamma}\frac{p'_i}{(x_i-X_s)^2}\frac{\langle\vartheta_1\ldots\widehat{\vartheta_i}\ldots\vartheta_N\rangle}{(x_1-X_s)^{k+1}}\:\frac{dx_1}{2\pi i}
\:.
\end{align*}

\begin{table}[ht]
\centering
\begin{tabular}{|l|l|l|l|l|l|l|}
\hline
$n$&$g$&$\deg\langle\vartheta\rangle$&$2g-2$&$0\leq k_i\leq n-3$&Diff.\ eq.\ ($d_s|_{x=X_s}$)  required for&$\sharp$ diff.\\ 
&&$=n-2$&&$(k_I)_I$: $k_{i+1}\geq k_i+2$&&eqs\\ 
\hline
$3$&$1$&$1$&$0$&$\emptyset,0$&$\1,\langle\vartheta\rangle$&$2$\\
\hline
$4$&$1$&$2$&$0$&$\emptyset,0,\underline{1}$&$\1,\langle\vartheta\rangle,\underline{\langle\vartheta'\rangle}$&$3$ ($\underline{2}$)\\
\hline
$5$&$2$&$3$&$2$&$\emptyset,0,1,2$&$\1,\langle\vartheta\rangle,\langle\vartheta'\rangle,\langle\vartheta''\rangle$&$5$\\
&&&&$02$&$\langle\vartheta\vartheta''\rangle$&\\
\hline
$6$&$2$&$4$&$2$&$\emptyset,0,1,2,\underline{3}$&$\1,\langle\vartheta\rangle,\langle\vartheta'\rangle,\langle\vartheta''\rangle,\underline{\langle\vartheta^{(3)}\rangle}$&$8$ ($\underline{7}$)\\
&&&&$02,03,13$&$\langle\vartheta\vartheta''\rangle$,$\langle\vartheta\vartheta^{(3)}\rangle$,$\langle\vartheta'\vartheta^{(3)}\rangle$&\\
\hline
$7$&$3$&$5$&$4$&$\emptyset,0,1,2,3,4$&$\1,\langle\vartheta\rangle,\langle\vartheta'\rangle,\ldots,\langle\vartheta^{(4)}\rangle$&$13$\\
&&&&$02,03,04,13,14,24$&$\langle\vartheta\vartheta''\rangle,\ldots,\langle\vartheta\vartheta^{(4)}\rangle$,$\langle\vartheta'\vartheta^{(3)}\rangle$,$\langle\vartheta'\vartheta^{(4)}\rangle$,$\langle\vartheta''\vartheta^{(4)}\rangle$&\\
&&&&$024$&$\langle\vartheta\vartheta''\vartheta^{(4)}\rangle$&\\
\hline
\end{tabular}\\
The counting of data required in the $(2,5)$ minimal model to establish the differential equation for $\1$ in genus $g$.
Here $\vartheta=\vartheta^{[1]}$.
The underlined data are not actually required since for even $n$, the \textsl{two} first leading coefficients of $\langle\vartheta\rangle$ are known.
\end{table}

%

\pagebreak

\section{Explicit results for the $(2,5)$ minimal model and $g=2$}

\subsection{The two-point function of $\vartheta$ for $g=2$}

\begin{claim}\label{claim: the field B}
We assume $n=5$ and the $(2,5)$ minimal model. 
We have the Galois splitting
\begin{align}\label{eq: splitting of two-point function of vartheta into Galois-even and Galois-odd part}
\langle\vartheta_1\vartheta_2\rangle
=\langle\vartheta_1\vartheta_2\rangle^{[1]}
+y_1y_2\langle\vartheta_1\vartheta_2\rangle^{[y_1y_2]}\:,
\end{align}
Here
\begin{align}\label{eq: formula Galois-even part of two-pt function of vartheta when n=5}
\langle\vartheta_1\vartheta_2\rangle^{[1]}
\=\frac{c}{4}\frac{p_1p_2}{(x_1-x_2)^4}\1
+\frac{c}{32}\frac{p'_1p'_2}{(x_1-x_2)^2}\1
+\frac{1}{2}\:\frac{p_1\langle\vartheta_2\rangle+p_2\langle\vartheta_1\rangle}{(x_1-x_2)^2}\nn\\
\+\frac{7}{50}(p_1''\langle\vartheta_2\rangle+p_2''\langle\vartheta_1\rangle)
+\frac{21c}{4000}p_1''p_2''\1\nn\\
\+B(x_1,x_2)+C\:(x_1-x_2)^2\:,
\end{align}
where $C$ is constant in position and
\begin{align*}
\beta_x
:=B(x,x)
\=\left\{-\frac{7c}{960}p'_xp^{(3)}_x+\frac{91c}{16000}[p''_x]^2+\frac{c}{8}\frac{p_xp^{(4)}_x}{24}\right\}\1\\
\+\frac{1}{20}p_x\langle\vartheta''_x\rangle
+\frac{3}{20}p'_x\langle\vartheta'_x\rangle-\frac{2}{25}p''_x\langle\vartheta_x\rangle
\end{align*}
is a polynomial of order $4$.
Moreover, when multiplied by $\frac{2}{p'_{X_s}}$,
\begin{displaymath}
\frac{c}{32}\frac{p'_1p'_{X_s}}{(x_1-X_s)^2}\1 
\end{displaymath}
cancels against line (\ref{third line}) in the ODE 
given by Lemma \ref{lemma: differential eq. for the N-pt function of vartheta} for $\langle\vartheta_1\rangle$ ($N=1$) when $\xi_t=0$ for $t\not=s$. 
In eq.\ (\ref{eq: splitting of two-point function of vartheta into Galois-even and Galois-odd part}),
\begin{align*}
\langle\vartheta_1\vartheta_2\rangle^{[y_1y_2]}
=\frac{c}{8}\frac{p_1+p_2}{(x_1-x_2)^4}\1
+\frac{1}{2}\frac{\langle\vartheta_1\rangle+\langle\vartheta_2\rangle}{(x_1-x_2)^2}
-\left(\frac{c}{16}\frac{p_1^{(4)}+p_2^{(4)}}{24}+\frac{1}{8}\langle\vartheta_1''+\vartheta_2''\rangle\right)\:.
\end{align*}
We have
\begin{align}\label{enhanced OPE of vartheta}
\langle\vartheta_1\vartheta_2\rangle
\=\frac{c}{32}f_{12}^2\1
+\frac{1}{4}f_{12}\langle\vartheta_1+\vartheta_2\rangle
+\frac{1}{2}\langle\psi_1+\psi_2\rangle\nn\\
\+\left(\frac{c}{16}\frac{p_1^{(4)}+p_2^{(4)}}{24}+\frac{1}{8}\langle\vartheta_1''+\vartheta_2''\rangle\right)\frac{(y_1-y_2)^2}{2}
+O((x_1-x_2)^2)\:
\end{align}
where terms up to $O((x_1-x_2)^2)$ are known.
\end{claim}

\begin{proof}
For $n=5$, $\Theta_x^{[y]}=0$ and we have \cite{L:2013} 
\begin{displaymath}
\langle\vartheta_x\rangle
=\frac{1}{4}\Theta_x^{[1]}
=-\frac{3c}{4}a_0x^3\1
+\frac{1}{4}\A x^2+O(x)\:,
\end{displaymath}
so as $x_1\rechts\infty$, 
\begin{displaymath}
\vartheta_1
=-\frac{3c}{4}a_0x_1^3.1
+O(x_1^2)\:, 
\end{displaymath}
and
\begin{align}\label{asymptotics of the two-point function of vartheta for large x when n=5}
\langle\vartheta_1\vartheta_2\rangle
=-\frac{3c}{4}a_0x_1^3\langle\vartheta_2\rangle
+O(x_1^2)
\:.
\end{align}
On the other hand, in eq.\ (\ref{eq.: graphical representation of vartheta's}),
\begin{align*}
f_{12}
\=\frac{p_1+p_2}{(x_1-x_2)^2}+2y_1y_2\frac{1}{(x_1-x_2)^2}\:,\\
f_{12}^2
\=\frac{p_1^2+6p_1p_2+p_2^2}{(x_1-x_2)^4}+4y_1y_2\frac{p_1+p_2}{(x_1-x_2)^4}\:.
\end{align*}
In eq.\ (\ref{eq: Galois-splitting of two-pt function of vartheta}).
the contribution
\begin{displaymath}
y_1\langle\vartheta_1\vartheta_2\rangle^{[y_1]}
+y_2\langle\vartheta_1\vartheta_2\rangle^{[y_2]}
\end{displaymath}
must be contained in $\langle\vartheta_1\vartheta_2\rangle_r$
and therefore equal zero for degree reasons.
This yields eq.\ (\ref{eq: splitting of two-point function of vartheta into Galois-even and Galois-odd part}).
\begin{enumerate}
\item 
The terms $\propto y_1y_2$ in the singular part of eq.\ (\ref{eq.: graphical representation of vartheta's})
are degree violating and must be compensated for by terms in $\langle\vartheta_1\vartheta_2\rangle_r$.
$\langle\vartheta_1\vartheta_2\rangle^{[y_1y_2]}$ is a rational function in $x_1$ and $x_2$ which vanishes for $x_1\rechts\infty$.
Indeed,
setting $\vartheta_x=\vartheta_x^{[1]}+y\:\vartheta_x^{[y]}$, we have
\begin{displaymath}
2y_2\langle\vartheta_1\vartheta_2^{[y]}\rangle
=
\langle\vartheta_1(\vartheta_2^{[1]}+y_2\vartheta_2^{[y]})\rangle
-\langle\vartheta_1(\vartheta_2^{[1]}-y_2\vartheta_2^{[y]})\rangle
=O(x_2^2)\:
\end{displaymath}
in the large $x_2$ limit, by eq.\ (\ref{asymptotics of the two-point function of vartheta for large x when n=5}). 
Thus
\begin{align}\label{asymptotics of the y1y2-part of the two-point function of vartheta for large x}
\langle\vartheta_1\vartheta_2\rangle^{[y_1y_2]}
=\langle\vartheta_1^{[y]}\vartheta_2^{[y]}\rangle
=O(x_2^{-0.5})\:.
\end{align}
As $x_1\rechts\infty$,
\begin{displaymath}
[f_{12}^2]^{[y_1y_2]}
=4a_0(x_1+4x_2)+4a_1+O(x_1^{-1})\:,
\end{displaymath}
and 
\begin{displaymath}
[f_{12}\langle\vartheta_1\rangle]^{[y_1y_2]}
=[f_{12}]^{[y_1y_2]}\langle\vartheta_1\rangle
=-\frac{3c}{2}a_0(x_1+2x_2)\1
+\frac{1}{2}\A
+O(x_1^{-1})\:.
\end{displaymath}
We conclude that for $x_1\rechts\infty$,
\begin{displaymath}
\left[\frac{c}{32}f_{12}^2\1
+\frac{1}{4}f_{12}\langle\vartheta_1\rangle\right]^{[y_1y_2]}
=-\frac{c}{4}a_0(x_1+x_2)\1
+\frac{c}{8}a_1\1
+\frac{1}{8}\A
+O(x_1^{-1})\:.
\end{displaymath}
Thus we compensate by addition of $y_1y_2C$, where
\begin{align}\label{expression: correction term propto y1y2}
C
=-\left(\frac{c}{16}\frac{p_1^{(4)}+p_2^{(4)}}{24}+\frac{1}{8}\langle\vartheta_1''+\vartheta_2''\rangle\right)
=\frac{c}{4}a_0(x_1+x_2)\1
-\frac{1}{8}\A
\:. 
\end{align}
\item
The term $\frac{c}{32}f_{12}^2\1$:
\begin{align*}
\frac{p_1^2+6p_1p_2+p_2^2}{(x_1-x_2)^4}
\=\frac{8p_1p_2}{(x_1-x_2)^4}
+\frac{1}{(x_1-x_2)^2}\left(\frac{p_1-p_2}{x_1-x_2}\right)^2\\
\=\frac{8p_1p_2}{(x_1-x_2)^4}
+\frac{p'_1p'_2}{(x_1-x_2)^2}
+\frac{1}{4}p_1''p_2''
-\frac{1}{12}(p_1'p_2^{(3)}+p_2'p_1^{(3)})
+O((x_1-x_2)^2)\\
\=\frac{8p_1p_2}{(x_1-x_2)^4}
+\frac{p'_1p'_2}{(x_1-x_2)^2}
-\frac{1}{12}\left\{\left(p_1'p_2^{(3)}+p_2'p_1^{(3)}\right)-\frac{3}{2}\left([p_1'']^2+[p_2'']^2\right)\right\}
+O((x_1-x_2)^2)
\end{align*}
The term $\frac{1}{4}f_{12}\left(\vartheta_1+\vartheta_2\right)$:
\begin{align*}
\frac{(p_1+p_2)(\langle\vartheta_1\rangle+\langle\vartheta_2\rangle)}{(x_1-x_2)^2}
\=2\:\frac{p_1\langle\vartheta_2\rangle+p_2\langle\vartheta_1\rangle}{(x_1-x_2)^2}
+\frac{p_1-p_2}{x_1-x_2}\frac{\langle\vartheta_2\rangle-\langle\vartheta_1\rangle}{x_2-x_1}\\
\=2\:\frac{p_1\langle\vartheta_2\rangle+p_2\langle\vartheta_1\rangle}{(x_1-x_2)^2}
+\frac{1}{2}(p_1'\langle\vartheta'_2\rangle+p_2'\langle\vartheta'_1\rangle)+O((x_1-x_2)^2)\:.
\end{align*}
Introduce 
\begin{displaymath}
\tilde{\vartheta}_x
:=\vartheta_x+\frac{3c}{80}p''_x.1\:,\quad
\deg\tilde{\vartheta}_x=2\:.
\end{displaymath}
Correcting the order violating singular terms and omitting the order violating regular terms in the previous expansions yields
\begin{align*}
\langle\tilde{\vartheta}_1\tilde{\vartheta}_2\rangle
\=\left[\langle\vartheta_1\vartheta_2\rangle\right]_{\text{order $\leq 2$}}
+B(x_1,x_2)+C(x_1-x_2)^2\\
\=\left[\frac{c}{32}f_{12}^2\1+\frac{1}{4}f_{12}\left(\vartheta_1+\vartheta_2\right)+\langle\vartheta_1\vartheta_2\rangle_r\right]_{\text{order $\leq 2$}}
+B(x_1,x_2)+C(x_1-x_2)^2\\
\=\frac{c}{4}\frac{p_1p_2}{(x_1-x_2)^4}\1
+\frac{c}{32}\frac{p'_1p'_2}{(x_1-x_2)^2}\1
+\frac{1}{2}\:\frac{p_1\langle\vartheta_2\rangle+p_2\langle\vartheta_1\rangle}{(x_1-x_2)^2}
+y_1y_2\left(\frac{c}{8}\frac{p_1+p_2}{(x_1-x_2)^4}+\frac{1}{2}\frac{\langle\vartheta_1+\vartheta_2\rangle}{(x_1-x_2)^2}\right)\\
\-\frac{1}{40}(p_1''\langle\vartheta_2\rangle+p_2''\langle\vartheta_1\rangle)-\frac{3c}{3200}p_1''p_2''
-y_1y_2\left(\frac{c}{16}\frac{p_1^{(4)}+p_2^{(4)}}{24}+\frac{1}{8}\langle\vartheta_1''+\vartheta_2''\rangle\right)\\
\+B(x_1,x_2)+C(x_1-x_2)^2\:,
\end{align*}
where $B(x_1,x_2)$ is a symmetric polynomial in $x_1$ and $x_2$ of order $\ord_iB(x_1,x_2)=2$ for $i=1,2$.
(The second line contains the order correcting terms of the singular terms.)
On the other hand, by the OPE (\ref{map: OPE of vartheta}) and by eq.\ (\ref{eq: definition of psi}),
\begin{align*}
\langle\tilde{\vartheta}_1\tilde{\vartheta}_2\rangle 
\=\langle\vartheta_1\vartheta_2\rangle
+\frac{3c}{80}(p''_1\langle\vartheta_2\rangle+p''_2\langle\vartheta_1\rangle)
+\left(\frac{3c}{80}\right)^2p''_1p''_2\\
\=\frac{c}{32}f_{12}^2\1
+\frac{1}{4}f_{12}\left(\langle\vartheta_1\rangle+\langle\vartheta_2\rangle\right)
+\langle\psi_1\rangle+O(x_1-x_2)
+\left(\frac{3c}{80}\right)^2p''_1p''_2
+\frac{3c}{80}(p''_1\langle\vartheta_2\rangle+p''_2\langle\vartheta_1\rangle)
\\
\=\frac{c}{4}\frac{p_1p_2}{(x_1-x_2)^4}\1
+\frac{c}{32}\frac{p'_1p'_2}{(x_1-x_2)^2}\1
+\frac{1}{2}\:\frac{p_1\langle\vartheta_2\rangle+p_2\langle\vartheta_1\rangle}{(x_1-x_2)^2}
+y_1y_2\left(\frac{c}{8}\frac{p_1+p_2}{(x_1-x_2)^4}+\frac{1}{2}\frac{\langle\vartheta_1+\vartheta_2\rangle}{(x_1-x_2)^2}\right)\\
\+\frac{3c}{40}p''_1\langle\vartheta_1\rangle
+\frac{9c^2}{6400}[p''_1]^2\\
\-\frac{7c}{960}(p_1'p_1^{(3)}-\frac{3}{2}[p''_1]^2)
+\frac{1}{5}p''_1\langle\vartheta_1\rangle
+\frac{3}{20}p'_1\langle\vartheta_1'\rangle
-\frac{1}{5}p_1\langle\vartheta''_1\rangle
+O(x_1-x_2)\\
\end{align*}
By comparison, we obtain
\begin{align*}
\langle\vartheta_1\vartheta_2\rangle^{[1]}
\=\frac{c}{4}\frac{p_1p_2}{(x_1-x_2)^4}\1
+\frac{c}{32}\frac{p'_1p'_2}{(x_1-x_2)^2}\1
+\frac{1}{2}\:\frac{p_1\langle\vartheta_2\rangle+p_2\langle\vartheta_1\rangle}{(x_1-x_2)^2}\\
\-\left(\frac{3c}{80}+\frac{1}{40}\right)(p_1''\langle\vartheta_2\rangle+p_2''\langle\vartheta_1\rangle)
-\left(\left(\frac{3c}{80}\right)^2+\frac{3c}{3200}\right)p_1''p_2''\1\\
\+B(x_1,x_2)+C\:(x_1-x_2)^2\:,
\end{align*}
where 
\begin{align*}
B(x,x)
\=-\frac{7c}{960}p'_xp^{(3)}_x
+\left(\frac{9c^2}{6400}+\frac{3c}{3200}+\frac{7c}{640}\right)[p''_x]^2
+\frac{c}{8}\frac{p_xp^{(4)}_x}{24}\\
\+\frac{1}{20}p_x\langle\vartheta''_x\rangle
+\frac{3}{20}p'_x\langle\vartheta'_x\rangle
+\left(\frac{1}{5}+\frac{3c}{40}+\frac{1}{20}\right)p''_x\langle\vartheta_x\rangle\:
\end{align*}
and thus as required.
$B(x,x)$ is a polynomial of order $4$, though it is not manifestly so.
\end{enumerate}
\end{proof}

\begin{remark}\label{remark on the symmetric quadratic polynomial B}
 We have
 \begin{center}
 \begin{tabular}{lll}
 $\beta_x$&$=$&$B(x,x)$\\
 $\beta'_x$&$=$&$(\partial_1B+\partial_2B)|_{x_1=x_2=x}$\\
 $\frac{1}{6}\beta''_x$&$=$&$\frac{1}{2}(\partial_1^2B+\partial_2^2B)|_{x_1=x_2=x}$\\
 &$=$&$2\partial_1\partial_2B|_{x_1=x_2=x}$\\
 $\frac{1}{6}\beta^{(3)}_x$&$=$&$\frac{1}{2}(\partial_2^2\partial_1B+\partial_1^2\partial_2B)|_{x_1=x_2=x}$
 \end{tabular} 
 \end{center}
 For evaluating the contour integral, the corresponding non-symmetric formulations are more suitable,
 \begin{tabular}{lll}
 $\beta_x$&$=$&$B(x,x)$\\
 $\frac{1}{2}\beta'_x$&$=$&$(\partial_2B)(x,x)$\\
 $\frac{1}{6}\beta''_x$&$=$&$\partial_2^2B(x,x)$\\
 $\frac{1}{6}\beta^{(3)}_x$&$=$&$(\partial_1^2\partial_2B)(x,x)$
 \end{tabular}
\end{remark}

\begin{claim}\label{claim: derivatives of beta}
We have
\begin{align*}
\beta'_x
\=\left\{\frac{49c}{12000}p''_xp^{(3)}_x
-\frac{c}{480}p'_xp^{(4)}_x
+\frac{c}{300}p_xp^{(5)}\right\}\1\\
\+\frac{1}{5}p'_x\langle\vartheta''_x\rangle
+\frac{7}{100}p''_x\langle\vartheta'_x\rangle
-\frac{2}{25}p^{(3)}_x\langle\vartheta_x\rangle
\end{align*}
and
\begin{align*}
\beta^{(3)}_x
\=\left(\frac{61c}{6000}p^{(3)}_xp^{(4)}_x
-\frac{23c}{1600}p''_xp^{(5)}\right)\1\\
\+\frac{13}{50}p^{(3)}_x\langle\vartheta''_x\rangle
-\frac{9}{100}p^{(4)}_x\langle\vartheta'_x\rangle
-\frac{2}{25}p^{(5)}_x\langle\vartheta_x\rangle
\end{align*}
\end{claim}

\begin{proof}
Direct computation, using that
\begin{align}\label{eq: third order derivative of vartheta when n=5}
\langle\vartheta^{(3)}\rangle
=-\frac{3c}{80}p^{(5)}\1\:,
\quad
p^{(5)}
=120 a_0\:.
\end{align}
\end{proof}

\subsection{The system of exact ODEs for $g=2$ ($n=5$)}

\begin{example}\label{Example: some expansions worked out to higher order, for n=5}
Let $g=2$ ($n=5$) and $a_0=\frac{1}{5!}p^{(5)}_{X_s}$.
\begin{align*}
f_{xX_s}
\=\frac{p_x}{(x-X_s)^2}\\
\=\frac{p'_{X_s}}{x-X_s}
+\frac{1}{2}p''_{X_s}
+\frac{1}{6}p^{(3)}_{X_s}(x-X_s)
+\frac{1}{24}p^{(4)}_{X_s}(x-X_s)^2
+\frac{1}{120}p^{(5)}(x-X_s)^3\:. 
\end{align*}
So
\begin{align*}
\frac{1}{p'_{X_s}}f_{xX_s}^2
\=\frac{p'_{X_s}}{(x-X_s)^2}
+\frac{p''_{X_s}}{x-X_s}
+\frac{1}{4}\frac{[p''_{X_s}]^2}{p'_{X_s}}+\frac{1}{3}p^{(3)}_{X_s}\\
\+\frac{1}{6}
\left(\frac{p''_{X_s}p^{(3)}_{X_s}}{p'_{X_s}}+\frac{1}{2}p^{(4)}_{X_s}\right)(x-X_s)\\
\+\frac{1}{12}
\left(
\frac{1}{3}\frac{[p^{(3)}]^2}{p'_{X_s}}+\frac{1}{2}\frac{p''_{X_s}p^{(4)}_{X_s}}{p'_{X_s}}+\frac{1}{5}p^{(5)}_{X_s}\right)(x-X_s)^2\\
\+O((x-X_s)^3)\:.
\end{align*}
When $n=5$, we have $\vartheta^{[1]}=\vartheta$ ($\vartheta^{[y]}$ is absent).
In line (\ref{first line}),
\begin{align*}
\frac{1}{2}\frac{1}{p'_{X_s}}f_{xX_s}&\left\{\vartheta_x+\vartheta_{X_s}\right\}\\
\=\frac{\vartheta_{X_s}}{x-X_s}\\
\+\frac{1}{2}\frac{p''_{X_s}}{p'_{X_s}}\vartheta_{X_s}+\frac{1}{2}\vartheta'_{X_s}\\
\+\left(\frac{1}{6}\frac{p^{(3)}_{X_s}}{p'_{X_s}}\vartheta_{X_s}
+\frac{1}{4}\frac{p''_{X_s}}{p'_{X_s}}\vartheta'_{X_s}
+\frac{1}{4}\vartheta''_{X_s}\right)
(x-X_s)\\
\+\left(\frac{1}{24}\frac{p^{(4)}_{X_s}}{p'_{X_s}}\vartheta_{X_s}
+\frac{1}{12}\frac{p^{(3)}_{X_s}}{p'_{X_s}}\vartheta'_{X_s}
+\frac{1}{8}\frac{p''_{X_s}}{p'_{X_s}}\vartheta''_{X_s}
+\frac{1}{12}\vartheta^{(3)}_{X_s}\right)
(x-X_s)^2\\
\+O((x-X_s)^3)\:.
\end{align*}
We use eq.\ (\ref{eq: third order derivative of vartheta when n=5}).
In line (\ref{second line}), we have
\begin{align*}
p'_x\:d_{X_s}\left(\frac{p'_x}{p_x}\right)
\=\Big(p'_{X_s}+p''_{X_s}(x-X_s)+\frac{1}{2}p^{(3)}_{X_s}(x-X_s)^2+\frac{1}{6}p^{(4)}_{X_s}(x-X_s)^3+\frac{1}{24}p^{(5)}(x-X_s)^4\Big)\:
\frac{\xi_s}{(x-X_s)^2}\\
\=\frac{p'_{X_s}}{(x-X_s)^2}+\frac{p''_{X_s}}{x-X_s}+\frac{1}{2}p^{(3)}_{X_s}+\frac{1}{6}p^{(4)}_{X_s}(x-X_s)+\frac{1}{24}p^{(5)}(x-X_s)^2\:.
\end{align*}
Moreover,
\begin{align*}
\langle\vartheta_{X_s}\vartheta_x\rangle_r
=\langle\psi_{X_s}\rangle
+\langle\vartheta_{X_s}\vartheta'_{X_s}\rangle_r(x-X_s)
+\frac{1}{2}\langle\vartheta_{X_s}\vartheta''_{X_s}\rangle_r(x-X_s)^2
+O((x-X_s)^3)
\end{align*}
where by Lemma \ref{Lemma: The first derivative of the Galois components of Psi at Xs} and eq.\ (\ref{eq: definition of psi}),
\begin{align*}
2\langle\vartheta_{X_s}\vartheta'_{X_s}\rangle_r
=\langle\psi'_{X_s}\rangle
\=-\frac{c}{480}
\left(p'_{X_s}p^{(4)}_{X_s}-2p''_{X_s}p^{(3)}\right)\1
+\frac{1}{5}p^{(3)}_{X_s}\langle\vartheta_{X_s}\rangle
+\frac{1}{10}p''_{X_s}\langle\vartheta'_{X_s}\rangle
-\frac{3}{10}p'_{X_s}\langle\vartheta''_{X_s}\rangle\:.
\end{align*}
\end{example}

\begin{cor}
Let $n=5$. 
The values of the following integral as a function of $X_s$:
\begin{displaymath}
\frac{2}{p'_{X_s}}\oint\frac{\langle\vartheta_{X_s}\vartheta_{x}\rangle}{(x-X_s)^{k+1}}\:\frac{dx}{2\pi i}\:. 
\end{displaymath}
For $k=0$:
\begin{align}
&\frac{c}{20}\left(\frac{1}{3}p^{(3)}_{X_s}+\frac{7}{16}\frac{[p''_{X_s}]^2}{p'_{X_s}}\right)\1 
+\frac{9}{10}\frac{p''_{X_s}}{p'_{X_s}}\langle\vartheta_{X_s}\rangle
+\frac{3}{10}\langle\vartheta'_{X_s}\rangle
\label{eq: k=0 integral for two-pt function of vartheta with one position at Xs}
\end{align}
For $k=1$:
\begin{align}
&\frac{c}{120}\left(\frac{7}{4}\frac{p''_{X_s}p^{(3)}_{X_s}}{p'_{X_s}}+p^{(4)}_{X_s}\right)\1
+\frac{11}{30}\frac{p^{(3)}_{X_s}}{p'_{X_s}}\langle\vartheta_{X_s}\rangle
+\frac{7}{20}\frac{p''_{X_s}}{p'_{X_s}}\langle\vartheta'_{X_s}\rangle
+\frac{1}{5}\langle\vartheta''_{X_s}\rangle
\label{eq: k=1 integral for two-pt function of vartheta with one position at Xs}
\end{align}
For $k=3$:
\begin{align}
\frac{11}{200}\frac{p^{(5)}}{p'_{X_s}}\langle\vartheta_{X_s}\rangle\:.
\label{eq: k=3 integral for two-pt function of vartheta with one position at Xs}
\end{align}
\end{cor}

\begin{proof}
For $n=5$, $\vartheta=\vartheta^{[1]}$ ($\vartheta^{[y]}$ is absent).
We use eq.\ (\ref{eq: formula Galois-even part of two-pt function of vartheta when n=5}) for $\langle\vartheta_x\vartheta_{X_s}\rangle^{[1]}$ 
and Claim \ref{claim: derivatives of beta}.
(Alternatively, for $k=0,1$, the proof follows from the OPE (\ref{map: OPE of vartheta}) and Example \ref{Example: some expansions worked out to higher order, for n=5}.)
For $k=3$, we also need eq.\ (\ref{eq: third order derivative of vartheta when n=5}) for $\langle\vartheta^{(3)}\rangle$. 
\end{proof}

\noi
The integral for $k=2$ is unknown and gives rise to the introduction of the auxiliary function
\begin{align}\label{defining eq. of auciliary function Bs}
\tilde{B}_s
:=\oint\frac{\langle\vartheta_x\vartheta_{X_s}\rangle}{(x-X_s)^3}\frac{dx}{2\pi i}
\:.
\end{align}
$\tilde{B}_s$ depends on $\langle\vartheta_{X_s}\vartheta''_{X_s}\rangle_r$ in the following way: 
\begin{align*}
\frac{2}{p'_{X_s}}\tilde{B}_s
\=\frac{c}{192}\left(\frac{1}{3}\frac{[p^{(3)}]^2}{p'_{X_s}}+\frac{1}{2}\frac{p''_{X_s}p^{(4)}_{X_s}}{p'_{X_s}}+\frac{1}{60}p^{(5)}_x\right)\1\nn\\
\+\left(\frac{1}{24}\frac{p^{(4)}_{X_s}}{p'_{X_s}}\langle\vartheta_{X_s}\rangle
+\frac{1}{20}\left(3\frac{p^{(3)}_{X_s}}{p'_{X_s}}
+\frac{[p''_{X_s}]^2}{[p'_{X_s}]^2}\right)\langle\vartheta'_{X_s}\rangle
+\frac{3}{40}\frac{p''_{X_s}}{p'_{X_s}}\langle\vartheta''_{X_s}\rangle
+\frac{1}{60}\langle\vartheta^{(3)}_{X_s}\rangle\right)\nn\\
\+\frac{1}{p'_{X_s}}\partial^2_x|_{x=X_s}\langle\vartheta_{X_s}\vartheta_x\rangle_r
\end{align*}

We shall also need the following integral:

\begin{claim}
We assume $n=5$ and the $(2,5)$ minimal model. 
Let $\langle\vartheta_1\vartheta_2\rangle^{[1]}$ be the Galois-even part of the Galois splitting (\ref{eq: splitting of two-point function of vartheta into Galois-even and Galois-odd part}).
The value of the integral
\begin{displaymath}
\oint
\frac{\langle\vartheta_x\vartheta'_{X_s}\rangle^{[1]}}{(x-X_s)^{k+1}}\frac{dx}{2\pi i}
\end{displaymath}
for $k=2$ is
\begin{align}\label{eq: k=2 integral for first derivative of Galois-even part of two-pt function of vartheta at Xs}
&\left(\frac{c}{9600}p''_{X_s}p^{(5)}_{X_s}
 +\frac{c}{288}p^{(3)}_{X_s}p^{(4)}_{X_s}\right)\1\nn\\
 \+\frac{11}{120}p^{(3)}_{X_s}\langle\vartheta''_{X_s}\rangle
 +\frac{1}{12}p^{(4)}_{X_s}\langle\vartheta'_{X_s}\rangle
 +\frac{1}{600}p^{(5)}\langle\vartheta_{X_s}\rangle\:.
\end{align}
\end{claim}

\begin{proof}
Direct computation, using eq.\ (\ref{eq: formula Galois-even part of two-pt function of vartheta when n=5}). 
\end{proof}

\begin{theorem}
We assume $n=5$ and the $(2,5)$ minimal model.
Let $X_s$ be a ramification point.
Set
\begin{displaymath}
\mathcal{D}_s
:=d_{X_s}-\frac{c}{8}\:\omega_s\:.
\end{displaymath}
Let $\tilde{B}_s$ be the auxiliary function $\tilde{B}_s$ given by eq.\ (\ref{defining eq. of auciliary function Bs}).
We have the following complete set of ODEs:
\begin{align*}
\mathcal{D}_s\1
\=\frac{2\xi_s}{p'_{X_s}}\langle\vartheta_{X_s}\rangle\:,\\
\mathcal{D}_s\langle\vartheta_x\rangle|_{x=X_s}
\=
\xi_s\left\{
-\frac{7c}{480}p'_{X_s}S(p_x)(X_s)\1
+\frac{9}{10}\frac{p''_{X_s}}{p'_{X_s}}\langle\vartheta_{X_s}\rangle
-\frac{7}{10}\langle\vartheta'_{X_s}\rangle\right\}\:,\\
\mathcal{D}_s\langle\vartheta'_x\rangle|_{x=X_s}
\=
\xi_s\left\{\frac{c}{480}\left(7\frac{p''_{X_s}p^{(3)}_{X_s}}{p'_{X_s}}-p^{(4)}_{X_s}\right)\1
+\frac{11}{30}\frac{p^{(3)}_{X_s}}{p'_{X_s}}\langle\vartheta_{X_s}\rangle
+\frac{7}{20}\frac{p''_{X_s}}{p'_{X_s}}\langle\vartheta'_{X_s}\rangle
-\frac{3}{10}\langle\vartheta''_{X_s}\rangle
\right\}\:,\\
\mathcal{D}_s\langle\vartheta''_x\rangle|_{x=X_s}
\=
\xi_s\:
\left\{
\frac{2}{p'_{X_s}}\tilde{B}_s
+\frac{7c}{1920}p^{(5)}\1
\right\}\:,\\
\mathcal{D}_s\tilde{B}_s
\=\xi_s\left\{\frac{c}{32000}p''_{X_s}p^{(5)}_{X_s}
+\frac{c}{960}p^{(3)}_{X_s}p^{(4)}_{X_s}\right\}\1\\
\+\xi_s\frac{1607}{24000}p^{(5)}_{X_s}\langle\vartheta_{X_s}\rangle
+\xi_s\frac{1}{40}p^{(4)}_{X_s}\langle\vartheta'_{X_s}\rangle
+\xi_s\left\{\frac{143}{2400}p^{(3)}_{X_s}+\frac{7c}{640}\frac{[p''_{X_s}]^2}{p'_{X_s}}\right\}\langle\vartheta''_{X_s}\rangle\\
\+\frac{9}{10}\frac{p''_{X_s}}{p'_{X_s}}\tilde{B}_s
\:.
\end{align*}
\end{theorem}

\begin{proof}
The ODE for $\1$ is eq.\ (\ref{ODE for 0-pt function}), which holds for any genus.
For $n=5$, $\vartheta=\vartheta^{[1]}$ ($\vartheta^{[y]}$ is absent).
For $k\geq 0$, we obtain from the differential equation in Lemma \ref{lemma: differential eq. for the N-pt function of vartheta} for $N=1$,
\begin{align*}
\frac{1}{k!}\left(d_{X_s}\langle\vartheta^{(k)}_x\rangle|_{x=X_s}-\frac{c}{8}\:\omega_s\langle\vartheta^{(k)}_{X_s}\rangle\right)
\=2
\frac{\xi_s}{p'_{X_s}}\oint_{\gamma}\frac{\langle\vartheta_{X_s}\vartheta_{x'}\rangle}{(x'-X_s)^{k+1}}\:\frac{dx'}{2\pi i}\\
\-\xi_s\oint_{\gamma}\frac{\langle\vartheta_{x'}\rangle}{(x'-X_s)^{k+2}}\:\frac{dx'}{2\pi i}\\
\-\frac{c}{16}\:\xi_s\1\oint_{\gamma}\frac{p'_{x'}}{(x'-X_s)^{k+3}}\:\frac{dx'}{2\pi i}
\:.
\end{align*}
In the following, we list the contributions without the factor of $\xi_s$.
Then for $k=0$, the first line yields (\ref{eq: k=0 integral for two-pt function of vartheta with one position at Xs}).
The second line yields
\begin{displaymath}
-\langle\vartheta'_{X_s}\rangle
\end{displaymath}
and the third
\begin{displaymath}
-\frac{c}{32}p^{(3)}_{X_s}\1\:. 
\end{displaymath}
For $k=1$, 
the first line yields (\ref{eq: k=1 integral for two-pt function of vartheta with one position at Xs}).
The second line yields
\begin{displaymath}
-\frac{1}{2}\langle\vartheta''_{X_s}\rangle
\end{displaymath}
and the third
\begin{displaymath}
-\frac{c}{96}p^{(4)}_{X_s}\1 
\end{displaymath}
For $k=2$, the first line yields $\frac{2}{p'_{X_s}}\tilde{B}_s$. 
The second line yields
\begin{displaymath}
-\frac{1}{6}\langle\vartheta^{(3)}_{X_s}\rangle
=\frac{c}{160}p^{(5)}_{X_s}\1\:,
\end{displaymath}
by eq.\ (\ref{eq: third order derivative of vartheta when n=5}),
and the third
\begin{displaymath}
-\frac{c}{384}p^{(5)}_{X_s}\1\:. 
\end{displaymath}
We address the ODE for $\tilde{B}_s$. 
Note that by the Galois splitting (\ref{eq: splitting of two-point function of vartheta into Galois-even and Galois-odd part}),
\begin{displaymath}
\langle\vartheta_1\vartheta_{X_s}\rangle
=\langle\vartheta_1\vartheta_{X_s}\rangle^{[1]}\:.
\end{displaymath}
We have
\begin{displaymath}
\tilde{B}_s
=\oint\frac{1}{(x_1-X_s)^3}\oint\frac{\langle\vartheta_1\vartheta_2\rangle^{[1]}}{x_2-X_s}\frac{dx_2}{2\pi i}\frac{dx_1}{2\pi i}
\:,
\end{displaymath}
and
\begin{displaymath}
\frac{\partial}{\partial X_s}\frac{1}{(x-X_s)^k}
=\frac{k}{(x-X_s)^{k+1}}\:,
\end{displaymath}
so
\begin{align}
\mathcal{D}_s\tilde{B}_s
\=3\xi_s\oint
\frac{\langle\vartheta_1\vartheta_{X_s}\rangle}{(x_1-X_s)^4}\frac{dx_1}{2\pi i}\nn\\
\+\xi_s\oint
\frac{1}{(x_1-X_s)^3}\oint\frac{\langle\vartheta_1\vartheta_2\rangle^{[1]}}{(x_2-X_s)^2}\frac{dx_2}{2\pi i}\frac{dx_1}{2\pi i}\label{o-o integral 2.2}\\
\+\oint
\frac{1}{(x_1-X_s)^3}\oint\frac{1}{x_2-X_s}\mathcal{D}_s\langle\vartheta_1\vartheta_2\rangle^{[1]}\frac{dx_2}{2\pi i}\frac{dx_1}{2\pi i}\:.\nn
\end{align}
Here
\begin{align*}
\left(\mathcal{D}_s\langle\vartheta_1\vartheta_2\rangle^{[1]}\right)|_{x_2=X_s}
=\left(\mathcal{D}_s\langle\vartheta_1\vartheta_2\rangle\right)|_{x_2=X_s}\:.
\end{align*}
Indeed, we have
\begin{displaymath}
y_2\sim (x_2-X_s)^{1/2}\:,\quad d_{X_s}y_2\sim(x_2-X_s)^{-1/2}\:,
\end{displaymath}
so $\mathcal{D}_s\left(y_1y_2\langle\vartheta_1\vartheta_2\rangle^{[y_1y_2]}\right)$ does not contribute to the integral 
\begin{displaymath}
\oint\frac{\mathcal{D}_s\langle\vartheta_1\vartheta_2\rangle}{x_2-X_s}\frac{dx_2}{2\pi i}\:. 
\end{displaymath}
Thus using the differential equation from Lemma \ref{lemma: differential eq. for the N-pt function of vartheta} for $N=2$,
\begin{align}
\mathcal{D}_s\tilde{B}_s
\=2\frac{\xi_s}{p'_{X_s}}\oint
\frac{1}{(x_1-X_s)^3}\oint\frac{\langle\vartheta_{X_s}\vartheta_2\vartheta_2\rangle}{x_2-X_s}
\frac{dx_2}{2\pi i}\frac{dx_1}{2\pi i}\label{o-o integral 1}\\
\+2\xi_s\oint
\frac{\langle\vartheta_1\vartheta_{X_s}\rangle}{(x_1-X_s)^4}\frac{dx_1}{2\pi i}\label{o-o integral 2}\\
\-\frac{c}{16}\xi_s
\oint\frac{p_1'}{(x_1-X_s)^5}
\oint\frac{\langle\vartheta_2\rangle}{x_2-X_s}
\:
\frac{dx_2}{2\pi i}\frac{dx_1}{2\pi i}\label{o-o integral 3.1}\\
\-\frac{c}{16}\xi_s
\oint\frac{\langle\vartheta_1\rangle}{(x_1-X_s)^3}
\oint\frac{p_2'}{(x_2-X_s)^3}
\:
\frac{dx_2}{2\pi i}\frac{dx_1}{2\pi i}\label{o-o integral 3.2}
\end{align}
(Note that line (\ref{o-o integral 2.2}) has dropped out.)
We address line (\ref{o-o integral 1}).
By the OPE,
\begin{align*}
\oint\frac{\langle\vartheta_{X_s}\vartheta_2\vartheta_1\rangle}{x_2-X_s}&\frac{dx_2}{2\pi i}\\
\=\lim_{x_2\rechts X_s}\left[\frac{c}{32}f_{2X_s}^2\langle\vartheta_1\rangle
+\frac{1}{4}f_{2X_s}(\langle\vartheta_{X_s}\vartheta_1\rangle+\langle\vartheta_2\vartheta_1\rangle)                                    
\right]_{\text{order zero in $(x_2-X_s)$}}
+\langle\psi_{X_s}\vartheta_1\rangle\:.
\end{align*}
By Example \ref{Example: some expansions worked out to higher order, for n=5} and eq.\ (\ref{eq: definition of psi}),
\begin{align*}
\frac{2}{p'_{X_s}}\oint\frac{\langle\vartheta_{X_s}\vartheta_2\vartheta_1\rangle}{x_2-X_s}\frac{dx_2}{2\pi i}
\=\frac{c}{20}\left(\frac{1}{3}p^{(3)}_{X_s}+\frac{7}{16}\frac{[p''_{X_s}]^2}{p'_{X_s}}\right)\langle\vartheta_1\rangle 
+\frac{9}{10}\frac{p''_{X_s}}{p'_{X_s}}\langle\vartheta_{X_s}\vartheta_1\rangle
+\frac{3}{10}\langle\vartheta'_{X_s}\vartheta_1\rangle\:,
\end{align*}
cf.\ eq.\ (\ref{eq: k=0 integral for two-pt function of vartheta with one position at Xs}).
It follows
\begin{align*}
\frac{2}{p'_{X_s}}\oint
\frac{1}{(x_1-X_s)^3}&\oint\frac{\langle\vartheta_{X_s}\vartheta_1\vartheta_2\rangle}{x_2-X_s}
\frac{dx_2}{2\pi i}\frac{dx_1}{2\pi i}\\
\=\frac{c}{40}\left(\frac{1}{3}p^{(3)}_{X_s}+\frac{7}{16}\frac{[p''_{X_s}]^2}{p'_{X_s}}\right)\langle\vartheta''_{X_s}\rangle
+\frac{9}{10}\frac{p''_{X_s}}{p'_{X_s}}\tilde{B}
+\frac{3}{10}\oint
\frac{\langle\vartheta'_{X_s}\vartheta_1\rangle}{(x_1-X_s)^3}\frac{dx_1}{2\pi i}
\end{align*}
where the latter integral is given by eq.\ (\ref{eq: k=2 integral for first derivative of Galois-even part of two-pt function of vartheta at Xs}).
Line (\ref{o-o integral 2}) is given by eq.\ (\ref{eq: k=3 integral for two-pt function of vartheta with one position at Xs}),
and gives
\begin{displaymath}
\frac{11}{200}\xi_sp^{(5)}_{X_s}\langle\vartheta_{X_s}\rangle\:.
\end{displaymath}
Line (\ref{o-o integral 3.1}) yields
\begin{displaymath}
-\frac{c}{384}\xi_sp^{(5)}_{X_s}\langle\vartheta_{X_s}\rangle
\:.
\end{displaymath}
Line (\ref{o-o integral 3.2}) yields
\begin{displaymath}
-\frac{c}{64}\xi_s
p^{(3)}_{X_s}\langle\vartheta''_{X_s}\rangle
\:.
\end{displaymath}
We conclude that
\begin{align*}
\mathcal{D}_s\tilde{B}_s
\=\xi_s\left\{\frac{c}{32000}p''_{X_s}p^{(5)}_{X_s}
+\frac{c}{960}p^{(3)}_{X_s}p^{(4)}_{X_s}\right\}\1\\
\+\xi_s\left(\frac{11}{200}-\frac{c}{384}+\frac{1}{2000}\right)p^{(5)}_{X_s}\langle\vartheta_{X_s}\rangle\\
\+\xi_s\frac{1}{40}p^{(4)}_{X_s}\langle\vartheta'_{X_s}\rangle\\
\+\xi_s\left\{\left(\frac{c}{120}-\frac{c}{64}+\frac{11}{400}\right)p^{(3)}_{X_s}+\frac{7c}{640}\frac{[p''_{X_s}]^2}{p'_{X_s}}\right\}\langle\vartheta''_{X_s}\rangle\\
\+\frac{9}{10}\frac{p''_{X_s}}{p'_{X_s}}\tilde{B}_s
\:.
\end{align*}
\end{proof}

\section{Comparison with the approach using transcendental methods}

We discuss the connection with the work by Mason \& Tuite \cite{M-T:2006}.

\subsection{The differential equation for the characters of the $(2,5)$ minimal model}\label{Section: differential equations for the (2,5) minimal model}

The character $\1$ of any CFT on the torus $\Sigma_1$ solves the ODE \cite{EO:1987}
\begin{displaymath}
\frac{d}{d\tau}\1
=\frac{1}{2\pi i}\oint\langle T(z)\rangle\:dz
=\frac{1}{2\pi i}\T\:,
\end{displaymath}
where the contour integral is along the real period, and $\oint dz=1$.
It is a particular feature of $g=1$ that $\T$ is constant in position.
$\T$ defines modular form of weight two in the modulus.
In the $(2,5)$ minimal model,
we find 
\begin{displaymath}
2\pi i\frac{d}{d\tau}\T
=\oint\langle T(w)T(z)\rangle\:dz
=-4\T G_2+\frac{22}{5}G_4\1\:.
\end{displaymath}
In terms of the \textsl{Serre derivative} 
\begin{align}\label{Serre differential operator}
\mathfrak{D}_{\ell}
:\=\frac{1}{2\pi\i}\frac{d}{d\tau}-\frac{\ell}{12}E_2(\tau)
\:
\end{align}
(for weight $\ell$),
the two first order ODEs combine to give the second order ODE \cite{MMS:1988,KNS:2012}
\begin{align*}
\mathfrak{D}_2\circ\mathfrak{D}_0\1
=\frac{11}{3600}\:E_4\1\:. 
\end{align*}
The two solutions are the famous Rogers-Ramanujan partition functions \cite{DiFranc:1997}
\begin{align}
\1_1
\=q^{\frac{11}{60}}\sum_{n\geq 0}\frac{q^{n^2+n}}{(q)_n}
=q^{\frac{11}{60}}\left(1+q^2+q^3+q^4+q^5+2q^6+\ldots\right)\:,
\label{eq: 11/60 Rogers-Ramanujan partition function}\\
\1_2
\=q^{-\frac{1}{60}}\sum_{n\geq 0}\frac{q^{n^2}}{(q)_n}
=q^{-\frac{1}{60}}\left(1+q+q^2+q^3+2q^4+2q^5+3q^6+\ldots\right)
\label{eq: -1/60 Rogers-Ramanujan partition function}\:.  
\end{align}
($q=e^{2\pi\i\tau}$) named after the Rogers-Ramanujan identities. The first is given by
\begin{displaymath}
q^{-\frac{11}{60}}\1_1=\prod_{n=\pm 2\:\text{mod}\:5}(1-q^n)^{-1}
\:
\end{displaymath}
and provides the generating function for the partition
which to a given holomorphic dimension $h\geq 0$ attributes the number of linearly independent holomorphic fields 
present in the $(2,5)$ minimal model.
There is a corresponding Rogers-Ramanujan identity for $q^{\frac{1}{60}}\1_2$ with a similar combinatorical interpretation,
but which involves non-holomorphic fields.

\subsection{Introduction of the transcendental coordinates}

Let $\omega=\omega_1,\:\omega'=\omega_3\in\C$ with $\IM(\omega/\omega')>0$ 
be the two elementary half periods so that $\omega_2=\omega_1+\omega_3$ is the midpoint of the fundamental cell.
The half periods are the points $z$ with $0=\partial_z\wp(z|\tau)=:\wp'(z|\tau)$.
At these points, the Weierstrass $\wp$-function is invariant under point reflection.

In the finite region, a genus one surface is defined by $y^2=p_3(x)$ where $p_3(x)$ is a order three polynomial of $x=\wp(z|\tau)$, and $y=\wp'(z|\tau)$. 
Thus the half periods are the ramification points of the $g=1$ surface in the finite region.
At these points, $x=\wp(z|\tau)$ is invariant under point reflection.
This leads us to considering the fundamental cell of the torus modulo point reflection at any fixed half period point.
The half periods are all equivalent with that regard, as they differ by full periods only. 
Considering the fundamental cell modulo point reflection at the chosen half period cuts the cell in two halves. 
The edge between these two halves is itself cut into two 
and the two pieces are identified through the reflection at its midpoint.   

When we perform a linear fractional transformation, close to a ramification point, 
the lift to the double cover has two possible values, one on each sheet. 
We map either of the two points to a corresponding pair of points on the double cover of the other $\CP^1$, one on each sheet.
The ambiguity of the lift disappears as we project down to the second $\CP^1$. 
The composition of these maps gives a well-defined map $\CP^1\rechts\CP^1$.
By the Riemann Theorem, all $\CP^1$s are isomorphic, so the map is an automorphism of $\CP^1$, 
thus a linear fractional transformations $x\mapsto\frac{ax+b}{cx+d}$.
By fixing the points $0$ and $\infty$, we are left with a scaling factor of $x$ as the only degree of freedom.

Let $z,\hat{z}$ be the coordinates on the two fundamental cells modulo point reflection.
We cut away a circle about $z=0$ and $\hat{z}=0$ and require 
\begin{align}\label{eq: z1z2=eps}
z\hat{z}
=\eps 
\end{align}
to identify some small annulus centered at $z=0$ and at $\hat{z}=0$, respectively.
The copy of $\CP^1$ covered by the torus defined by the modulus $\tau$ respectively $\hat{\tau}$ comes with the natural coordinate
\begin{align}\label{def: coordinates on the two CP1 corresponding to the two tori}
\xi=\wp=\wp(z|\tau)\:,\quad 
\hat{\xi}=\hat{\wp}=\wp(\hat{z}|\hat{\tau})\:, 
\end{align}
respectively. 
By the expansion of $\wp(z|\tau)$ about $z=0$ and by (\ref{eq: z1z2=eps}), 
we have on the annulus
\begin{align}\label{approx eq: eps condition for auxiliary coordinates}
\xi\hat{\xi}\sim\frac{1}{\eps^2} 
\end{align}
so $\eps\wp_1\sim\frac{1}{\eps\wp_2}$, but these are not exact equations.
We are glueing here annuli centered at $\infty$ and zero, respectively, on either $\CP^1$; 
the respective center point is excluded from the annulus.
The result is topologically a $\CP^1$, and it is covered by a $g=2$ surface. 

\subsection{Pair of almost global coordinates}

The new $\CP^1$ comes with a pair $(X,\hat{X})$ of coordinates satisfying the following properties:
\begin{enumerate}
 \item $X$ is defined on $\CP^1$ except for the point ($\infty$) where $\hat{\xi}=0$,
 and $\hat{X}$ is defined on $\CP^1$ except for the point (zero) where $\xi=0$.
 \item We have $X\approx\xi$ where $\xi$ is defined, and $\hat{X}\approx\hat{\xi}$ where $\hat{\xi}$ is defined.
 On the annulus on which the formerly separate two copies of $\CP^1$ overlap (and nowhere else), both approximate equations hold simultaneously. 
 \item The pair $X,\hat{X}$ satifies the exact identity
 \begin{align}\label{eq: eps condition}
 X\hat{X}
=\frac{1}{\eps^2}
\:,
\end{align}
on all of $\CP^1$. 
\end{enumerate}

We shall construct these almost global coordinates.
On the annulus, by the approximate eq.\ (\ref{approx eq: eps condition for auxiliary coordinates}), 
\begin{displaymath}
\log\xi+\log\hat{\xi}
=f(\xi)\:. 
\end{displaymath}
To this corresponds to the transition rule on the annulus
\begin{align}\label{eq: transition between xi and xi'}
\hat{\xi}=\xi^{-1}e^{f(\xi)}\:. 
\end{align}
More specifically, we have by eq.\ (\ref{eq: z1z2=eps}),
\begin{displaymath}
\hat{\xi}
=\wp\left(\frac{\eps}{\wp^{-1}(\xi|\tau)}\Big|\hat{\tau}\right)
\:.
\end{displaymath}
Now the argument goes as follows:
$f=\log\xi\hat{\xi}=\log\wp\hat{\wp}$ is nearly constant on the annulus 
by the fact that $\wp\sim\frac{1}{z^2}$ and by eq.\ (\ref{eq: z1z2=eps}).
The corrections are small for small $\eps$.
Thus $f$ has a Laurent series expansion part of which can be analytically continued to small $\xi$,
and the other part to small $\hat{\xi}$, i.e.\ to the outside of the annulus (using holomorphicity of $f$ in $\eps$). 
For $X_0$ inside the annulus,
\begin{displaymath}
f(X_0)
=\oint_{\text{outer}}\frac{f(X)}{X-X_0}dX
-\oint_{\text{inner}}\frac{f(\hat{X})}{\hat{X}-X_0}d\hat{X}
\:.
\end{displaymath}
Here by outer resp.\ inner contour we mean the circle bounding the annulus in the $\hat{\tau}$ and the $\tau$ part, respectively.
The integral over the outer contour can be extended to the $\tau$ part, giving rise to a holomorphic function $A$,
while the integral over the inner contour can be extended to the $\hat{\tau}$ part, giving rise to a holomorphic function $\hat{A}$,
\begin{displaymath}
f
=\log\wp\hat{\wp}
=A+\hat{A}
\end{displaymath}
It follows that
\begin{displaymath}
e^f
=\xi\hat{\xi}
=e^Ae^{\hat{A}}\:,
\end{displaymath}
or
\begin{align*}
\frac{\xi}{e^A}\frac{\hat{\xi}}{e^{\hat{A}}}=1\:. 
\end{align*}
This is the general argument, and we perform the computation for $X,\hat{X}$ explicitely as an expansion in $\eps$.

\begin{claim}\label{claim: the local coordinates on the single CP1}
Let $\xi_1,\xi_2$ be given by eqs (\ref{def: coordinates on the two CP1 corresponding to the two tori}).
$\CP^1$ admits a pair of global coordinates $X=X(\xi,\hat{\xi})$, $\hat{X}=\hat{X}(\xi,\hat{\xi})$ which satisfies eq.\ (\ref{eq: eps condition}). 
In the notations 
\begin{align*}
\hat{z}^2\hat{\wp}
\=1+\sum_{m=1}^{\infty}a_m\hat{z}^{2m+2}\:,\\
z^2\wp
\=1+\sum_{m=1}^{\infty}\ta_mz^{2m+2}\:,
\end{align*}
these coordinates are given up to terms of order $\eps^6$, by
\begin{align*}
X
\=\wp
\:\left(1+a_1\eps^4\left(\wp^2-2\ta_1\right)+a_2\eps^6\left(\wp^3-5\ta_1\wp-3\ta_2\right)+O(\eps^8)\right)^{-1}
\:,\\
\hat{X}
\=\hat{\wp}
\:\left(1+\ta_1\eps^4\left(\hat{\wp}^2-2a_1\right)+\ta_2\eps^6\left(\hat{\wp}^3-5a_1\hat{\wp}-3a_2\right)+O(\eps^8)\right)^{-1}
\:.
\end{align*}
\end{claim}

\begin{proof}
With the notations introduced above,
we define
\begin{align*}
\log X
:\=\log\xi-\sum_{n=1}^{\infty}A_{n}\xi^{n}
\:,\\
\log\hat{X}
:\=\log\hat{\xi}-\sum_{n=1}^{\infty}B_{n}\hat{\xi}^{n}\:.
\end{align*}
It follows that 
\begin{displaymath}
X=\frac{\xi}{e^{\sum_{n=1}^{\infty}A_{n}\xi^{n}}}\:,
\quad 
\hat{X}
=\frac{\hat{\xi}}{e^{\sum_{n=1}^{\infty}B_{n}\hat{\xi}^{n}}}\:,
\end{displaymath}
and $\log X+\log\hat{X}=-2\log\eps$, or
\begin{displaymath}
X\hat{X}=\frac{1}{\eps^2}\:.
\end{displaymath}
Here the coefficients $A_n,B_n$ are determined by the expansion 
\begin{align}\label{eq: expansion of log z squared xi + log hat z squared hat xi}
\log(z^2\xi)+\log(\hat{z}^2\hat{\xi})
=\sum_{n=1}^{\infty}A_n\xi^n
+\sum_{n=1}^{\infty}B_n{\hat{\xi}}^n\:
\end{align}
on the annulus, and depend both on $\tau,\hat{\tau}$ and $\eps$.
The series converge for small enough $\eps$.
Details of the computation are left to the reader.
\end{proof}

The closed form of the denominator of $X$ and $\hat{X}$, respectively, defines coefficient matrices 
which satisfy a system of equations equivalent to that in \cite{M-T:2006}.   

\subsection{Ramification points using transcendental methods}

In the conventions of \cite{M-T:2006}, the $g=1$ fundamental cell is spanned by $2\omega=2\pi\i$ and $2\omega'=2\pi i\tau$,
(with $\IM(1/\tau)=\IM(\bar{\tau})>0$).
The Eisenstein series is
\begin{displaymath}
E^{\text{MT}}_{2,\tau}
=-\frac{1}{12}E_{2,\tau}
=-\frac{1}{12}+2q+\ldots\:.
\end{displaymath}
The half-periods $\omega_1,\omega_2,\omega_3$
are $\omega$, $\omega'$ and $\omega+\omega'$ in some order.
Let  \cite[p.\ 633]{A-S}
\begin{displaymath}
\wp(\omega_k|\tau)
=\xi_{k-1}\:,\quad(k=1,2,3)\:.
\end{displaymath}
We have
\begin{displaymath}
[\wp'(z)]^2
=p_3(\wp)
=4\prod_{k=0}^2(\wp(z)-\xi_k)\:.
\end{displaymath}
The specific cubic polynomial is given by 
\begin{align*}
[\wp']^2
=4(\wp^3-30G_4\wp-70G_6)
\end{align*}
and implies that
\begin{align}\label{eq: sum of half-periods is zero}
\xi_1+\xi_2+\xi_3
=0\:.
\end{align}
Another natural definition is 
\begin{align}\label{def: e k}
e_{k}
=-2\frac{\D\vartheta_k}{\vartheta_k}\:,\quad(k=2,3,4)
\end{align}
where $\D$ is the Serre differential operator defined by eq.\ (\ref{Serre differential operator})
(the theta functions have weight $1/2$).
In the normalisation of Mason and Tuite ($\omega=i\pi$), 
we have for either torus \cite[p.\ 650]{A-S} 
\begin{align*}
e_4
\=\frac{1}{12}(\vartheta_2^4+\vartheta_3^4)
=\xi_1\\
e_3
\=\frac{1}{12}(-\vartheta_2^4+\vartheta_4^4)
=\xi_0\\
e_2
\=\frac{1}{12}(-\vartheta_3^4-\vartheta_4^4)
=\xi_2\:.
\end{align*}
Note that by the Jacobi identity (\ref{Jacobi id}), 
\begin{align*}
\xi_0-\xi_2
\=\frac{1}{4}\vartheta_4^4\\
\xi_1-\xi_0
\=\frac{1}{4}\vartheta_2^4\\
\xi_1-\xi_2
\=\frac{1}{4}\vartheta_3^4
\end{align*}
Let the second torus have modulus $\hat{\tau}$ and ramification points $\hat{\xi}_k$.
Then the corresponding equations hold for $\hat{\xi}_k$ in terms of the theta functions in $\hat{\tau}$.
The ramification points for the $g=2$ surface obtained by sewing are $\xi_0,\xi_1,\xi_2$ and, for $k=0,1,2$,
\begin{align}\label{eq: eps condition for ramification points}
\xi_{k+3}
=\frac{1}{\eps^2\hat{\xi}_k}
\:. 
\end{align}

\begin{claim}\label{claim: effect of the lnear fractiona trsf which maps the xi k to 0,1,infty, respectively, on X k, up to order eps to the power of 6}
Let $X_k$ be the point corresponding to $\xi_k$ by means of Claim \ref{claim: the local coordinates on the single CP1}.
The linear rational transformation mapping $X_0,X_1,X_2$ to $0,1,\infty$
differs from that mapping $\xi_0,\xi_1,\xi_2$ to $0,1,\infty$ only to order at least $\eps^6$.
Thus it maps 
\begin{displaymath}
X_{k+3}=\frac{1}{\eps^2\hat{X}_k} 
\end{displaymath}
to 
\begin{align*}
f\left(\frac{1}{\eps^2\hat{X}}\right)
\=\frac{\vartheta_3^4}{\vartheta_2^4}
\Big(1
-\frac{\vartheta_4^4}{4}\eps^2\hat{X}_k
-\frac{\vartheta_4^4}{4}\xi_2\eps^4\hat{X}_k^2
+O(\eps^6)\Big)\:.\label{hat xi after linear fractional trsf}
\end{align*}
\end{claim}

\begin{proof}
Cf.\ Appendix \ref{proof: claim: effect of the lnear fractiona trsf which maps the xi k to 0,1,infty, respectively, on X k, up to order eps to the power of 6}.
\end{proof}

\subsection{Ramification points using algebraic methods, for $g=2$}

We set
\begin{displaymath}
e\{x\}
=\exp(2\pi\i x)
\:.
\end{displaymath}
Following \cite{M:1984-2}, we define
\begin{align*}
\theta
\left[\begin{tabular}{c}
$\vec{a}$\\
$\vec{b}$
\end{tabular}\right](\vec{z},\Omega)
\=\sum_{\vec{n}\in\Z^g}e\{\frac{1}{2}(\vec{n}+\vec{a})^t\Omega(\vec{n}+\vec{a})+(\vec{n}+\vec{a})^t(\vec{z}+\vec{b})\}
\:,\quad\forall\:\vec{a},\vec{b}\in\Q^g\:.
\end{align*}
also called the first order theta function with characteristic $\left[\begin{tabular}{c}
$\vec{a}$\\
$\vec{b}$
\end{tabular}\right]$ 
for $\vec{a},\vec{b}\in\Q^g$.
We assume $g=2$ and period matrix
\begin{displaymath}
\Omega
=
\begin{pmatrix}
\Omega_{11}&\nu\\
\nu&\Omega_{22}
\end{pmatrix}\:,\quad\IM(\Omega_{jj})\:,\IM(\nu)>0\:.
\end{displaymath}
In \cite{M-T:2006}, $\Omega_{12}=\Omega_{21}=\nu=O(\eps)$. 
We adopt the convention 
\begin{displaymath}
\lim_{\nu\rechts 0}\Omega_{jj}
=\tau_j\:, 
\end{displaymath}
where $\tau_1=\tau$ and $\tau_2=\hat{\tau}$.
To leading order $\Omega_{jj}$ and $\tau_j$ are the same and their difference lies in $O(\nu^2)$. 
For terms of order $\nu^2$ and higher, greater care must be taken.

In what follows we take $\vec{z}=\vec{0}$, 
and if $\vec{a}=(a_1,a_2)^t$ and $\vec{b}=(b_1,b_2)^t$, we write
\begin{displaymath}
\theta
\left[\begin{tabular}{c}
$\vec{a}$\\
$\vec{b}$
\end{tabular}\right](0,\Omega)
=\theta
\left[\begin{tabular}{c}
$a_1,a_2$\\
$b_1,b_2$
\end{tabular}\right](\Omega)
\:.
\end{displaymath}
We set 
\begin{align*}
\varrho_j
\=e^{2\pi\i\Omega_{jj}}\:, \\
\lambda\:
\=e^{2\pi\i\nu}
=e^{\tnu}
=1+\tnu+\frac{1}{2}\tnu^2+\frac{1}{3!}\tnu^3+\ldots\hspace{1.5cm}(\tnu=2\pi\i\nu)
\end{align*}
So
\begin{align*}
\theta
\left[\begin{tabular}{c}
$a_1,a_2$\\
$b_1,b_2$
\end{tabular}\right](\Omega)
\=\sum_{\vec{n}\in\Z^2}
e\{\frac{1}{2}\Omega_{11}(n_1+a_1)^2+\nu(n_1+a_1)(n_2+a_2)+\frac{1}{2}\Omega_{22}(n_2+a_2)^2\}
\;e\{(\vec{n}+\vec{a})^t\vec{b}\}\\
\=\sum_{\vec{n}\in\Z^2}
\varrho_1^{\frac{1}{2}(n_1+a_1)^2}
\varrho_2^{\frac{1}{2}(n_2+a_2)^2}
\lambda^{(n_1+a_1)(n_2+a_2)}
e^{2\pi i\{(n_1+a_1)b_1+(n_2+a_2)b_2\}}
\:.
\end{align*}
In the following, we assume 
\begin{displaymath}
\vec{a}\cdot\vec{b}=0\:.
\end{displaymath}
Thus
\begin{align*}
\theta
\left[\begin{tabular}{c}
$a_1,a_2$\\
$b_1,b_2$
\end{tabular}\right](\Omega)
\=
\sum_{n_1\in\Z}
\sum_{n_2\in\Z}
\sum_{k=0}^{\infty}
\frac{\tnu^k}{k!}(n_1+a_1)^k(n_2+a_2)^k
e^{2\pi in_1b_1}\varrho_1^{\frac{1}{2}(n_1+a_1)^2}\:
e^{2\pi in_2b_2}\varrho_2^{\frac{1}{2}(n_2+a_2)^2}
\:.
\end{align*}
Observe that when $a_i=0$ for at least one $i\in\{1,2\}$ then all summands to odd $k$ drop out. 
Consider e.g.
\begin{displaymath}
\theta
\left[\begin{tabular}{c}
$0,a_2$\\
$b_1,b_2$
\end{tabular}\right](\Omega)
=
\sum_{n_1\in\Z}
\sum_{k=0}^{\infty}
\frac{\tnu^k}{k!}\:
e^{2\pi in_1b_1}n_1^k
\varrho_1^{\frac{1}{2}n_1^2}\:
\sum_{n_2\in\Z}
(n_2+a_2)^k
e^{2\pi in_2b_2}\varrho_2^{\frac{1}{2}(n_2+a_2)^2}
\:.
\end{displaymath}
Since
\begin{displaymath}
(\pi i)^{k}(n_j+a_j)^{2k}
\varrho_j^{\frac{1}{2}(n_j+a_1)^2}
=\frac{d^{k}}{d\Omega_{11}^{k}}\varrho_j^{\frac{1}{2}(n_j+a_1)^2}
\:,
\end{displaymath}
we find (using the definition of $\tnu$)
\begin{align*}
\theta
\left[\begin{tabular}{c}
$0,a_2$\\
$b_1,b_2$
\end{tabular}\right](\Omega)
\=
\sum_{n_1\in\Z}
\sum_{k=0}^{\infty}
\frac{\tnu^{2k}}{(2k)!}e^{2\pi in_1b_1}n_1^{2k}
\varrho_1^{\frac{1}{2}n_1^2}\:
\sum_{n_2\in\Z}
(n_2+a_2)^{2k}
e^{2\pi in_2b_2}\varrho_2^{\frac{1}{2}(n_2+a_2)^2}\\
\=\sum_{k=0}^{\infty}
\frac{(2\nu)^{2k}}{(2k)!}
\sum_{n_1\in\Z}
\sum_{n_2\in\Z}
\frac{d^{k}}{d\Omega_{11}^{k}}\left(e^{2\pi in_1b_1}\varrho_1^{\frac{1}{2}n_1^2}\right)\:
\frac{d^{k}}{d\Omega_{22}^{k}}\left(e^{2\pi in_2b_2}\varrho_2^{\frac{1}{2}(n_2+a_2)^2}\right)
\:.
\end{align*}
Writing $\vartheta_{k,\Omega_{jj}}=\vartheta_k(0,\varrho_j)$, we obtain
\begin{align*}
\theta
\left[\begin{tabular}{c}
$\frac{1}{2}$\\
$0$
\end{tabular}\right](\Omega_{jj})
\=\sum_{n_j\in\Z}\varrho_j^{\frac{1}{2}(n_j+\frac{1}{2})^2}
=2\varrho^{\frac{1}{8}}\sum_{n_j=0}^{\infty}\varrho_j^{\frac{1}{2}n_j(n_j+1)}
=\vartheta_{2,\Omega_{jj}}\\
\theta
\left[\begin{tabular}{c}
$0$\\
$0$
\end{tabular}\right](\Omega_{jj})
\=\sum_{n_j\in\Z}\varrho_j^{\frac{1}{2}n_j^2}
=1+2\sum_{n_j=1}^{\infty}\varrho_j^{\frac{1}{2}n_j^2}
=\vartheta_{3,\Omega_{jj}}\\
\theta
\left[\begin{tabular}{c}
$0$\\
$\frac{1}{2}$
\end{tabular}\right](\Omega_{jj})
\=\sum_{n_j\in\Z} 
(-1)^{n_j}\varrho_2^{\frac{1}{2}n_j^2}
=1+2\sum_{n_j=1}^{\infty}(-1)^{n_j}q^{\frac{1}{2}n_j^2}
=\vartheta_{4,\Omega_{jj}}\:.
\end{align*}
Moreover,
\begin{align*}
\Theta_{3,3}
:=
\theta
\left[\begin{tabular}{c}
$0,0$\\
$0,0$
\end{tabular}\right](\Omega)
\=\sum_{\vec{n}\in\Z^2} 
e\{\frac{1}{2}\Omega_{11}n_1^2+\nu n_1n_2+\frac{1}{2}\Omega_{22}n_2^2\}\\
\=\sum_{\vec{n}\in\Z^2}
\varrho_1^{\frac{1}{2}n_1^2}\:\varrho_2^{\frac{1}{2}n_2^2}\:\lambda^{n_1n_2}\\
\=\vartheta_{3,\Omega_{11}}\:\vartheta_{3,\Omega_{22}}\times\\
&\times\left(1
+\frac{(2\nu)^2}{2!}\:\frac{\vartheta'_{3,\Omega_{11}}\:\vartheta'_{3,\Omega_{22}}}{\vartheta_{3,\Omega_{11}}\:\vartheta_{3,\Omega_{22}}}
+\frac{(2\nu)^4}{4!}\:\frac{\vartheta''_{3,\Omega_{11}}\:\vartheta''_{3,\Omega_{22}}}{\vartheta_{3,\Omega_{11}}\:\vartheta_{3,\Omega_{22}}}
+\frac{(2\nu)^6}{6!}\:\frac{\vartheta^{(3)}_{3,\Omega_{11}}\:\vartheta^{(3)}_{3,\Omega_{22}}}{\vartheta_{3,\Omega_{11}}\:\vartheta_{3,\Omega_{22}}}
+O(\nu^8)\right)\\
\=\theta
\left[\begin{tabular}{c}
$0$\\
$0$
\end{tabular}\right](\Omega_{11})\:
\theta
\left[\begin{tabular}{c}
$0$\\
$0$
\end{tabular}\right](\Omega_{22})
(1+O(\nu^2)\:,
\end{align*}
\begin{align*}
\Theta_{2,3}
:=
\theta
\left[\begin{tabular}{c}
$\frac{1}{2},0$\\
$0,0$
\end{tabular}\right](\Omega)
\=\sum_{\vec{n}\in\Z^2} 
e\{\frac{1}{2}\Omega_{11}(n_1+\frac{1}{2})^2+\nu(n_1+\frac{1}{2})n_2+\frac{1}{2}\Omega_{22}n_2^2\}\\
\=\sum_{\vec{n}\in\Z^2} 
\varrho_1^{\frac{1}{2}(n_1+\frac{1}{2})^2}\:\varrho_2^{\frac{1}{2}n_2^2}\:\lambda^{(n_1+\frac{1}{2})n_2}\\
\=\vartheta_{2,\Omega_{11}}\:\vartheta_{3,\Omega_{22}}\times\\
&\times\left(1
+\frac{(2\nu)^2}{2!}\:\frac{\vartheta'_{2,\Omega_{11}}\:\vartheta'_{3,\Omega_{22}}}{\vartheta_{2,\Omega_{11}}\:\vartheta_{3,\Omega_{22}}}
+\frac{(2\nu)^4}{4!}\:\frac{\vartheta''_{2,\Omega_{11}}\:\vartheta''_{3,\Omega_{22}}}{\vartheta_{2,\Omega_{11}}\:\vartheta_{3,\Omega_{22}}}
+\frac{(2\nu)^6}{6!}\:\frac{\vartheta^{(3)}_{2,\Omega_{11}}\:\vartheta^{(3)}_{3,\Omega_{22}}}{\vartheta_{2,\Omega_{11}}\:\vartheta_{3,\Omega_{22}}}
+O(\nu^8)\right)
\end{align*}

\begin{align*}
\Theta_{3,2}
:=
\theta
\left[\begin{tabular}{c}
$0,\frac{1}{2}$\\
$0,0$
\end{tabular}\right](\Omega)
\=\sum_{\vec{n}\in\Z^2} 
e\{\frac{1}{2}\Omega_{11}n_1^2+\nu n_1(n_2+\frac{1}{2})+\frac{1}{2}\Omega_{22}(n_2+\frac{1}{2})^2\}\\
\=\sum_{\vec{n}\in\Z^2} 
\varrho_1^{\frac{1}{2}n_1^2}\:\varrho_2^{\frac{1}{2}(n_2+\frac{1}{2})^2}\:\lambda^{n_1(n_2+\frac{1}{2})}\\
\=\vartheta_{3,\Omega_{11}}\:\vartheta_{2,\Omega_{22}}\times\\
&\times\left(1
+\frac{(2\nu)^2}{2!}\:\frac{\vartheta'_{3,\Omega_{11}}\:\vartheta'_{2,\Omega_{22}}}{\vartheta_{3,\Omega_{11}}\:\vartheta_{2,\Omega_{22}}}
+\frac{(2\nu)^4}{4!}\:\frac{\vartheta''_{3,\Omega_{11}}\:\vartheta''_{2,\Omega_{22}}}{\vartheta_{3,\Omega_{11}}\:\vartheta_{2,\Omega_{22}}}
+\frac{(2\nu)^6}{6!}\:\frac{\vartheta^{(3)}_{3,\Omega_{11}}\:\vartheta^{(3)}_{2,\Omega_{22}}}{\vartheta_{3,\Omega_{11}}\:\vartheta_{2,\Omega_{22}}}
+\ldots\right)
\end{align*}
\begin{align*}
\Theta_{2,4}
:=
\theta
\left[\begin{tabular}{c}
$\frac{1}{2},0$\\
$0,\frac{1}{2}$
\end{tabular}\right](\Omega)
\=\sum_{\vec{n}\in\Z^2} 
e\{\frac{1}{2}\Omega_{11}(n_1+\frac{1}{2})^2+\nu(n_1+\frac{1}{2})n_2+\frac{1}{2}\Omega_{22}n_2^2\}
\;e\{\frac{n_2}{2}\}\\
\=\sum_{n_1\in\Z}\sum_{n_2\in\Z} 
\varrho_1^{\frac{1}{2}(n_1+\frac{1}{2})^2}\:e^{\pi in_2}\varrho_2^{\frac{1}{2}n_2^2}\:\lambda^{(n_1+\frac{1}{2})n_2}\\
\=\vartheta_{2,\Omega_{11}}\:\vartheta_{4,\Omega_{22}}\times\\
&\times\left(1
+\frac{(2\nu)^2}{2!}\:\vartheta'_{2,\Omega_{11}}\:\vartheta'_{4,\Omega_{22}}
+\frac{(2\nu)^4}{4!}\:\vartheta''_{2,\Omega_{11}}\:\vartheta''_{4,\Omega_{22}}
+\frac{(2\nu)^6}{6!}\:\vartheta^{(3)}_{2,\Omega_{11}}\:\vartheta^{(3)}_{4,\Omega_{22}}+\ldots\right)
\:
\end{align*}
\begin{align*}
\Theta_{3,4}
:=
\theta
\left[\begin{tabular}{c}
$0,0$\\
$0,\frac{1}{2}$
\end{tabular}\right](\Omega)
\=\sum_{\vec{n}\in\Z^2} 
e\{\frac{1}{2}\Omega_{11}n_1^2+\nu n_1n_2+\frac{1}{2}\Omega_{22}n_2^2\}
\;e\{\frac{n_2}{2}\}\\
\=\sum_{\vec{n}\in\Z^2} 
\varrho_1^{\frac{1}{2}n_1^2}\:e^{\pi in_2}\varrho_2^{\frac{1}{2}n_2^2}\:\lambda^{n_1n_2}\\
\=\vartheta_{3,\Omega_{11}}\:\vartheta_{4,\Omega_{22}}\times\\
&\times\left(1
+\frac{(2\nu)^2}{2!}\:\frac{\vartheta'_{3,\Omega_{11}}\:\vartheta'_{4,\Omega_{22}}}{\vartheta_{3,\Omega_{11}}\:\vartheta_{4,\Omega_{22}}}
+\frac{(2\nu)^4}{4!}\:\frac{\vartheta''_{3,\Omega_{11}}\:\vartheta''_{4,\Omega_{22}}}{\vartheta_{3,\Omega_{11}}\:\vartheta_{4,\Omega_{22}}}
+\frac{(2\nu)^6}{6!}\:\frac{\vartheta^{(3)}_{3,\Omega_{11}}\:\vartheta^{(3)}_{4,\Omega_{22}}}{\vartheta_{3,\Omega_{11}}\:\vartheta_{4,\Omega_{22}}}
+\ldots\right)
\:
\end{align*}
The following does not fit into this scheme, but a similar argument applies:
Here we need $a_i=\frac{1}{2}$ for at least one $i\in\{1,2\}$. 
For example,
\begin{align*}
\theta
\left[\begin{tabular}{c}
$\frac{1}{2},a_2$\\
$0,b_2$
\end{tabular}\right](\Omega)
\=
\sum_{n_1\in\Z}
\sum_{k=0}^{\infty}
\frac{\tnu^k}{k!}(n_1+\frac{1}{2})^k(n_2+a_2)^k
\varrho_1^{\frac{1}{2}(n_1+\frac{1}{2})^2}\:
\sum_{n_2\in\Z}
e^{2\pi in_2b_2}\varrho_2^{\frac{1}{2}(n_2+a_2)^2}
\:.
\end{align*}
Since for $n\in\Z$,
\begin{displaymath}
((n-1)+\frac{1}{2})^k+(-n+\frac{1}{2})^k
=((n-1)+\frac{1}{2})^k+(-1)^k(n-\frac{1}{2})^k
\end{displaymath}
vanishes for $k$ odd, we restrict again the summation to even $k$. We conclude that
\begin{align*}
\Theta_{2,2}
:=
\theta
\left[\begin{tabular}{c}
$\frac{1}{2},\frac{1}{2}$\\
$0,0$
\end{tabular}\right](\Omega)
\=\sum_{\vec{n}\in\Z^2} 
e\{\frac{1}{2}\Omega_{11}(n_1+\frac{1}{2})^2+\nu(n_1+\frac{1}{2})(n_2+\frac{1}{2})+\frac{1}{2}\Omega_{22}(n_2+\frac{1}{2})^2\}\\
\=\sum_{\vec{n}\in\Z^2} 
\varrho_1^{\frac{1}{2}(n_1+\frac{1}{2})^2}\:\varrho_2^{\frac{1}{2}(n_2+\frac{1}{2})^2}\:\lambda^{(n_1+\frac{1}{2})(n_2+\frac{1}{2})}\\
\=
\sum_{n_1\in\Z}
\sum_{n_2\in\Z}
\sum_{k=0}^{\infty}
\frac{\tnu^k}{k!}(n_1+\frac{1}{2})^k(n_2+\frac{1}{2})^k
\varrho_1^{\frac{1}{2}(n_1+\frac{1}{2})^2}\:
\varrho_2^{\frac{1}{2}(n_2+\frac{1}{2})^2}\\
\=\vartheta_{2,\Omega_{11}}\:\vartheta_{2,\Omega_{22}}\times\\
&\times\left(1
+\frac{(2\nu)^2}{2!}\:\frac{\vartheta'_{2,\Omega_{11}}\:\vartheta'_{2,\Omega_{22}}}{\vartheta_{2,\Omega_{11}}\:\vartheta_{2,\Omega_{22}}}
+\frac{(2\nu)^4}{4!}\:\frac{\vartheta''_{2,\Omega_{11}}\:\vartheta''_{2,\Omega_{22}}}{\vartheta_{2,\Omega_{11}}\:\vartheta_{2,\Omega_{22}}}
+\frac{(2\nu)^6}{6!}\:\frac{\vartheta^{(3)}_{2,\Omega_{11}}\:\vartheta^{(3)}_{2,\Omega_{22}}}{\vartheta_{2,\Omega_{11}}\:\vartheta_{2,\Omega_{22}}}
+O(\nu^8)\right)
\end{align*}
(The notation $\Theta_{i,j}$ is non-standard.)

In the conventions of \cite[see references therein]{H:2009}, for $g=2$, the ramification points are 
\begin{displaymath}
X_0=0=b_5\:,\quad
X_1=1=b_6\:,\quad
X_2=b_4\:. 
\end{displaymath}
Then 
\begin{align*}
b_1
=X_3
\=\frac{\Theta_{3,3}^2\Theta_{3,2}^2}{\Theta_{2,3}^2\Theta_{2,2}^2}
=X^{3,3,3,2}_{2,3,2,2}\:,\\
b_2
=X_4
\=\frac{\Theta_{3,2}^2\Theta_{3,4}^2}{\Theta_{2,2}^2\Theta_{2,4}^2}
=X^{3,2,3,4}_{2,2,2,4}\:,\\
b_3
=X_5
\=\frac{\Theta_{3,3}^2\Theta_{3,4}^2}{\Theta_{2,3}^2\Theta_{2,4}^2}
=X^{3,3,3,4}_{2,3,2,4}\:.
\end{align*}
Note that we have
\begin{align}\label{eq: symmetry of ramification points}
X^{i,j,k,\ell}_{u,v,s,t}
=X^{k,\ell,i,j}_{s,t,u,v}
\:.
\end{align}
Moreover, we define as in \cite{H:2009}
\begin{align}\label{def: b0 from Lotte Hollands}
b_0
=\frac{\vartheta^4_{3,\tau_1}}{\vartheta^4_{2,\tau_1}}
\:.
\end{align}
We note that when $q=\exp(2\pi\i\tau)$, we have
\begin{align}\label{claim: b0-type expression and rho}
\frac{\vartheta^4_{2,\tau}}{\vartheta^4_{3,\tau}}
\=16q^{\frac{1}{2}}
(1
-8q^{\frac{1}{2}}
+44q
-64q^{\frac{3}{2}}
+O(q^2))
\end{align} 
The linear fractional transformation that sends 
$X_0$ to $0$ and $X_1$ to $1$, maps $X_2$ to $b_0$.
In particular, when $X_0=0,X_1=1$ then $b_0=X_2$.

The finite ramification points on the first torus are obtained from $X_0,X_1,X_2$ in the limit $\nu\rechts 0$.

\begin{claim}
We have
\begin{align}\label{eq: b0 in terms of the ramification points on the first torus}
b_0
=\lim_{\nu\rechts 0}\frac{X_2-X_0}{X_1-X_0}
=\frac{\xi_2-\xi_0}{\xi_1-\xi_0}
\:. 
\end{align}
In particular, as $\rho_1\rechts 0$, $b_0\rechts\infty$.
\end{claim}

Eq.\ (\ref{eq: b0 in terms of the ramification points on the first torus}) for the first torus
is analogous to eq.\ (\ref{eq: the quotient of X3-X4 and X3-X5}) for the second torus,
which we prove in Claim \ref{claim: X5}.

Let $X_0,\ldots,X_5$ be the ramification points of the $g=2$ surface.

\begin{claim}\label{claim: Formula implying expressions for the ramification points in terms of Omega}
Setting, for $k\geq 0$
\begin{displaymath}
R^{(k)}_{i,j}
:=\frac{\vartheta^{(k)}_{i,\Omega_{11}}}{\vartheta_{i,\Omega_{11}}}\frac{\vartheta^{(k)}_{j,\Omega_{22}}}{\vartheta_{j,\Omega_{22}}}
\:,        
\end{displaymath}
where $\vartheta^{(k)}_{i,\Omega_{jj}}=\frac{d^k}{d\Omega_{jj}^k}\vartheta_{i,\Omega_{jj}}$,
we have
\begin{align}\label{eq: dfinition of rational function of theta's}
\frac{\Theta_{i,j}^2\Theta_{k,\ell}^2}{\Theta_{u,v}^2\Theta_{s,t}^2}
=\frac{\vartheta_{i,\Omega_{11}}^2}{\vartheta_{u,\Omega_{11}}^2}
\frac{\vartheta_{j,\Omega_{22}}^2}{\vartheta_{v,\Omega_{22}}^2}
\frac{\vartheta_{k,\Omega_{11}}^2}{\vartheta_{s,\Omega_{11}}^2}
\frac{\vartheta_{\ell,\Omega_{22}}^2}{\vartheta_{t,\Omega_{22}}^2}\:
R^{i,j,k,\ell}_{u,v,s,t}
\end{align}
with
\begin{align*}
R^{i,j,k,\ell}_{u,v,s,t}
=1
\+4\nu^2\:(R^{(1)}_{i,j}+R^{(1)}_{k,\ell}-R^{(1)}_{u,v}-R^{(1)}_{s,t})\\
\+4\nu^4\left(4R^{(1)}_{i,j}R^{(1)}_{k,\ell}+4R^{(1)}_{u,v}R^{(1)}_{s,t}
+\left[R^{(1)}_{i,j}\right]^2+\left[R^{(1)}_{k,\ell}\right]^2+3\left[R^{(1)}_{u,v}\right]^2+3\left[R^{(1)}_{s,t}\right]^2\right)\\
\-16\nu^4\left(R^{(1)}_{i,j}+R^{(1)}_{k,\ell}\right)\left(R^{(1)}_{u,v}+R^{(1)}_{s,t}\right)\\
\+\frac{4}{3}\nu^4
\left(
R^{(2)}_{i,j}+R^{(2)}_{k,\ell}
-R^{(2)}_{u,v}-R^{(2)}_{s,t}
\right)\\
\+16\nu^6\:(R^{(1)}_{i,j}+R^{(1)}_{k,\ell})\left(4R^{(1)}_{u,v}R^{(1)}_{s,t}
+3\left[R^{(1)}_{u,v}\right]^2-\frac{1}{3}\:R^{(2)}_{u,v}
+3\left[R^{(1)}_{s,t}\right]^2-\frac{1}{3}\:R^{(2)}_{s,t}\right)\\
\-16\nu^6\left(R^{(1)}_{u,v}+R^{(1)}_{s,t}\right)\left(4R^{(1)}_{i,j}R^{(1)}_{k,\ell}+\left[R^{(1)}_{i,j}\right]^2+\frac{1}{3}\:R^{(2)}_{i,j}
+\left[R^{(1)}_{k,\ell}\right]^2+\frac{1}{3}\:R^{(2)}_{k,\ell}\right)\\
\+8\nu^6
\Big(2R^{(1)}_{i,j}\left(\left[R^{(1)}_{k,\ell}\right]^2+\frac{1}{3}\:R^{(2)}_{k,\ell}\right)
+2R^{(1)}_{k,\ell}\left(\left[R^{(1)}_{i,j}\right]^2+\frac{1}{3}\:R^{(2)}_{i,j}\right)\\
&\quad+\frac{1}{3}\left(R^{(1)}_{i,j}R^{(2)}_{i,j}+\frac{1}{15}\:R^{(3)}_{i,j}\right)
+\frac{1}{3}\left(R^{(1)}_{k,\ell}R^{(2)}_{k,\ell}+\frac{1}{15}\:R^{(3)}_{k,\ell}\right)
\Big)\\
\+8\nu^6
\Big(
-2R^{(1)}_{u,v}\left(3\left[R^{(1)}_{s,t}\right]^2-\frac{1}{3}\:R^{(2)}_{s,t}\right)
-2R^{(1)}_{s,t}\left(3\left[R^{(1)}_{u,v}\right]^2-\frac{1}{3}\:R^{(2)}_{u,v}\right)\\
&\quad
+R^{(1)}_{u,v}R^{(2)}_{u,v}
-\frac{1}{45}\:R^{(3)}_{u,v}
-4\left[R^{(1)}_{u,v}\right]^3
+R^{(1)}_{s,t}R^{(2)}_{s,t}
-\frac{1}{45}\:R^{(3)}_{s,t}
-4\left[R^{(1)}_{s,t}\right]^3
\Big)
+O(\nu^8)\:.
\end{align*}
We have
\begin{displaymath}
R^{i,j,k,\ell}_{u,v,s,t}
=R^{k,\ell,i,j}_{s,t,uv}
\:. 
\end{displaymath}
\end{claim}

\begin{proof}
Direct calculation.
\end{proof}

In the limit as $\nu\rechts 0$, the $g=2$ surface reduces to a single torus, corresponding to the modulus $\tau_1$.
While $X_0,X_1,X_2$ are the ramification points of the first torus, we find:

\begin{claim}\label{claim: X5}
As $\nu\rechts 0$,
\begin{align*}
X_3,X_4,X_5\rechts b_0
\:,
\end{align*}
where $b_0$ is defined by eq.\ (\ref{def: b0 from Lotte Hollands}) and Claim \ref{claim: b0-type expression and rho}.  
We have
\begin{align*}
X_5
\=b_0(1+O(\nu^2))\\
\=\frac{1}{16q_1^{1/2}}(1+O(q_1^{1/2}))(1+O(\nu^2))\:.
\end{align*}
\end{claim}

\begin{proof}
$X_3,X_4,X_5$ are of the form
\begin{align*}
X^{3,j,3,\ell}_{2,j,2,\ell}
=\frac{\vartheta^4_{3,\Omega_{11}}}{\vartheta^4_{2,\Omega_{11}}}
R^{3,j,3,\ell}_{2,j,2,\ell}
\:,\quad j,\ell\in\{2,3,4\}
\:.
\end{align*}
In particular, for $\nu$ small,
\begin{displaymath}
X^{3,j,3,\ell}_{2,j,2,\ell}
=b_0(1+O(\nu^2))
\:. 
\end{displaymath}
In particular, by eq.\ (\ref{def: b0 from Lotte Hollands}) and Claim \ref{claim: b0-type expression and rho}
\begin{align*}
X_5
\=\frac{1}{16q_1^{1/2}}(1+O(q_1^{1/2}))(1+O(\nu^2))\:.
\end{align*}
(Recall that $\Omega_{jj}$ and $\tau_j$ differ by $O(\nu^2)$ only.)
\end{proof}

\begin{claim}\label{claim: the quotient of X3-X4 and X3-X5}
We have
\begin{align}\label{eq: the quotient of X3-X4 and X3-X5}
\frac{X_5-X_3}{X_4-X_3}
\=\frac{\vartheta^4_{3,\Omega_{22}}}{\vartheta^4_{2,\Omega_{22}}}(1+O(\nu^2))
\:. 
\end{align}
So when $\rho_2$ is small, then so is the distance between $X_3$ and $X_4$.
(Geometrically, the fundamental cell of the torus is stretched to infinity, since as $\rho_2\rechts 0$, we have $\Omega_{22}\rechts\infty$).
\end{claim}

\begin{proof}
Cf.\ Appendix \ref{appendix: Proof of Claim: the quotient of X3-X4 and X3-X5}.  
\end{proof}


\begin{claim}\label{claim: quotients of differences of ramification points}
We have
\begin{align*}
\frac{X_4-X_5}{X_3-X_5}
\=\left(1-\frac{\vartheta^4_{2,\Omega_{22}}}{\vartheta^4_{3,\Omega_{22}}}\right)\left(1+O(\nu^2)\right)\:.
\end{align*}
\end{claim}

\begin{proof}
The proof is similar to that of Claim \ref{claim: the quotient of X3-X4 and X3-X5}.
\end{proof}

\begin{claim}\label{claim: X3-X4}
With the conventions of \cite{H:2009}, 
we have
\begin{displaymath}
X_3-X_4
=\frac{\vartheta_{3,\Omega_{11}}^4}{\vartheta_{2,\Omega_{11}}^4}
\nu^2
\Big(\frac{\pi^2}{4}\vartheta^4_{4,\Omega_{11}}\vartheta^4_{2,\Omega_{22}}+O(\nu^2)
\Big)
\:.
\end{displaymath}
Moreover,
\begin{displaymath}
\frac{X_3-X_5}{X_5}
=\frac{\pi^2}{4}\nu^2
\vartheta_{3,\Omega_{22}}^4\vartheta_{2,\Omega_{11}}^4
+O(\nu^4)
\:.
\end{displaymath}
In particular, when $\rho_1,\rho_2$ are small,
\begin{displaymath}
\frac{X_3-X_5}{X_5}
\sim\frac{\pi^2}{4}\nu^2
+(1+O(\nu^4))
\:.
\end{displaymath}
\end{claim}

\begin{proof}
This calculation is straightforward.
\end{proof}

\subsection{Comparison of the $g=2$ partition functions obtained through either method}

\begin{theorem}(Tuite et al.)
Let 
\begin{displaymath}
h_0(q)=\1_1\:,
\quad
g_0(q)=\1_2 
\end{displaymath}
be the $g=1$ Rogers-Ramanujan partition functions 
defined by eq.\ (\ref{eq: 11/60 Rogers-Ramanujan partition function}) and eq.\ (\ref{eq: -1/60 Rogers-Ramanujan partition function}), respectively.
In the $(2,5)$ minimal model, the $g=2$ partition function satisfies a second order PDE
whose solutions are, to order $\eps^2$,
\begin{align*}
Z^{(2)}_{V,V}(q_1,q_2,\eps)
\=h_0(q_1)h_0(q_2)+O(\eps^2)\:,\\
Z^{(2)}_{W,W}(q_1,q_2,\eps)
\=g_0(q_1)g_0(q_2)+O(\eps^2)\:,\\
Z^{(2)}_{V,W}(q_1,q_2,\eps)
\=h_0(q_1)g_0(q_2)+O(\eps^2)\:,\\
Z^{(2)}_{W,V}(q_1,q_2,\eps)
\=g_0(q_1)h_0(q_2)+O(\eps^2)
\:,
\end{align*}
where the second order terms are obtained through differentiation of $g_0$ resp.\ $h_0$.
In addition, there is a fifth solution given by
\begin{displaymath}
Z^{(2)}_I(q_1,q_2,\eps)
=\eps^{-1/5}\left\{\eta_{\tau_1}^{-2/5}\eta_{\tau_2}^{-2/5}+O(\eps^4)\right\}
\:.
\end{displaymath}
\end{theorem}
 
Our equations yield the same result up to the expected metric factor \cite{L:2017}, and a power of $\eps$ which requires a separate argument.
Our analysis will be published in the coming weeks.
In detail one can write 
\begin{displaymath}
\left(\frac{\vartheta_{2,\tau_i}}{\vartheta_{4,\tau_i}}\right)^2\:,\quad i=1,2
\:, 
\end{displaymath}
as cross ratio of ramification points and the $g_0$, $h_0$ as hypergeometric functions with parameters \cite{L:2017}
\begin{align*}
(a,b,c)
\=\left(\frac{3}{10},-\frac{1}{10},\frac{3}{5}\right)\:,\\
(a,b,c)
\=\left(\frac{7}{10},\:\frac{11}{10},\frac{7}{5}\right)\:,
\end{align*}
respectively.

\pagebreak

\section{Results to leading order in $X=X_1-X_2$ only}

\subsection{Conventions and basic formulae}

We shall vary $X_s$ and leave $X_1,\ldots,\widehat{X_s}\ldots,X_n$ fixed.
Thus 
\begin{displaymath}
\xi_i=\delta_{is}\:, 
\end{displaymath}
where $\delta_{ij}$ is the Kronecker symbol.
We have $d=d_{X_s}$ for 
\begin{displaymath}
d_{X_s}:=\xi_s\frac{\partial}{\partial X_s}\:, 
\end{displaymath}
and $\omega=\omega_s$.
We take
\begin{displaymath}
X_s=X_1,\quad X:=X_1-X_2 
\end{displaymath}
and assume $X$ is small.

\begin{definition}
By definition, two expressions $A,B$ satisfy 
\begin{displaymath}
A\cong B 
\end{displaymath}
if the leading (i.e.\ lowest order) terms in $X$ in $A$ and $B$ are equal.  
\end{definition}

\noi
For instance,
\begin{align}\label{eq: omega giving rise to an OE with regular singularities}
\omega_1
\cong\:\frac{\xi_1}{X_1-X_2}
=\:X^{-1}\xi_1\:. 
\end{align}
We shall need the following:
Suppose $p_x=a_0\prod_{i=1}^n(x-X_i)$.  We have 
\begin{align}
p'_x
\=a_0\sum_{k=1}^n\prod_{i\not=k}(x-X_i)\nn\\
p'_{X_s}
\=\frac{d}{dx}|_{x=X_s}p_x
=a_0\prod_{i\not=s}(X_s-X_i)\label{eq. for p'(Xs)}\:,\\
d_{X_s}p_x
\=-\xi_sa_0\prod_{i\not=s}(x-X_i)\nn\\
(d_{X_s}p)(X_s)
\=d_{X_s}|_{x=X_s}\:p_x
=-\xi_s p'_{X_s}\:.\nn
\end{align}

\begin{claim}
For $k\geq 1$, we have
\begin{align*}
p^{(k)}_{X_s}
\=k\frac{\partial^{k-1}}{\partial X_s^{k-1}}p'_{X_s}\\
d_{X_s}p^{(k)}_{X_s}
\=\frac{k}{k+1}\xi_s p^{(k+1)}_{X_s}
\end{align*}
and to leading (=lowest) order in $X=X_s-X_2$,
\begin{align*}
p'_{X_s}\tilde{\in}\:O(X)\:,\quad p^{(k)}_{X_s}\tilde{\in}\:O(1)\quad\text{for}\:k>1\:,\quad d_{X_s}p'_{X_s}\tilde{\in}\:O(1)\:. 
\end{align*}
\end{claim}

\begin{proof}(Sketch)
For
\begin{displaymath}
f(x,X_s,X_3,\ldots)
=(x-X_s)\:g(x,X_3,\ldots) 
\end{displaymath}
(where in the following we omit the $X_3,\ldots$, which by assumption are all different from $X_s$),
we have
\begin{align*}
f^{(k)}(X_s,X_s)\=kg^{(k-1)}(X_s)\:,\quad k\geq 0\:,
\end{align*}
since $(x-X_s)$ is linear and vanishes at $x=X_s$.
On the other hand,
\begin{align*}
\frac{\partial}{\partial X_s}(f'(X_s,X_s))
\=\frac{\partial}{\partial X_s}g(X_s)
=g'(X_s) 
\end{align*}
since in $g$, $X_s$ stands at the place of $x$.
Now to apply these formulae to $p$, take
\begin{displaymath}
g(x)
=a_0\prod_{i\not=s}(x-X_i)\:,
\quad
f_{xX_s}=(x-X_s)a_0\prod_{i\not=s}(x-X_i)=p_x\:,
\end{displaymath}
and observe that
\begin{displaymath}
g(X_s)=p'_{X_s}\:. 
\end{displaymath}
In particular, for $X=X_s-X_2$,
\begin{displaymath}
p^{(k)}_{X_s}
=ka_0g^{(k-1)}(X_s)
=k\frac{\partial^{k-1}}{\partial X_s^{k-1}}p'_{X_s}\:,
\end{displaymath}
and
\begin{align*}
d_{X_s}p'_{X_s}
\=d_{X_s}g(X_s)
=\xi_s\frac{\partial}{\partial X_s}g(X_s)
=\xi_sg'(X_s)
=\frac{\xi_s}{2}p''_{X_s}\\
d_{X_s}p''_{X_s}
\=2d_{X_s}g'(X_s)
=2\xi_sg''(X_s)
=\frac{2}{3}\xi_sp^{(3)}_{X_s}\\
d_{X_s}p^{(k)}_{X_s}
\=kd_{X_s}g^{(k-1)}(X_s)
=k\xi_s g^{(k)}(X_s)
=\frac{k}{k+1}\xi_s p^{(k+1)}_{X_s}
\:.
\end{align*}
\end{proof}
Moreover,
\begin{align}
\frac{p''_{X_s}}{p'_{X_s}}
\cong&\:
2X^{-1}\:,\label{p''(Xs) over p'(Xs), leading term only}\\
d_{X_s}\frac{1}{p'_{X_s}} 
\=-\xi_s\frac{d_{X_s} p'_{X_s}}{[p'_{X_s}]^2}
\cong\:
-\frac{1}{2}\xi_s\frac{p''_{X_s}}{[p'_{X_s}]^2}\cong\:
-\frac{\xi_s}{a_0}X^{-2}
\:.\label{ds of 1 over p'_{X_s}, leading term only}
\end{align}

Likewise, 
\begin{align}
\frac{p^{(3)}_{X_s}}{p'_{X_s}}
\=6X^{-2}
+48X^{-1}\sum_{i\not=s,2}\frac{1}{X_s-X_i} 
+6\left\{\sum_{i\not=s,2}\frac{1}{(X_s-X_i)^2}
+4\underset{i\not=j}{\sum_{i,j\not=s,2}}\frac{1}{(X_s-X_i)(X_s-X_j)}\right\}\label{p'''(Xs) over p'(Xs)}
\end{align}
The relevant term for us is
\begin{align*}
\left[\frac{p^{(3)}_{X_s}}{p'_{X_s}}\right]_{-1}
\=\:48X^{-1}\sum_{i\not=s,2}\frac{1}{X_s-X_i}\:.
\end{align*}
Note that $\sum_{i\not=s,2}\frac{1}{X_s-X_i}$ is not a number.

\subsection{The first two values for the leading order in the Frobenius ansatz}
\label{Subsection: The first two equations for u}

In the $(2,5)$ minimal model for any genus and to leading (=lowest) order in $X=X_1-X_2$ only, 
we have the closed system of ODEs
\begin{align}
\left(d_{X_s}-\frac{c}{8}\omega_s\right)\1
\=\frac{2\xi_s}{p'_{X_s}}\langle\vartheta_{X_s}\rangle\:,\nn\\
%
\left(d_{X_s}-\frac{c}{8}\omega_s\right)\frac{\langle\vartheta_{X_s}\rangle}{p'_{X_s}} 
\cong&\:2\xi_s
\left[
\frac{7c}{640}\left[\frac{p''_{X_s}}{p'_{X_s}}\right]^2\1
+\frac{1}{5}\frac{p''_{X_s}}{p'_{X_s}}\frac{\langle\vartheta_{X_s}\rangle}{p'_{X_s}}
\right]\:.
\label{ODE for 1-pt function over p'(Xs)}
\end{align}
Because of eq.\ (\ref{eq: omega giving rise to an OE with regular singularities}),
these ODEs have regular singularities, for which the Frobenius method is available.

\begin{claim}\label{claim: the first two values for bar{u}} 
Let $g\geq 1$.
Let $u\in\R$ be the leading order of $\1$ and $\langle\vartheta_{X_s}\rangle$ in the Frobenius ansatz,
and let 
\begin{displaymath}
\bar{u}:=u-\frac{c}{8}\:.
\end{displaymath}
In the $(2,5)$ minimal model,
two values of $\bar{u}$ are given by
\begin{displaymath}
\frac{11}{10}\:,\quad\frac{7}{10}\:.
\end{displaymath}
\end{claim}

\begin{proof}
We have a reason to assume that $\1$ and $\langle\vartheta_{X_s}\rangle$ are of the same leading order,
and use eq.\ (\ref{eq. for p'(Xs)}).
So the Frobenius ansatz reads
\begin{align*}
\1
\cong&\:a(X_s-X_2)^u\\
\frac{\langle\vartheta_{X_s}\rangle}{p'_{X_s}} 
\cong&\:b(X_s-X_2)^{u-1}\:,\quad u\in\R\:,
\end{align*}
where $a,b$ do not depend on $X_s$. 
Thus for $\bar{u}=u-\frac{c}{8}$, this yields 
\begin{align*}
\bar{u}a\=2b\:,\\
(\bar{u}-1)b\=\frac{7c}{80}a+\frac{4}{5}b
\quad
\Leftrightarrow\quad
\left(\bar{u}-\frac{9}{5}\right)b=\frac{7c}{80}a\:,
\end{align*}
It follows that
\begin{align}
\bar{u}\left(\bar{u}-\frac{9}{5}\right)\=\frac{7c}{40}\:.\label{quadratic eq. for bar{u}}
\end{align}
In the $(2,5)$ minimal model, $c=-\frac{22}{5}$, so
\begin{align*}
\bar{u}_{1/2}\=\frac{9}{10}\pm\sqrt{\frac{81}{100}-\frac{77}{100}}
=\frac{9}{10}\pm\frac{1}{5}
=\begin{cases}
\frac{11}{10}&\text{for}\:+\\
\frac{7}{10}&\text{for}\:-
\end{cases}
\:.
\end{align*}
\end{proof}

\noi
Since $\frac{c}{8}=-\frac{11}{20}$, it follows that
\begin{align}\label{eq: exponents of Frobenius ansatz for g geq 1}
u
=\bar{u}-\frac{11}{20}
=\begin{cases}
\frac{11}{20}&\text{for}\:+\\
\frac{3}{20}&\text{for}\:-
\end{cases}
\end{align}

Thus instead of considering the differential eq.\ for $\langle\vartheta_x\rangle$,
we specialise to that for $\langle\vartheta_{X_s}\rangle$.
Since $\langle\vartheta_x\rangle=\langle\vartheta_x\rangle_r$ is a polynomial,
only finitely many equations are to be established.

\subsection{The ODE for $\langle(\vartheta^{[1]})^{(k)}_{X_s}\rangle$ and $\langle(\vartheta^{[y]})^{(k)}_{X_s}\rangle$}

In Subsection \ref{Subsection: The first two equations for u},
we have established the differential eqs 
(\ref{ODE for 0-pt function}) and (\ref{ODE for 1-pt function of vartheta at Xs})
for the $0$-and $1$-point function of $\vartheta$ for arbitrary genus,
and two values of $\bar{u}$. 
We shall now restrict to the $(2,5)$ minimal model and establish 
the third differential equation 
and the third value for $\bar{u}$.

\begin{claim}(The third value)
Let $g\geq 1$ and $k\geq 0$. We assume the $(2,5)$ minimal model.
\begin{enumerate}
\item 
To lowest order in $X=X_s-X_2$,
we have 
\begin{align}\label{eq: general formulation of differential eq. for <k-th derivative of vartheta at Xs>}
\alignedbox{\left(d_{X_s}-\frac{c}{8}\omega_s\right)\langle(\vartheta^{[1]})^{(k)}_{X_s}\rangle}
{\cong\:\frac{2\xi_s}{p'_{X_s}}\left[\langle\vartheta_{X_s}(\vartheta^{[1]})^{(k)}_x\rangle\right]_{\reg}}
\:.
\end{align}
and
\begin{align}\label{eq: general formulation of differential eq. for <k-th derivative of the y-part of vartheta at Xs>}
\alignedbox{\left(d_{X_s}-\frac{c}{8}\omega_s\right)\langle(\vartheta^{[y]})^{(k)}_{X_s}\rangle}
{\cong\:\frac{2\xi_s}{p'_{X_s}}\left[\langle\vartheta_{X_s}(\vartheta^{[y]})^{(k)}_x\rangle\right]_{\reg}}
\:.
\end{align}
\item
In particular, for $k=1$,
\begin{align}\label{the third eq}
\left(d_{X_s}-\frac{c}{8}\omega_s\right)\frac{\langle(\vartheta^{[1]})'_{X_s}\rangle}{p'_{X_s}}
\cong&\:\xi_s\left[\frac{7c}{480}\frac{p''_{X_s}}{p'_{X_s}}\frac{p^{(3)}_{X_s}}{p'_{X_s}}\1
+\frac{11}{30}\frac{p^{(3)}_{X_s}}{p'_{X_s}}\frac{\langle\vartheta_{X_s}\rangle}{p'_{X_s}}
-\frac{3}{20}\frac{p''_{X_s}}{p'_{X_s}}\frac{\langle(\vartheta^{[1]})'_{X_s}\rangle}{p'_{X_s}}\right]\:, 
\end{align}
\item
The third value for $\bar{u}=u-\frac{c}{8}$ is
\begin{displaymath}
\frac{7}{10}\:. 
\end{displaymath}
\end{enumerate}
\end{claim}

In order to establish the corresponding differential eq.\ for $\langle\vartheta''_{X_s}\rangle$,
we need to take the terms $\propto(x-X_s)^2$ into consideration.
Comparison with $N_2(T,T)$ in the ordinary OPE of $T$ for the $(2,5)$ minimal model
does not lead us any further 
since the space of fields of dimension $6$ is two- (rather than one-) dimensional. 

\begin{proof}
\begin{enumerate}
\item 
Only contributions from $\frac{2}{p'_{X_s}}\left[\langle\vartheta_{X_s}\vartheta^{[1]}_x\rangle\right]_{\text{no pole}}$
(resp.\ $\frac{2\xi_s}{p'_{X_s}}\:\left[\langle\vartheta_{X_s}\vartheta^{[y]}_x\rangle\right]_{\text{no pole}}$)
contribute to leading (lowest) order in $X=X_s-X_2$ to the differential equation in Lemma \ref{lemma: differential eq. for the N-pt function of vartheta}.
Replacing $\vartheta^{[1]}_x$ (resp.\ $\vartheta^{[y]}_x$) by its Taylor expansion about $x=X_s$
on both sides of the equation for $\vartheta=\vartheta^{[1]}$ 
(resp.\ for $\vartheta=\vartheta^{[y]}$)  
and comparing the respective coefficient of $(x_2-X_s)^k$ yields the claimed differential eq.\ for $\langle\vartheta^{(k)}_x\rangle$
(resp.\ $\langle\vartheta^{[y]}_x\rangle$).

More specifically, by the Frobenius method,
\begin{align*}
\langle\vartheta(x_2)\ldots\rangle
\=(X_1-X_2)^u(a+O(X_1-X_2))\\
\cong&\:a(X_1-X_2)^u\:, 
\end{align*}
where $a$ is in general a function of $X_2,\ldots,X_n$ and $x_2$,
and
\begin{align*}
\langle\vartheta^{(k)}(X_2)\ldots\rangle
\=(X_1-X_2)^u(\frac{\partial^k}{dx_2^k}|_{x_2=X_2}a+O(X_1-X_2))\:.
\end{align*}
\item
In order to actually compute the r.h.s.\ of eq.\ (\ref{eq: general formulation of differential eq. for <k-th derivative of vartheta at Xs>}) for $k=1$,
we use that 
\begin{align}\label{eq: differentiation can be pulled out of the state} 
\langle\vartheta_{X_s}(\vartheta^{[1]})^{(k)}_x\rangle 
=\frac{\partial^k}{\partial x^k}\langle\vartheta_{X_s}\vartheta^{[1]}_x\rangle
\quad\text{for}\:k\geq 0\:, 
\end{align}
and
\begin{align*}
\Big[\frac{\partial^k}{\partial x^k}\langle\vartheta_{X_s}\vartheta^{[1]}_x\rangle\Big]_{\reg}
=\frac{\partial^k}{\partial x^k}\Big[\langle\vartheta_{X_s}\vartheta^{[1]}_x\rangle\Big]_{\reg}\:.
\end{align*}
The splitting of $\vartheta$ induces a splitting
\begin{align}\label{eq: Galois splitting of <vartheta(Xs) vartheta>}
\langle\vartheta_{X_s}\vartheta_x\rangle
=\langle\vartheta_{X_s}\vartheta^{[1]}_x\rangle+y\langle\vartheta_{X_s}\vartheta^{[y]}_x\rangle\:.
\end{align}
Here by the graphical representation of $\langle\vartheta_{X_s}\vartheta_x\rangle$, eq.\ (\ref{eq.: graphical representation of vartheta's}),
we have (\ref{eq: regular part of <vartheta(Xs) 1-part of vartheta>}) and (\ref{eq: regular part of <vartheta(Xs) y-part of vartheta>}).
Since we aim at a differential eq.\ to leading order terms only
and since $\frac{d}{dx}p'_{X_s}=0$, 
we can immediately restrict our consideration to leading order terms in 
$\left[\frac{2}{p'_{X_s}}\langle\vartheta_{X_s}\vartheta^{[1]}_x\rangle\right]_{\reg}$. 
Using eqs 
(\ref{eq. for f(x,Xs) squared over p'(Xs) up to terms of order (x-Xs) to the cube, to leading order}) 
and 
(\ref{eq. for f(x,Xs) times (vartheta(x)+vartheta(Xs)) over p'(Xs) up to terms of order (x-Xs) to the cube, to leading order})
in the proof of Claim \ref{claim: regular part of <vartheta Xs vartheta Xs>}, 
\begin{align}
\frac{2}{p'_{X_s}}&\frac{\partial}{\partial x}
\left[\frac{c}{32}f_{xX_s}^2\1+\frac{1}{4}f_{xX_s}\left\{\langle\vartheta_{X_s}\rangle+\langle\vartheta^{[1]}_x\rangle\right\}\right]_{\reg}\nn\\
\cong&\:
\frac{c}{96}\frac{p''_{X_s}}{p'_{X_s}}p^{(3)}_{X_s}\1
+\frac{1}{6}\left(\frac{p^{(3)}_{X_s}}{p'_{X_s}}\langle\vartheta_{X_s}\rangle
+\frac{3}{2}\frac{p''_{X_s}}{p'_{X_s}}\langle(\vartheta^{[1]})'_{X_s}\rangle\right)
\nn\\
\+\frac{c}{96}
\left(
\frac{1}{2}\frac{p''_{X_s}}{p'_{X_s}}p^{(4)}_{X_s}+\frac{1}{3}\frac{[p^{(3)}_{X_s}]^2}{p'_{X_s}}
\right)\1(x-X_s)\nn\\
&\hspace{2cm}+\frac{1}{12}\left(\frac{p^{(4)}_{X_s}}{p'_{X_s}}\langle\vartheta_{X_s}\rangle
+2\frac{p^{(3)}_{X_s}}{p'_{X_s}}\langle(\vartheta^{[1]})'_{X_s}\rangle
+3\frac{p''_{X_s}}{p'_{X_s}}\langle(\vartheta^{[1]})''_{X_s}\rangle
\right)(x-X_s)\nn\\
\+O((x-X_s)^2)\:.\label{eq: first derivative of the regular content of the singular part of the graphical rep of <vartheta(Xs) vartheta(x)>}
\end{align}
Moreover, 
\begin{align}\label{eq: first derivative of <vartheta(Xs) vartheta>r}
\frac{\partial}{\partial x}\langle\vartheta_{X_s}\vartheta^{[1]}_x\rangle_r
\=\frac{1}{2}\langle\psi'_x\rangle
+\langle\vartheta'_x\vartheta'_x\rangle_r\:(X_s-x)
+O\left((X_s-x)^2\right)\:,
\end{align}
where $\langle\psi_x\rangle=\langle\psi^{[1]}_x\rangle+y\langle\psi^{[y]}_x\rangle$ and
\begin{align*}
\langle\psi'_x\rangle
\=\langle(\psi^{[1]})'_x\rangle+y\left(\partial_x+\frac{1}{2}\frac{p'_x}{p_x}\right)\langle\psi^{[y]}_x\rangle\\
\=\partial_x\langle\vartheta^{[1]}_x\vartheta^{[1]}_x\rangle_r
+p'_x\langle\vartheta^{[y]}_x\vartheta^{[y]}_x\rangle_r
+p_x\partial_x\langle\vartheta^{[y]}_x\vartheta^{[y]}_x\rangle_r+O(y)\:.
\end{align*}
It follows that
\begin{align}\label{eq: first derivative of <vartheta(Xs) 1-part of vartheta>r}
\frac{\partial}{\partial x}|_{X_s}\langle\vartheta_{X_s}\vartheta^{[1]}_x\rangle_r
\=\frac{1}{2}\partial_x|_{X_s}\langle\vartheta^{[1]}_x\vartheta^{[1]}_x\rangle_r
+\frac{1}{2}p'_{X_s}\langle\vartheta^{[y]}_{X_s}\vartheta^{[y]}_{X_s}\rangle_r
\:.
\end{align}
To obtain the first term on the r.h.s.\ of eq.\ (\ref{eq: first derivative of <vartheta(Xs) 1-part of vartheta>r}) , we differentiate $\langle\psi_x\rangle$ given eq.\ (\ref{eq: definition of psi}),
\begin{align}\label{eq: first derivative of Psi(x)}
\partial_x\langle\vartheta^{[1]}_x\vartheta^{[1]}_x\rangle_r
\=-\frac{c}{480}\left(p'_xp^{(4)}_x-2p''_xp^{(3)}_x\right)\1\nn\\
\+\left\{
-\frac{1}{5}p_x\:\partial_x^3
-\frac{3}{10}p'_x\:\partial_x^2
+\frac{1}{10}p''_x\:\partial_x
+\frac{1}{5}p^{(3)}_x
\right\}\langle\vartheta^{[1]}\rangle
\:.
\end{align}
$\langle\psi'_x\rangle$ is regular at $x=X_s$, and its derivative at $X_s$ equals
\begin{align*}
\partial_x|_{X_s}\langle\vartheta^{[1]}_x\vartheta^{[1]}_x\rangle_r
\=-\frac{c}{480}\left(p'_{X_s}p^{(4)}_{X_s}-2p''_{X_s}p^{(3)}_{X_s}\right)\1\\
\+\left\{
-\frac{3}{10}p'_{X_s}\:\partial_x^2|_{X_s}
+\frac{1}{10}p''_{X_s}\:\partial_x|_{X_s}
+\frac{1}{5}p^{(3)}_{X_s}
\right\}\langle\vartheta^{[1]}\rangle
\:.
\end{align*}
The second term on the r.h.s.\ of of eq.\ (\ref{eq: first derivative of <vartheta(Xs) 1-part of vartheta>r}) does not contribute to leading order.

Multiplying eq.\ (\ref{eq: first derivative of <vartheta(Xs) vartheta>r}) by $\frac{2}{p'_{X_s}}$
and adding to eq.\ (\ref{eq: first derivative of the regular content of the singular part of the graphical rep of <vartheta(Xs) vartheta(x)>})
yields $\xi_s^{-1}\Big[\langle\vartheta_{X_s}(\vartheta^{[1]})'_x\rangle \Big]_{\reg}$.
Evaluated at $x=X_s$, this yields according to eq.\ (\ref{eq: general formulation of differential eq. for <k-th derivative of vartheta at Xs>}),
\begin{align*}
\left(d_{X_s}-\frac{c}{8}\omega_s\right)\langle(\vartheta^{[1]})'_{X_s}\rangle
\cong&\:\xi_s\left[\frac{7c}{480}\frac{p''_{X_s}}{p'_{X_s}}p^{(3)}_{X_s}\1
+\frac{11}{30}\frac{p^{(3)}_{X_s}}{p'_{X_s}}\langle\vartheta_{X_s}\rangle
+\frac{7}{20}\frac{p''_{X_s}}{p'_{X_s}}\langle(\vartheta^{[1]})'_{X_s}\rangle\right]\:.
\end{align*}
(Note that since this is an equation to leading order only, we have omitted terms $\propto p'_{X_s}$.)
Using eqs (\ref{p''(Xs) over p'(Xs), leading term only}) and  (\ref{ds of 1 over p'_{X_s}, leading term only}) 
yields the claimed differential eq.\ for $\langle(\vartheta^{[1]})'_{X_s}\rangle$.
\item
After change to the basis 
$\1,\frac{\langle\vartheta_{X_s}\rangle}{p'_{X_s}},\frac{\langle(\vartheta^{[1]})'_{X_s}\rangle}{p'_{X_s}}$,
we have
\begin{align*}
\left(d_{X_s}-\frac{c}{8}\omega_s\right)\1
\=2\xi_s\frac{\langle\vartheta_{X_s}\rangle}{p'_{X_s}}\:,\nn\\
\left(d_{X_s}-\frac{c}{8}\omega_s\right)\frac{\langle\vartheta_{X_s}\rangle}{p'_{X_s}} 
\cong&\:\xi_s
\left[
\frac{7c}{80}X^{-2}\1
+\frac{4}{5}X^{-1}\frac{\langle\vartheta_{X_s}\rangle}{p'_{X_s}}
\right]\nn\\
\left(d_{X_s}-\frac{c}{8}\omega_s\right)\frac{\langle(\vartheta^{[1]})'_{X_s}\rangle}{p'_{X_s}}
\cong&\:\xi_s\left[\frac{7c}{240}X^{-1}\frac{p^{(3)}_{X_s}}{p'_{X_s}}\1
+\frac{11}{30}\frac{p^{(3)}_{X_s}}{p'_{X_s}}\frac{\langle\vartheta_{X_s}\rangle}{p'_{X_s}}
-\frac{3}{10}X^{-1}\frac{\langle(\vartheta^{[1]})'_{X_s}\rangle}{p'_{X_s}}\right]\:,
\end{align*}
or
\begin{align*}
\begin{pmatrix}
\bar{u}&-2&0\\
-\frac{7c}{80}&\bar{u}-\frac{9}{5}&0\\
-\frac{7c}{240}\frac{p^{(3)}_{X_s}}{p'_{X_s}}&-\frac{11}{30}\frac{p^{(3)}_{X_s}}{p'_{X_s}}&\bar{u}-\frac{7}{10}
\end{pmatrix}
\begin{pmatrix}
a\\
b\\
c
\end{pmatrix}
=0\:,
\end{align*}
and
\begin{align*}
0
=\det
\=\left(\bar{u}-\frac{7}{10}\right)
\det
\begin{pmatrix}
\bar{u}&-2\\
-\frac{7c}{40}&\bar{u}-\frac{9}{5}    
\end{pmatrix}\:.
\end{align*}
So the third value is $\bar{u}=\frac{7}{10}$.
\end{enumerate}
\end{proof}

\subsection{Check: The differential equation for $N$-point functions 
of $\vartheta$ and its $k$th derviative, for arbitray genus}

We check that no logarithmic solutions can arise in the system.

\begin{lemma}\label{lemma: diff. eq. for the N point function}(Differential eq.\ for the $N$-point function)\\
Let
\begin{displaymath}
dX_i=\xi_i\quad\text{with}\quad\xi_1\not=0\:,\quad\xi_i=0\quad\text{for}\:i\not=1\:. 
\end{displaymath}
Let $k\geq 0$.
We have
\begin{align*}
\left(d_{X_s}-\frac{c}{8}\omega_s\right)\langle(\vartheta^{[1]})^{(k)}_{X_s}\ldots\rangle
\cong&\:\frac{2\xi_s}{p'_{X_s}}\left[\langle\vartheta_{X_s}(\vartheta^{[1]})^{(k)}_x\ldots\rangle\right]_{\underset{x=X_s}{\reg}}\:,\\
\left(d_{X_s}-\frac{c}{8}\omega_s\right)\frac{\left[\langle\vartheta_{X_s}(\vartheta^{[1]})^{(k)}_x\ldots\rangle\right]_{\underset{x=X_s}{\reg}}}{p'_{X_s}}
\cong&\:2\xi_s
\left[
\frac{7c}{640}\left[\frac{p''_{X_s}}{p'_{X_s}}\right]^2\langle(\vartheta^{[1]})^{(k)}_{X_s}\ldots\rangle
+\frac{1}{5}\frac{p''_{X_s}}{p'_{X_s}}\frac{\left[\langle\vartheta_{X_s}(\vartheta^{[1]})^{(k)}_x\ldots\rangle\right]_{\underset{x=X_s}{\reg}}}{p'_{X_s}}
\right]\:,
\end{align*}
where $\left[\:\right]_{\reg}$ denotes the restriction to the terms regular in the first two positions. 
The system closes up, and setting
\begin{align*}
\left[\langle(\vartheta^{[1]})^{(k)}_{X_s}\ldots\rangle\right]_{\reg}
\cong&\:aX^u\\
\frac{\left[\langle\vartheta_{X_s}(\vartheta^{[1]})^{(k)}_x\ldots\rangle\right]_{\underset{x=X_s}{\reg}}}{p'_{X_s}}
\cong&\:bX^{u-1}\:,\quad u\in\R\:,
\end{align*}
the two values for $\bar{u}=u-\frac{c}{8}$ are
\begin{displaymath}
\frac{11}{10}\:,\quad \frac{7}{10}\:. 
\end{displaymath}
\end{lemma}

\begin{remark}
We don't know what $\left[\langle\vartheta_{X_s}\vartheta^{(k)}_x\rangle\right]_{\underset{x=X_s}{\reg}}$ is in general.
However, we can conclude (for $k=2$) that  
\begin{align*}
\bar{u}_{4/5}=
\begin{cases}
\frac{11}{10}&\\
\frac{7}{10} &
\end{cases}\:.
\end{align*}
For $k=3$, we have an explicit expression for $n=5$.
\end{remark}

\begin{proof}
In the following, let 
\begin{displaymath}
\vartheta=\vartheta^{[1]}\:. 
\end{displaymath}
Let $X_2,\ldots,X_n$ be fixed and $x_2,\ldots$ are arbitrary, but mutually different and different from $X_s$.
By Lemma \ref{lemma: differential eq. for the N-pt function of vartheta} 
and Remark \ref{remark: In the differential eq. for varthetas, singular terms at x=Xs drop out},
we have
\begin{align}
\left(d_{X_s}-\frac{c}{8}\omega_s\right)
\langle\vartheta(x_1)\ldots\rangle
\cong\:\frac{2\xi_s}{p'_{X_s}}\left[\langle\vartheta_{X_s}\vartheta(x_1)\ldots\rangle\right]_{\reg}\:,
\label{ODE for N-pt function, to leading order only }
\end{align}
to leading order in $X$. 
(On the r.h.s.\ we have restricted to terms which are regular at $x_2=X_s$.)
Here 
\begin{align*}
\left(d_{X_s}-\frac{c}{8}\omega_s\right)\langle\vartheta(x_1)\vartheta(x_2)\ldots\rangle|_{x_1=X_s}
\cong\:
\left(d_{X_s}-\frac{c}{8}\omega_s\right)\langle\vartheta_{X_s}\vartheta(x_2)\ldots\rangle\:,
\end{align*}
so we can replace $x_1$ by $X_s$ on both sides,
\begin{align*}
\left(d_{X_s}-\frac{c}{8}\omega_s\right)
\langle\vartheta_{X_s}\ldots\rangle
\cong\:\frac{2\xi_s}{p'_{X_s}}\left[\langle\vartheta_{X_s}\vartheta_x\ldots\rangle\right]_{\underset{x=X_s}{\reg}}\:,
\end{align*}
yielding the first of the claimed equations for $k=0$.
We address the second equation. 
The same arguments that prove eq.\  (\ref{ODE for 1-pt function over p'(Xs)})
also show 
\begin{align*}
\left(d_{X_s}-\frac{c}{8}\omega_s\right)\frac{\langle\vartheta_{X_s}\vartheta(x_2)\ldots\rangle}{p'_{X_s}}
\cong&\:2\xi_s
\left[
\frac{7c}{640}\left[\frac{p''_{X_s}}{p'_{X_s}}\right]^2\langle\vartheta(x_2)\ldots\rangle
+\frac{1}{5}\frac{p''_{X_s}}{p'_{X_s}}\frac{\langle\vartheta_{X_s}\vartheta(x_2)\ldots\rangle}{p'_{X_s}}
\right]\:.
\end{align*}
We restrict to the terms regular at $x_2=X_s$.
Since by holomorphy of $\langle\vartheta_{X_s}\vartheta(x_2)\rangle$ outside $x_2=X_s$, 
the coefficients of its Laurent series expansion can be defined by contour integrals,
we have
\begin{align}\label{eq: restricting the N=pt function to its regular part commutes with taking ds}
\left[\left(d_{X_s}-\frac{c}{8}\omega_s\right)\langle\vartheta_{X_s}\vartheta(x_2)\ldots\rangle\right]_{\reg}
=\left(d_{X_s}-\frac{c}{8}\omega_s\right)\left[\langle\vartheta_{X_s}\vartheta(x_2)\ldots\rangle\right]_{\reg}\:.
\end{align}
Now setting $x_2=X_s$ yields the second claimed equation for $k=0$.
Alternatively, we replace of $\vartheta(x_2)$ by its Taylor series expansion about $x_2=X_s$. 
Comparing the terms $\propto (x_2-X_s)^k$ yields  the claimed system. 
This system closes up. 
For the given Frobenius ansatz,
the arguments used in the proof of Claim \ref{claim: the first two values for bar{u}}, 
we obtain the two claimed values for $\bar{u}$. 
\end{proof}

\subsection{The number of equations to leading order}

We have 
\begin{displaymath}
d_{X_s}\1\sim\langle\vartheta^{[1]}_{X_s}\rangle\:. 
\end{displaymath}
and all differential equations for $N$-point functions of $\vartheta^{[1]}$ and its derivatives do not involve $\vartheta^{[y]}$.
So set $\vartheta=\vartheta^{[1]}$, and let $N\geq 1$.
By Lemma \ref{lemma: diff. eq. for the N point function}, for $k=0,\ldots,n-3$, (with $\sharp\{k\}=\deg\langle\vartheta\rangle$),
\begin{displaymath}
d_{X_s}\langle\vartheta^{(k)}_{X_s}\rangle\sim\left[\langle\vartheta_{X_s}\vartheta^{(k)}_x\rangle\right]_{\underset{x-X_s}{\reg}}\:.
\end{displaymath}
In the $(2,5)$ minimal model, the r.h.s.\ is known for both $k=0,1$.
For the remaining $n-4$ values of $k$ we have by Lemma \ref{lemma: diff. eq. for the N point function},
\begin{displaymath}
d_{X_s}\left[\langle\vartheta_{X_s}\vartheta^{(k)}_x\rangle\right]_{\underset{x-X_s}{\reg}}
\sim\langle\vartheta^{(k)}\rangle+\left[\langle\vartheta_{X_s}\vartheta^{(k)}_x\rangle\right]_{\underset{x-X_s}{\reg}}\:. 
\end{displaymath}
So
\begin{displaymath}
d_{X_s}^2\langle\vartheta^{(k)}\rangle
\:\sim\:
d_{X_s}\langle\vartheta^{(k)}\rangle 
+\langle\vartheta^{(k)}\rangle\:,
\end{displaymath}
and both $\langle\vartheta^{(k)}\rangle$ and $\left[\langle\vartheta_{X_s}\vartheta^{(k)}_x\rangle\right]_{\underset{x-X_s}{\reg}}$ are known
as functions of $X$. 
So to leading order in $X$, $\1$ is determined by 
\begin{displaymath}
1+(n-2)+(n-4)
=2n-5 
\end{displaymath}
equations, whenever $n\geq 4$, and $n-1$ otherwise.

\pagebreak

\section{Application to the $(2,5)$ minimal model for $g=2$}

\subsection{The fifth equation}

We need to know $\langle\vartheta_x\rangle$ and $\langle\vartheta_x\vartheta_{X_s}\rangle$ in the $(2,5)$ minimal model for $n=5$.

\begin{enumerate}
\item
In the limit of $\langle\vartheta_1\vartheta_2\rangle_r$ as $x_1\rechts x_2$, 
\begin{displaymath}
B_0(x_1^2+x_2^2)+\mathbf{B_{1,1}}x_1x_2\:\mapsto\:(2B_0+\mathbf{B_{1,1}})x^2
\end{displaymath}
so knowledge of $\langle\vartheta^2\rangle_r$ determines $\langle\vartheta_1\vartheta_2\rangle_r$ only up to one unknown.
\item 
We computed $\langle\vartheta_2^2\vartheta_3\rangle_r$
with the $(2,5)$ minimal model property (the formula for $\psi_2$) implemented. 
$\langle\vartheta_2^2\vartheta_3\rangle_r$ is a function of $\1,\langle\vartheta\rangle$ and of $\langle\vartheta_2\vartheta_3\rangle_r$.
We considered the change in $\langle\vartheta_2^2\vartheta_3\rangle_r$ produced by
\begin{align}\label{change of two-point function of vartheta}
\langle\vartheta_i\vartheta_j\rangle_r\:\mapsto\:\langle\vartheta_i\vartheta_j\rangle_r+(x_i-x_j)^2\:,
\end{align}
for $(i,j)=(2,3)$.
(Since all terms of order $O(x^3)$ are known in $\langle\vartheta_1\vartheta_2\rangle_r$, this is the only change to consider.) 
The new terms in $\langle\vartheta_2^2\vartheta_3\rangle_r$ resulting from (\ref{change of two-point function of vartheta}) 
all lift, along the projection $x_1\mapsto x_2$, to symmetric polynomials $[k\ell m]$ (i.e.\ $x_1^k,x_2^{\ell},x_3^m+$ permutations) of order $k+\ell+m=5$,
namely $[500],[320],[311],[221]$, and $y_1y_2+y_2y_3+y_3y_1$ ($[410]$ does not occur). 
Projecting $[221]$ yields
\begin{displaymath}
x_1^2x_2^2x_3+x_1x_2^2x_3^2+x_1^2x_2x_3^2
\:\overset{x_1\rechts x_2}{\longrightarrow}\:
\underset{\text{with known coeff.\ $B_{0,0,1}$}}{x_2^4x_3+2x_2^3x_3^2}
\sim\underset{\ldots(x_2^2+x_3^2+x_2x_3)}{\underset{\uparrow}{\vartheta_2^h(2x_3^2+x_2x_3)}}
\:\overset{x_2\rechts x_3}{\longrightarrow}\:\underset{3B_{0,0,1}\propto 2B_0+\mathbf{B_{1,1}}}{3\vartheta_3^h x_3^2}\:.
\end{displaymath}
\item 
$\langle\vartheta_1\vartheta_2\vartheta_3\rangle_r$ is a function of $\1,\langle\vartheta\rangle$ and $\langle\vartheta_i\vartheta_j\rangle_r$.
We computed $\langle\vartheta_1\vartheta_2\vartheta_3\rangle_r^h=\langle\vartheta_1\vartheta_2\rangle_r\vartheta_3^h$ with $\vartheta_3^h=-\frac{3ca_0}{4}x_3^3.1$
(Laurent series as $x_3\rechts\infty$) 
and the terms produced by the change (\ref{change of two-point function of vartheta}) in $\langle\vartheta_1\vartheta_2\vartheta_3\rangle_r^h$ 
for $i,j=1,2,3$.
These are $[500],[320],[311]$, which are known as they are of order $\geq 3$ in one variable, and $y_1y_2+y_2y_3+y_3y_1$.
$[221]$ is not produced due to our restriction to $\langle\vartheta_1\vartheta_2\vartheta_3\rangle_r^h$.
Unexpectedly, the coefficients of all occurring terms perfectly match, yielding no constraint on $[221]$.
\end{enumerate}
\noi
The $(2,5)$ minimal model constraint does not provide any further information on the $3$-point function.

For $n=5$, we are interested in $\mathbf{B_{1,1}}$,
so it suffices to formulate the fifth differential eq.\ for $\langle\vartheta_3\vartheta_3''\rangle_r$ at $x_3=X_s$,
or for $\left[\langle\vartheta_{X_s}\vartheta''_x\rangle\right]_{\underset{x=X_s}{\reg}}$.
By Lemma \ref{lemma: diff. eq. for the N point function}, we have 
\begin{align}\label{eq: The fifth equation}
\left(d_{X_s}-\frac{c}{8}\omega_s\right)\frac{\left[\langle\vartheta_{X_s}\vartheta''_x\rangle\right]_{\underset{x=X_s}{\reg}}}{p'_{X_s}}
\cong&\:2\xi_s
\left[
\frac{7c}{640}\frac{[p''_{X_s}]^2}{p'_{X_s}}\frac{\langle\vartheta''_{X_s}\rangle}{p'_{X_s}}
+\frac{1}{5}\frac{p''_{X_s}}{p'_{X_s}}\frac{\left[\langle\vartheta_{X_s}\vartheta''_{X_s}\rangle\right]_{\underset{x=X_s}{\reg}}}{p'_{X_s}}
\right]\:. 
\end{align}

\begin{remark}
From the formula for $\langle\vartheta_2\vartheta_2\vartheta_3\rangle_r$ we can deduce
\begin{align*}
\langle\vartheta_1\vartheta_2\vartheta_3\rangle_r\mod\text{terms in $\ker_{1\rechts 2}$}
\end{align*}
where the kernel is of the form
\begin{displaymath}
(x_1-x_2)^2(x_1-x_3)^2(x_2-x_3)^2\times\text{polynomial}\:. 
\end{displaymath}
However, the latter is of order $O(x^4)$ and thus known. 
\end{remark}

\subsection{The full matrix of the system of differential equations for $\1$ and derivatives of $\langle\vartheta\rangle$ for $n=5$}

For $n=5$, $\vartheta^{[y]}_x$ is absent for degree reason, so
\begin{displaymath}
\vartheta_x=\vartheta^{[1]}_x\:. 
\end{displaymath}

\begin{theorem}\label{theorem: the full system of differential equations for n=5}
We assume the $(2,5)$ minimal model for $g=2$ ($n=5$), (with $a_1=0$).
To leading (=lowest) order in $X=X_1-X_2$, we have the system of ODEs
\begin{align*}
\left(d_{X_s}-\frac{c}{8}\omega_s\right)\1
\=2\xi_s\frac{\langle\vartheta_{X_s}\rangle}{p'_{X_s}}\:,\nn\\
\left(d_{X_s}-\frac{c}{8}\omega_s\right)\frac{\langle\vartheta_{X_s}\rangle}{p'_{X_s}} 
\cong&\:\xi_s
X^{-1}
\left[
\frac{7c}{80}X^{-1}\1
+\frac{4}{5}\frac{\langle\vartheta_{X_s}\rangle}{p'_{X_s}}
\right]\nn\\
\left(d_{X_s}-\frac{c}{8}\omega_s\right)\frac{\langle\vartheta'_{X_s}\rangle}{p'_{X_s}}
\cong&\:\xi_s\frac{1}{3}\frac{p^{(3)}_{X_s}}{p'_{X_s}}
\left[\frac{7c}{80}X^{-1}\1
+\frac{11}{10}\frac{\langle\vartheta_{X_s}\rangle}{p'_{X_s}}\right]
-\xi_s
\left[\frac{3}{10}X^{-1}\frac{\langle\vartheta'_{X_s}\rangle}{p'_{X_s}}\right]\nn\\
\left(d_{X_s}-\frac{c}{8}\omega_s\right)\frac{\langle\vartheta''_{X_s}\rangle}{p'_{X_s}}
\cong&\:\xi_s\frac{c}{96}\left[
X^{-1}\frac{p^{(4)}_{X_s}}{p'_{X_s}}
+\frac{1}{3}\left[\frac{p^{(3)}_{X_s}}{p'_{X_s}}\right]^2
\right]\1\nn\\
\+\xi_s\left[\frac{1}{12}\frac{p^{(4)}_{X_s}}{p'_{X_s}}\frac{\langle\vartheta_{X_s}\rangle}{p'_{X_s}}
+\frac{1}{6}\frac{p^{(3)}_{X_s}}{p'_{X_s}}\frac{\langle\vartheta'_{X_s}\rangle}{p'_{X_s}}
-\frac{1}{2}X^{-1}\frac{\langle\vartheta''_{X_s}\rangle}{p'_{X_s}}\right]\nn\\
\+2\xi_s\frac{\langle\vartheta_{X_s}\vartheta''_{X_s}\rangle_r }{[p'_{X_s}]^2}\:,
\\
\left(d_{X_s}-\frac{c}{8}\omega_s\right)\frac{\langle\vartheta_{X_s}\vartheta_{X_s}''\rangle_r}{[p_{X_s}']^2}
\cong\:
&\:\xi_s
\frac{c}{640}\left[\frac{1}{2}X^{-2}\frac{p_{X_s}^{(4)}}{p_{X_s}'}
-\frac{1}{9}X^{-1}\left[\frac{p_{X_s}^{(3)}}{p_{X_s}'}\right]^2\right]\1\\ 
\+\frac{1}{80}\xi_s\left[2X^{-1}\frac{p_{X_s}^{(4)}}{p_{X_s}'}
-\frac{11}{9}\left[\frac{p_{X_s}^{(3)}}{p_{X_s}'}\right]^2\right]\frac{\langle\vartheta_{X_s}\rangle}{p_{X_s}'}\nn\\
\+\frac{1}{20}\xi_sX^{-1}\frac{p_{X_s}^{(3)}}{p_{X_s}'}\frac{\langle\vartheta_{X_s}'\rangle}{p_{X_s}'}
-\frac{3}{50}\xi_sX^{-2}\frac{\langle\vartheta_{X_s}''\rangle}{p_{X_s}'}\nn\\
\-\frac{7}{10}\xi_sX^{-1}\frac{\langle\vartheta_{X_s}\vartheta_{X_s}''\rangle_r}{[p_{X_s}']^2}
+O(X_s-X_2)\:.
\end{align*}
(Note that with assumption $X_1=X_2$ made for the $5$th equation, $(p'_{X_s})^{-1}$ is not defined. We have to pull $p''_{X_s}$ out.)
\end{theorem}

Note that the first three equations have been shown for arbitrary $g\geq 1$.
(The first two have derived from the exact equations 
(\ref{ODE for 0-pt function}) and (\ref{ODE for 1-pt function of vartheta at Xs}).
The third is eq.\ (\ref{the third eq}).)

\begin{proof}
Under the assumptions of the Theorem, the fourth differential equation reads, 
to leading (=lowest) order in $X=X_1-X_2$, 
\begin{align*}
\left(d_{X_s}-\frac{c}{8}\omega_s\right)\frac{\langle\vartheta''_{X_s}\rangle}{p'_{X_s}}
\cong&\:\xi_s\left[
\frac{c}{192}\frac{p''_{X_s}}{p'_{X_s}}\frac{p^{(4)}_{X_s}}{p'_{X_s}}
+\frac{c}{288}\left[\frac{p^{(3)}_{X_s}}{p'_{X_s}}\right]^2
\right]\1\nn\\
\+\xi_s\left[\frac{1}{12}\frac{p^{(4)}_{X_s}}{p'_{X_s}}\frac{\langle\vartheta_{X_s}\rangle}{p'_{X_s}}
+\frac{1}{6}\frac{p^{(3)}_{X_s}}{p'_{X_s}}\frac{\langle\vartheta'_{X_s}\rangle}{p'_{X_s}}
-\frac{1}{4}\frac{p''_{X_s}}{p'_{X_s}}\frac{\langle\vartheta''_{X_s}\rangle}{p'_{X_s}}\right]\nn\\
\+2\xi_s\frac{\langle\vartheta_{X_s}\vartheta''_{X_s}\rangle_r }{[p'_{X_s}]^2}\:.
\end{align*} 
(For the proof, cf.\ Appendix \ref{proof: fourth differential equation when $n=5$} or \ref{proof: alternative proof of fourth differential equation when $n=5$}.)
Furthermore, eqs (\ref{p''(Xs) over p'(Xs), leading term only}) and (\ref{ds of 1 over p'_{X_s}, leading term only}) apply.
The fifth eq.\ is obtained from eq.\ (\ref{eq: The fifth equation}).
\end{proof}

\subsection{Monodromy matrix for $n=5$}

Let $\vec{Y}$ be the fundamental system in the basis that corresponds to the Frobenius expansion in powers of $X=X_1-X_2$,
\begin{align*}
\vec{y}
=
\begin{pmatrix}
\frac{\1}{p'_{X_s}}\\
\frac{\langle\vartheta_{X_s}\rangle}{p'_{X_s}}\\
\frac{\langle\vartheta'_{X_s}\rangle}{p'_{X_s}}\\
\frac{\langle\vartheta''_{X_s}\rangle}{p'_{X_s}}\\
\frac{\langle\vartheta_{X_s}\vartheta_{X_s}''\rangle_r}{p_{X_s}'} 
\end{pmatrix}
\cong
\begin{pmatrix}
a(X_1-X_2)^{u-1}\\
b(X_1-X_2)^{u-1}\\
c(X_1-X_2)^{u-1}\\
d(X_1-X_2)^{u-1}\\
e(X_1-X_2)^{u-1}
\end{pmatrix}
\:,
\quad
\vec{Y}
=(\vec{y}_1,\ldots,\vec{y}_5)\:
\end{align*}
For $s=1$,
\begin{align*}
\frac{d}{dX_1}\vec{y}
\:\cong\:
\left(
\sum_{i\not=1}\frac{c/8}{X_1-X_i}
+\frac{B}{p'_{X_s}}
\right)\:\vec{y}
\end{align*}
where
\begin{align*}
B=
\begin{pmatrix}
-\frac{1}{2}p''_{X_1}&2&0&0&0\\
\frac{7c}{320}[p''_{X_1}]^2
&\frac{2}{5}p''_{X_1}
&0&0&0\\
\frac{7c}{480}p''_{X_1}p^{(3)}_{X_1}&\frac{11}{30}p^{(3)}_{X_s}&-\frac{3}{20}p''_{X_s}&0&0\\
\frac{c}{192}p''_{X_s}p^{(4)}_{X_s}
+\frac{c}{288}[p^{(3)}_{X_s}]^2
&\frac{1}{12}p^{(4)}_{X_s}
&\frac{1}{6}p^{(3)}_{X_s}
&-\frac{1}{4}p''_{X_s}
&2\\
0&0&0&
\frac{7c}{320}[p''_{X_s}]^2
&\frac{2}{5}p''_{X_s}
\end{pmatrix}
\:.
\end{align*}

\pagebreak

\begin{cor}
Using the Frobenius ansatz 
\begin{align*}
\1
\cong&\:a(X_s-X_2)^u\\
\frac{\langle\vartheta_{X_s}\rangle}{p'_{X_s}}
\cong&\:b(X_s-X_2)^{u-1}\\
\frac{\langle\vartheta'_{X_s}\rangle}{p'_{X_s}}
\cong&\:c(X_s-X_2)^{u-1}\\
\frac{\langle\vartheta''_{X_s}\rangle}{p'_{X_s}}
\cong&\:d(X_s-X_2)^{u-1}\\
\frac{\langle\vartheta_{X_s}\vartheta_{X_s}''\rangle_r}{[p'_{X_s}]^2}
\cong&\:e(X_s-X_2)^{u-2}\:,\quad u\in\R\:,
\end{align*}
the system of DEs takes the form
(Note that with assumption $X_1=X_2$ made for the $5$th equation, $(p'_{X_s})^{-1}$ is not defined. We have to pull $p''_{X_s}$ out.)
\begin{align*}
\begin{pmatrix}
\bar{u}&-2&0&0&0\\
-\frac{7c}{80}&\bar{u}-\frac{9}{5}&0&0&0\\
-\frac{7c}{240}\left[\frac{p^{(3)}_{X_s}}{p'_{X_s}}\right]_{-1}&-\frac{11}{30}\left[\frac{p^{(3)}_{X_s}}{p'_{X_s}}\right]_{-1}&\bar{u}-\frac{7}{10}&0&0\\
\frac{c}{96}\left[\frac{p^{(4)}_{X_s}}{p'_{X_s}}\right]_{-1}
+\frac{c}{288}\left[\left(\frac{p^{(3)}_{X_s}}{p'_{X_s}}\right)^2\right]_{-1}
&\frac{1}{12}\left[\frac{p^{(4)}_{X_s}}{p'_{X_s}}\right]_{-1}
&\frac{1}{6}\left[\frac{p^{(3)}_{X_s}}{p'_{X_s}}\right]_{-1}
&\bar{u}-\frac{1}{2}
&-2\\
\frac{c}{40}\left\{\frac{1}{32}\left[\frac{p_{X_s}^{(4)}}{p_{X_s}'}\right]_{-1}
-\frac{1}{144}\left[\left(\frac{p_{X_s}^{(3)}}{p_{X_s}'}\right)^2\right]_{-1}\right\}
&\frac{1}{80}\left\{-\frac{11}{9}\left[\left(\frac{p_{X_s}^{(3)}}{p_{X_s}'}\right)^2\right]_{-1}
+2\left[\frac{p_{X_s}^{(4)}}{p_{X_s}'}\right]_{-1}\right\}
&\frac{1}{20}\left[\frac{p_{X_s}^{(3)}}{p_{X_s}'}\right]_{-1}
&\frac{3}{50}
&\bar{u}-\frac{13}{10}
\end{pmatrix}
\begin{pmatrix}
a\\
b\\
c\\
d\\
e
\end{pmatrix}
=0\:.
\end{align*}
Here by  $\left[\:\right]_{-1}$ we mean to say that we take the coefficient of the order $X^{-1}$ term only.
The determinant is
\begin{align*}
\left\{\left(\bar{u}-\frac{13}{10}\right)\left(\bar{u}-\frac{1}{2}\right)+\frac{3}{25}\right\}\:
\det
 \begin{pmatrix}
\bar{u}&-2&0\\
-\frac{7c}{80}&\bar{u}-\frac{9}{5}&0\\
-\frac{7c}{240}\left[\frac{p^{(3)}_{X_s}}{p'_{X_s}}\right]&-\frac{11}{30}\left[\frac{p^{(3)}_{X_s}}{p'_{X_s}}\right]_{-1}&\bar{u}-\frac{7}{10}\\
\end{pmatrix}
\end{align*}
where the first factor vanishes for $\bar{u}=\frac{7}{10}$ and $\bar{u}=\frac{9}{10}$.

The matrix $A$ in the eigenvalue equation 
$\bar{u}\begin{pmatrix}
a\\
b\\
c\\
\cdot
\end{pmatrix}
=A\begin{pmatrix}
a\\
b\\
c\\
\cdot
\end{pmatrix}$ 
is
...
and the Jordan normal form reads
\begin{align*}
\begin{pmatrix}
\frac{7}{10}&0&0&0\\
0& \frac{7}{10}&0&0\\
0&0&\frac{11}{10}&*\\
0&0&0&*
\end{pmatrix}
\end{align*}
\end{cor}

\begin{proof}
 Only the Jordan normal form of $A$ remains to be proved.
\begin{align*}
0
=\det(A-\lambda)
\=\ldots\det
\begin{pmatrix}
-\lambda&2&0\\
\frac{7c}{80}&\frac{9}{5}-\lambda&0\\
\frac{7c}{240}\left[\frac{p^{(3)}_{X_s}}{p'_{X_s}}\right]_{-1}&\frac{11}{30}\left[\frac{p^{(3)}_{X_s}}{p'_{X_s}}\right]_{-1}&\frac{7}{10}-\lambda
\end{pmatrix}\\
\=\ldots(\frac{7}{10}-\lambda)
\det
\begin{pmatrix}
-\lambda&2\\
\frac{7c}{80}&\frac{9}{5}-\lambda
\end{pmatrix}\\
\=\ldots(\frac{7}{10}-\lambda)\left\{\lambda(\lambda-\frac{9}{5})-\frac{7c}{40}\right\}
\end{align*}
so eigenvalues are
\begin{displaymath}
\lambda_1=\lambda_2=\frac{7}{10}\:,
\quad
\lambda_3=\frac{11}{10}\:,
\quad
\lambda_4=\ldots
\end{displaymath}
To determine the eigenvectors
$\begin{pmatrix}
v_1\\
v_2\\
v_3
\end{pmatrix}$
to $\lambda=\frac{7}{10}$, we consider the system
\begin{align*}
-\frac{7}{10}v_1+2v_2
\=0\\
\frac{7c}{80}v_1+\frac{11}{10}v_2
\=0\\
\frac{7c}{240}v_1
+\frac{11}{30}v_2
\=0\:.
\end{align*}
All three equations are compatible and yield 
$
\vec{v}
=\begin{pmatrix}
20\\
7\\
0
\end{pmatrix} 
$.
Another linearly independent eigenvector is 
$\begin{pmatrix}
0\\
0\\
1
\end{pmatrix}$.
Thus to the double eigenvalue we have two linearly independent eigenvectors.
This proves the claim about the Jordan normal form
(the minor $3\times3$ matrix of $A$ is diagonalisable).
\end{proof}

\begin{remark}
There is no logarithmic solution,
despite the fact that several eigenvalues have an integer difference (equal to zero).

Using $-\frac{c}{24}=\frac{11}{60}$, we find that the eigenvector of $\lambda_3=\frac{11}{10}$
is
$\vec{v}=
v_1\begin{pmatrix}
1\\
\frac{11}{10}\\
\frac{11}{60}\left[\frac{p^{(3)}_{X_s}}{p'_{X_s}}\right]_{-1}
\end{pmatrix}
$ with $v_1\in\C^*$. 
\end{remark}

\begin{remark}
We have $\langle\vartheta_x\rangle=\frac{1}{4}\Theta(x)$. 
For $n=5$, 
\begin{displaymath}
\langle\vartheta^{(3)}_{X_s}\rangle=\frac{3!}{4}\text{A}_0
=-\frac{9c}{2}\1
\end{displaymath}
($a_0=1$),
so by the differential eq.\ (\ref{ODE for 0-pt function}),
\begin{align*}
\left(d_{X_s}-\frac{c}{8}\omega_s\right)\langle\vartheta^{(3)}_{X_s}\rangle
\cong& 
-9c\xi_s\frac{\langle\vartheta_{X_s}\rangle}{p'_{X_s}}
\end{align*} 
\end{remark}

\pagebreak

\pagebreak

\section{General results}

We consider the hyperelliptic Riemann surface
\begin{displaymath}
\Sigma_g: y^2=p(x)\:,\quad\deg p=2g+1,2g+2\:.
\end{displaymath}
with branch points $X_1,X_2,\ldots,X_{2g+1},X_{2g+2}$ (where $X_{2g+2}$ may be the point at infinity).

\subsection{Branch points as primary (twist) fields}

Twist fields are a way to make the dependence of $N$-point functions on the position of the branch points explicit.
If $\langle\varphi(x)\ldots\rangle_{X_1,X_2,\ldots,X_{2g+1},X_{2g+2}}$ denotes an $N$-point function on $\Sigma_g$, we can write
\begin{displaymath}
\langle\varphi(x)\ldots\sheaf{T}(X_1)\sheaf{T}(X_2)\ldots\rangle_{X_3,\ldots,X_{2g+1},X_{2g+2}} 
:=\langle\varphi(x)\ldots\rangle_{X_1,X_2,\ldots,X_{2g+1},X_{2g+2}}\:. 
\end{displaymath}
As $X_1\rechts X_2$, the ramification at $X_2$ is dissolved, and the surface $\Sigma_g$ degenerates to a surface of genus $g-1$, 
with $2g$ branch points $X_3,\ldots,X_{2g+1},X_{2g+2}$. 
For $X_1\approx X_2$, we have an expansion
\begin{align}\label{OPE of twist fields}
\langle\varphi(x)\ldots\sheaf{T}(X_1)\sheaf{T}(X_2)\ldots\rangle_{X_3,\ldots}
=\sum_k(X_1-X_2)^k\langle\varphi(x)\ldots\chi_k^+(X_2)\chi_k^-(X_2)\ldots\rangle_{X_3,\ldots}
\end{align}
where $\chi_k^+$ and $\chi_k^-$ are primary fields corresponding to the two different sheets. They don't depend on $X_1$
and so the $N+2$ point functions on the r.h.s.\ of eq.\ (\ref{OPE of twist fields}) are defined on the degenerate (genus $g-1$) hyperelliptic surface.
The range of $k$ remains to be specified. If $k\geq 0$ and $k=0$ occurs, then
\begin{align}\label{two point function of the primary field associated to the twist fields}
\lim_{X_1\rechts X_2}\langle\sheaf{T}(X_1)\sheaf{T}(X_2)\rangle_{X_3,\ldots}
=\lim_{X_1\rechts X_2}\1_{X_1,X_2,\ldots}
=\langle\chi_0^+(X_2^+)\chi_0^-(X_2^-)\rangle_{X_3,\ldots}\:. 
\end{align}

\begin{remark}
\begin{enumerate}
\item 
We cannot make a statement about $k$ because we are working with a singular metric which affects the power of $(X_1-X_2)$.
\item
We actually have two types of pairs of fields $\chi^+\otimes\chi^-$, namely the field $1_+\otimes1_-$ for $h=0$ 
and some other pair $\varphi_+\otimes\varphi_-$ for $h=-\frac{1}{5}$. 
Here $h=-\frac{1}{5}=\bar{h}$, and somehow $\frac{11}{10}-\frac{7}{10}=-(h+\bar{h})$. 
\item
A problem is that we don't have the global partition function $\sum_{\text{$i$ Fibonacci}}|Z_i|^2$, 
but the $Z_i$ are only defined up to unitary transformation (monodromy). Going around one ramification point gives a factor of $e^{2\pi\i\frac{11}{10}}$,
going around the other $e^{2\pi\i\frac{7}{10}}$.
\end{enumerate}

\end{remark}

Since $\chi^{\pm}$ is a primary field, it has the OPE with the 
\begin{displaymath}
T(x)\otimes\chi^{\pm}(X)
\mapsto\frac{h_{\chi^{\pm}}}{(x-X)^2}\:\chi^{\pm}(X)
+\frac{1}{x-X}\:[\chi^{\pm}]'(X)+\reg 
\end{displaymath}
with the Virasoro field. Letting 
\begin{displaymath}
p(x)
=:(x-X_1)(x-X_2)\tp(x)
\end{displaymath}
we have by eq.\ (\ref{OPE of twist fields}) 
\begin{align}
\frac{1}{(x-X_2)^2\tp}\langle\vartheta(x)\ldots\chi_n^+(X_2^+)\chi_n^-(X_2^-)\ldots\rangle_{X_3,\ldots}
\=\frac{h_{\chi}}{(x-X_2)^2}\:\langle\ldots\chi_n^+(X_2^+)\chi_n^-(X_2^-)\ldots\rangle_{X_3,\ldots}\nn\\
\+\frac{1}{x-X_2}\:\langle\ldots(\chi_n^{\pm})'(X_2)\:\chi_n^{\mp}(X_2)\ldots\rangle_{X_3,\ldots}+\reg\nn\\
\-\frac{c}{32}\left[\frac{p'}{p}\right]^2\:\langle\ldots\chi_n^+(X_2^+)\chi_n^-(X_2^-)\ldots\rangle_{X_3,\ldots}
\label{Laurent coefficient of OPE of twistfields, with a vartheta inside}
\end{align}
Here 
\begin{displaymath}
\left[\frac{p'}{p}\right]^2
=\frac{4}{(x-X_2)^2}+\frac{4}{x-X_2}\frac{\tp'}{\tp}+\left[\frac{\tp'}{\tp}\right]^2+O(X_1-X_2). 
\end{displaymath}
Thus (\ref{Laurent coefficient of OPE of twistfields, with a vartheta inside}) reads
\begin{align*}
\langle\vartheta(x)\ldots\chi_n^+(X_2^+)\chi_n^-(X_2^-)\ldots\rangle_{X_3,\ldots}
\=\left(h_{\chi}-\frac{c}{8}\right)\tp\:\langle\ldots\chi_n^+(X_2^+)\chi_n^-(X_2^-)\ldots\rangle_{X_3,\ldots}
+O(x-X_2)\:.
\end{align*}
By eq.\ (\ref{two point function of the primary field associated to the twist fields}),
in absence of fields other than $\vartheta$ resp.\ $1$,
(\ref{Laurent coefficient of OPE of twistfields, with a vartheta inside}) reads,
\begin{align*}
\left(h_{\chi}-\frac{c}{8}\right)\lim_{X_1\rechts X_2}\1_{X_1,X_2,\ldots}
\=\lim_{X_1\rechts X_2}\frac{\langle\vartheta(x)\rangle_{X_1,X_2,\ldots}}{\tp(x)}\:.
\end{align*}
We may evaluate at $x=X_1$ since $\vartheta(x)\otimes\sheaf{T}(X_1)$ is non-singular ($\langle\vartheta(x)\rangle_{X_1,X_2,\ldots}$ is a polynomial in $x$),
provided $X_1$ is finite,
and use $p'(X_1)=(X_1-X_2)\tp(X_1)$:
\begin{align}
\left(h_{\chi}-\frac{c}{8}\right)\lim_{X_1\rechts X_2}\1_{X_1,X_2,\ldots}
\=\lim_{X_1\rechts X_2}(X_1-X_2)\frac{\langle\vartheta(X_1)\rangle_{X_1,X_2,\ldots}}{p'(X_1)}\:.\label{DE derived from twist fields}
\end{align}
(Note that this makes sense since both sides are $\sim(X_1-X_2)^u$ this way.)

We can compare eq.\ (\ref{DE derived from twist fields}) with the ODE (\ref{ODE for 0-pt function}):
If eq.\ (\ref{DE derived from twist fields}) holds before the limit is taken, 
for $X_1\approx X_2$
we must have
\begin{displaymath}
\bar{u}
=2\left(h_{\chi_{\pm}}-\frac{c}{8}\right)
\end{displaymath}
Moreover, in the minimal model with $c=-\frac{22}{5}$, this reproduces
\begin{align*}
\bar{u}
=2\left(h_{\chi_{\pm}}+\frac{11}{20}\right)
=
\begin{cases}
\frac{11}{10}&\quad h_{\chi_{\pm}}=0 \\
\frac{7}{10}&\quad h_{\chi_{\pm}}=-\frac{1}{5}
\end{cases}
\end{align*}
which is correct.

\subsection{The number of ODEs in general}

Let $g=1$.
We introduce a non-holomorphic (physical) field $\varphi$ with the OPE
\begin{displaymath}
T(z)\varphi(u,\bar{u}) 
\:\mapsto\:
\frac{-1/5}{(z-u)^2}\varphi
+\frac{1}{z-u}\partial\varphi
+\ldots\partial^2\varphi
+O(z-u)\:.
\end{displaymath}
By comparison,
\begin{displaymath}
\langle T(z)\varphi(u,\bar{u})\rangle
=-\frac{1}{5}\wp(z-u)\langle\varphi(u,\bar{u})\rangle\:.
\end{displaymath}
(as there is no periodic function with a first order pole).
We have
\begin{displaymath}
\int_0^1 \wp(z-u|\tau)\:dz
=-\frac{\pi^2}{3}E_2(\tau)\:.
\end{displaymath}
So when the contour integal is taken along the real period and $\oint dz=1$ then
\begin{displaymath}
\partial_{\tau}\langle\varphi(u,\bar{u})\rangle
=\frac{1}{2\pi i}\oint\langle T(z)\varphi(u,\bar{u})\rangle\:dz
\sim\:-\frac{1}{5}E_2\langle\varphi(u,\bar{u})\rangle
\end{displaymath}
so $\langle\varphi(u,\bar{u})\rangle\sim\eta^{-2/5}$.
This gives the solution for the $h=\bar{h}=-\frac{1}{5}$ conformal block,
\begin{displaymath}
\langle\varphi(u,\bar{u})\rangle
\sim\eta^{-2/5} 
\end{displaymath}
For every node (singularity) between two tori we can introduce a field $1$ or $\varphi$.
For $g=2$, there are two tori connected by one node, and we have
\begin{center}
\begin{tabular}{ccc}
node&number of choices&solutions\\
\hline
$\1$&$2$&Rogers-Ramanujan functions\\
$\langle\varphi\rangle$&$1$&$\eta^{-2/5}$\\
\end{tabular} 
\end{center}
For $g=3$ there are three tori (I-III) connected by two nodes. 
Only the middle torus (II) has two marked points, and inserting a field on either node may give rise to a $2$-point function.
We obtain
\begin{center}
\begin{tabular}{cccccc}
torus I&node 1&torus II&node 2&torus III&number of choices\\
\hline
$\1$&$1$&$\1$&$1$&$\1$&$2^3$\\
$\1$&$1$&$\langle\varphi\rangle$&$\varphi$&$\langle\varphi\rangle$&$2$\\
$\langle\varphi\rangle$&$\varphi$&$\langle\varphi\rangle$&$1$&$\1$&$2$\\
$\langle\varphi\rangle$&$\varphi$&$\langle\varphi\varphi\rangle$&$\varphi$&$\langle\varphi\rangle$&$3$\\
\hline
\end{tabular} \\
For $g=3$, we must have an equation of order $15$ for $\1$.
\end{center}
We need to explain the $3$ choices for $\langle\varphi\varphi\rangle$.
Consider the torus II with two marked points.
It is obtained by sqeezing a genus $g=2$ surface.
On the torus we have a choice between the partition functions only,
while on the $g=2$ surface we have $5$.
$\langle\varphi\varphi\rangle$ must make up for this difference.

\pagebreak

\appendix

\section{Proof of Theorem \ref{theorem: graph rep for N-pt fct of vartheta}}\label{proof: theorem: graph rep for N-pt fct of vartheta}

Notation: Let $A\sqcup B$ denote the union of sets $A,B$ with $A\cap B=\emptyset$.

Let $\sheaf{F}$ be the bundle of holomorphic fields.
Let $\sheaf{T}\subset\sheaf{F}$ be the subbundle with fiber $\C T$,
and let $\sheaf{T}_+=\eps\oplus\sheaf{T}\subset\sheaf{F}$,
where $\eps$ is the trivial bundle. 
For $N\geq 1$,
let $I_N:=\{x_1,\ldots,x_N\}$, 
and $\mathcal{P}_N:=\mathcal{P}(I_N)$ be the powerset of $I_N$.
For $I\in\mathcal{P}_N$, let $\Graph(I)$ be the set of admissible graphs whose vertices are the points of $I$,
and let $\Graph_N=\Graph(I_N)$.
For any $N\geq 0$, we consider the map 
\begin{align*}
w:\sqcup_{N\geq 0}\sheaf{T}_+^{\boxtimes N}
\rechts\sqcup_{N\geq 0}\sheaf{T}^{\boxtimes N}
\end{align*}
defined as follows: 
For $\varphi\in\sheaf{T}_+^{\boxtimes N}$ over $(x_1,\ldots,x_N)\in U^N\subset\Sigma^N\setminus\Delta_N$
(symmetrised product) with $x\not=x_i$ $\forall\:i$,
\begin{align*}
w(1_x\times_s\varphi)
\=w(\varphi)\:.
\end{align*}
For $(x_1,\ldots,x_N)$ as above,
and $\prod_{i=1}^NT(x_i)p(x_i)\in\Gamma(U^N,\sheaf{T}^{\boxtimes N})$,
\begin{align*}
w\left(\prod_{i=1}^NT(x_i)p(x_i)\right)
\=\sum_{\Gamma\in\Graph_N}\tw\left(\Gamma_N,\prod_{i=1}^NT(x_i)p(x_i)\right)\:,
\end{align*}
where
\begin{align*}
\tw\left(\Gamma_N,\prod_{i=1}^NT(x_i)p(x_i)\right)
=\left(\frac{c}{2}\right)^{\sharp\text{loops}}
\prod_{(x_i,x_j)\in\Gamma}\left(\frac{1}{4}\:f(x_i,x_j)\right)
\:\bigotimes_{k\in A_N\cap {E_N}^c}\vartheta_k
\bigotimes_{\ell\in(A_N\cup E_N)^c}T(x_{\ell})p_{\ell}\:.
\end{align*}
By the theorem about the graphical representation of $\langle T\ldots T\rangle\:p\ldots p$,
$w$ is such that
\begin{displaymath}
\langle\:\rangle=\langle\:\rangle_r\circ w\: 
\end{displaymath}
on $\sheaf{T}_+^{\boxtimes N}$.
Note that for $I\in\mathcal{P}_N$, $\Gamma_I\in\Graph(I)$,
\begin{align*}
\tw\left(\Gamma_I\sqcup\Gamma_{I^c},\prod_{i=1}^NT(x_i)p(x_i)\right)
=\tw\left(\Gamma_I,\prod_{i\in I}T(x_i)p(x_i)\right)
\cdot\tw\left(\Gamma_{I^c},\prod_{i\in I^c}T(x_i)p(x_i)\right)\:.
\end{align*}
Here $I^c=I_N\setminus I$.

Since both $\langle\:\rangle$ and $\langle\:\rangle_r$ are linear,
We also have for $\varphi,\psi\in\sheaf{T}_+^{\boxtimes N}$,
\begin{align*}
w\left(\sum_{I\in\mathcal{P}_N}\prod_{x\in I}\varphi(x)\prod_{x\in I_N\setminus I}\psi(x)\right)
=\sum_{I\in\mathcal{P}_N}w\left(\prod_{x\in I}\varphi(x)\prod_{x\in I_N\setminus I}\psi(x)\right)\:.
\end{align*}
Now $\vartheta\in\sheaf{T}_+$.
For $P=-\frac{c}{32}\frac{[p']^2}{p}.1$,
\begin{align*}
w\left(\prod_{j=1}^N\vartheta_j\right) 
\=w\left(\prod_{x\in\mathcal{P}_N}(T(x)p_x+P(x))\right)\\
\=w\left(\sum_{I\in\mathcal{P}_N}\prod_{x\in I}T(x)p_x\prod_{x\in I_N\setminus I}P(x)\right)\\
\=\sum_{I\in\mathcal{P}_N}\left(\prod_{x\in I_N\setminus I}P(x)\right)\cdot 
w\left(\prod_{x\in I}T(x)p_x\right)\\
\=\sum_{I\in\mathcal{P}_N}\left(\prod_{x\in I_N\setminus I}P(x)\right)\sum_{\Gamma\in\Graph(I)}
\tw\left(\Gamma,\prod_{x\in I}T(x)p_x\right)\\
\end{align*}
Let $\ess$ be the projection
\begin{align*}
\ess:\cup_{I\in\mathcal{P}_N}\Graph(I)\rechts\mathcal{P}_N
\end{align*}
which assigns to a graph its (essential support consisting of its) set of non-isolated vertices. 
For $I\in\mathcal{P}_N$, 
let $Is(\Gamma):=(A_{\Gamma}\cup E_{\Gamma})^c$, the set of isolated points of $\Gamma\in\Graph(I)$.
Let $\Gamma_0(x)$ be the graph consisting of the point $x$ (with no links). 
Every graph $\Gamma$ can be written as
\begin{align*}
\Gamma
=\Gammared\sqcup\left(\cup_{x\in\ess(\Gamma)^c}\Gamma_0(x)\right) 
=\Gammared\sqcup\left(\cup_{x\in\Is(\Gamma)}\Gamma_0(x)\right) 
\end{align*}
(disjoint unions).
By the previous computation, 
\begin{align*}
w\left(\prod_{j=1}^N\vartheta_j\right)
\=\sum_{I\in\mathcal{P}_N}\left(\prod_{x\in I_N\setminus I}P(x)\right)
\sum_{\Gamma\in\Graph(I)}\tw\left(\Gamma,\prod_{x\in I}T(x)p_x\right)\\
\=\sum_{I\in\mathcal{P}_N}
\left(\prod_{x\in I_N\setminus I}P(x)\right)
\sum_{\Gamma\in\Graph(I)}
\tw\left(\Gamma,\prod_{x\in\ess(\Gamma)}T(x)p_x\prod_{x\in\Is(\Gamma)}T(x)p_x\right)\\
\=\sum_{I\in\mathcal{P}_N}
\left(\prod_{x I_N\setminus I}P(x)\right)
\sum_{\Gamma\in\Graph(I)}
\tw\left(\Gammared,\prod_{x\in\ess(\Gammared)}T(x)p_x\right)
\left(\prod_{x\in\Is(\Gamma)}T(x)p_x\right)\\
\=\sum_{\tI\in\mathcal{P}_N}
\sum_{\Gammared\in\Graph(\tI)}
\tw\left(\Gammared,\prod_{x\in\ess(\Gammared)}T(x)p_x\right)
\underset{I\supset\tI}{\sum_{I\in\mathcal{P}_N}}
\left(\prod_{x\in I\setminus\tI}T(x)p_x\right)
\left(\prod_{x\in I_N\setminus I}P(x)\right)\\
\=\sum_{\tI\in\mathcal{P}_N}\sum_{\Gammared\in\Graph(\tI)}
\tw\left(\Gammared,\prod_{x\in\ess(\Gammared)}T(x)p_x\right)
\prod_{x\in I_N\setminus\tI}\left(T(x)p_x+P(x)\right)\\
\=\sum_{\tI\in\mathcal{P}_N}\sum_{\Gammared\in\Graph(\tI)}
\tw\left(\Gammared,\prod_{x\in\ess(\Gammared)}T(x)p_x\right)
\prod_{x\in I_N\setminus\tI}\vartheta_x
\\
\=\sum_{\tI\in\mathcal{P}_N}
\sum_{\Gammared\in\Graph(\tI)}
\left(\frac{c}{2}\right)^{\sharp\text{loops of $\Gammared$}}
\prod_{(x_i,x_j)\in\Gammared}\left(\frac{1}{4}\:f(x_i,x_j)\right)
\:\bigotimes_{k\in A_{\Gammared}\cap {E_{\Gammared}}^c}\vartheta_k
\prod_{x\in I_N\setminus\tI}\vartheta_x\\
\=\sum_{\Gamma\in\Graph_N}
\left(\frac{c}{2}\right)^{\sharp\text{loops of $\Gamma$}}
\prod_{(x_i,x_j)\in\Gamma}\left(\frac{1}{4}\:f(x_i,x_j)\right)
\:\bigotimes_{k\in {E_{\Gamma}}^c}\vartheta_k
\:,
\end{align*}
and application of $\langle\:\rangle_r$ yields the claimed formula.

\section{Sketch of the proof of Lemma \ref{lemma: differential eq. for the N-pt function of vartheta}}
\label{proof: lemma: differential eq. for the N-pt function of vartheta}

By induction.
Sketch of the argument:
From eq.\ (\ref{definition of vartheta}) follows
\begin{align*}
d\langle\vartheta\rangle
\=
p\:\left\{d\langle T\rangle-\frac{c}{32}\left(\frac{p'}{p}\right)^2d\1\right\}
+\langle\vartheta\rangle\frac{dp}{p}
-\frac{c}{16}\1 pd\left(\frac{p'}{p}\right)
\:,
\end{align*}
where by eq.\ (\ref{differential eq. for N-point function of T}) for $N=0,1$,
\begin{align*}
p\left\{d\langle T\rangle-\frac{c}{32}\left(\frac{p'}{p}\right)^2d\1\right\}
\=2p\sum_{s=1}^n
\frac{\xi_s}{p'_{X_s}}\:\left\{\langle\vartheta_{X_s}T\rangle-\frac{c}{32}\left(\frac{p'}{p}\right)^2\langle\vartheta_{X_s}\rangle\right\}
+\frac{c}{8}\:\omega\left\{\langle T\rangle-\frac{c}{32}\left(\frac{p'}{p}\right)^2\1\right\}\\
\=2\sum_{s=1}^n
\frac{\xi_s}{p'_{X_s}}\:\langle\vartheta_{X_s}\vartheta\rangle+\frac{c}{8}\:\omega\langle\vartheta\rangle
\:.
\end{align*}
For $N=2$,
\begin{align*}
d\langle\vartheta_1\vartheta_2\rangle
\=p_1p_2\left\{\:d\langle T_1T_2\rangle-\frac{c^2}{(32)^2}\frac{[p_1'p_2']^2}{(p_1p_2)^2}d\1\right\}\\
\+\langle\vartheta_1\vartheta_2\rangle\frac{d(p_1p_2)}{p_1p_2}
-\frac{c}{16}\langle\vartheta_1\rangle p_1p_2'd\left(\frac{p_2'}{p_2}\right)
-\frac{c}{16}\langle\vartheta_2\rangle p_1'p_2d\left(\frac{p_1'}{p_1}\right)
-\frac{c}{32}\frac{[p_2']^2}{p_2}d\langle\vartheta_1\rangle
-\frac{c}{32}\frac{[p_1']^2}{p_1}d\langle\vartheta_2\rangle
\:.
\end{align*}
On the other hand, in eq.\ (\ref{differential eq. for N-point function of T}) for $N=0,2$
\begin{align*}
p_1p_2&\sum_{s=1}^n
\frac{\xi_s}{p'_{X_s}}\Big\{\langle\vartheta_{X_s}T_1T_2\rangle
-\frac{c^2}{(32)^2}\left(\frac{p_1'p_2'}{p_1p_2}\right)^2\langle\vartheta_{X_s}\rangle\Big\}\\
\=\sum_{s=1}^n
\frac{\xi_s}{p'_{X_s}}\langle\vartheta_{X_s}\vartheta_1\vartheta_2\rangle
+\frac{c}{32}\frac{[p_1']^2}{p_1}\sum_{s=1}^n
\frac{\xi_s}{p'_{X_s}}\langle\vartheta_{X_s}\vartheta_2\rangle
+\frac{c}{32}\frac{[p_2']^2}{p_2}\sum_{s=1}^n
\frac{\xi_s}{p'_{X_s}}\langle\vartheta_{X_s}\vartheta_1\rangle
\end{align*}
The last two terms drop out as $-\frac{c}{32}\frac{[p_2']^2}{p_2}d\langle\vartheta_1\rangle
-\frac{c}{32}\frac{[p_1']^2}{p_1}d\langle\vartheta_2\rangle$ are added.

\section{Proof of Lemma \ref{Lemma: The first derivative of the Galois components of Psi at Xs}}
\label{proof: lemma: The first derivative of the Galois components of Psi at Xs}

Application of $\langle\quad\rangle$ to eq.\ (\ref{eq.: graphical representation of vartheta's}) of $\vartheta_x\otimes\vartheta_{X_s}$ yields an identity of states 
\begin{align}\label{eq: starting point for computations with psi}
\langle\psi_1\rangle+\langle\psi_2\rangle
\=2\langle\vartheta_1\vartheta_2\rangle-2\{\text{$f_{12}$-terms}\}+O((x_1-x_2)^2)\:.
\end{align}
where
\begin{align*}
\langle\vartheta_1\vartheta_2\rangle
\=\langle\vartheta^{[1]}_1\vartheta^{[1]}_2\rangle
+y_1y_2\langle\vartheta^{[y]}_1\vartheta^{[y]}_2\rangle
+y_1\langle\vartheta^{[y]}_1\vartheta^{[1]}_2\rangle
+y_2\langle\vartheta^{[1]}_1\vartheta^{[y]}_2\rangle
\:,
\end{align*}
and
\begin{align*}
\{\text{$f_{12}$-terms}\}
=\frac{c}{32}f_{12}^2\1
+\frac{1}{4}f_{12}\left\{\langle\vartheta^{[1]}_1\rangle+\langle\vartheta^{[1]}_2\rangle\right\}
+\frac{1}{4}f_{12}\left\{y_1\langle\vartheta^{[y]}_1\rangle+y_2\langle\vartheta^{[y]}_2\rangle\right\}\:.
\end{align*}
In the $(2,5)$ minimal model, 
\begin{align*}
\langle\psi_2\rangle
\=\lim_{x_1\rechts x_2}\left[\langle\vartheta_1\vartheta_2\rangle\right]_{\text{no pole}}\\
\=\lim_{x_1\rechts x_2}\left[\langle\vartheta_1\vartheta_2\rangle-\{\text{$f_{12}$-terms}\}\right]\\
\=\lim_{x_1\rechts x_2}\langle\vartheta_1\vartheta_2\rangle_r\\
\=\langle\vartheta_2\vartheta_2\rangle_r\\
\=\langle\vartheta^{[1]}_2\vartheta^{[1]}_2\rangle_r
+p_2\langle\vartheta^{[y]}_2\vartheta^{[y]}_2\rangle_r
+2y_2\langle\vartheta^{[y]}_2\vartheta^{[1]}_2\rangle_r\:.
\end{align*}
is known. $\langle\psi\rangle$ has a Galois splitting
\begin{align*}
\langle\psi\rangle
=\langle\psi^{[1]}\rangle+y\langle\psi^{[y]}\rangle\:,
\end{align*}
so $\langle\psi^{[1]}\rangle$ and $\langle\psi^{[y]}\rangle$ are known, where
\begin{align*}
\langle\psi^{[1]}_2\rangle
\=\langle\vartheta^{[1]}_2\vartheta^{[1]}_2\rangle_r
+p_2\langle\vartheta^{[y]}_2\vartheta^{[y]}_2\rangle_r\\
\langle\psi^{[y]}_2\rangle
\=2\langle\vartheta^{[1]}_2\vartheta^{[y]}_2\rangle_r\:.
\end{align*}
It follows that also
\begin{align*}
\partial_{x_2}\langle\psi^{[1]}_2\rangle
\=\partial_{x_2}\langle\vartheta^{[1]}_2\vartheta^{[1]}_2\rangle_r
+p'_2\langle\vartheta^{[y]}_2\vartheta^{[y]}_2\rangle_r
+p_2\:\partial_{x_2}\langle\vartheta^{[y]}_2\vartheta^{[y]}_2\rangle_r
\:,
\\
\partial_{x_2}\langle\psi^{[y]}_2\rangle
\=2\partial_{x_2}\langle\vartheta^{[1]}_2\vartheta^{[y]}_2\rangle_r
\end{align*}
are known, and thus $\langle(\psi^{[1]})'_{X_s}\rangle$ and $\langle(\psi^{[y]})'_{X_s}\rangle$.
To be specific, we go back to eq.\ (\ref{eq: starting point for computations with psi}).
Thus
\begin{align*}
\partial_1\langle\psi^{[1]}_1\rangle+\partial_2\langle\psi^{[1]}_2\rangle
\=\left(\partial_{x_1}+\partial_{x_2}\right)
\Big(\langle\vartheta^{[1]}_1\vartheta^{[1]}_2\rangle+y_1y_2\langle\vartheta^{[y]}_1\vartheta^{[y]}_2\rangle\\
&\hspace{2cm}-\left\{\frac{c}{32}f_{12}^2\1
+\frac{1}{4}f_{12}\left\{\langle\vartheta^{[1]}_1\rangle+\langle\vartheta^{[1]}_2\rangle\right\}\right\}\Big)
+O(x_1-x_2)\:,
\end{align*}
and
\begin{align*}
\partial_{x_2}&\langle\psi^{[1]}_2\rangle\\
\=\frac{1}{2}\lim_{x_1\rechts x_2}
\Big[\langle(\vartheta^{[1]}_1)'\vartheta^{[1]}_2\rangle+\langle\vartheta^{[1]}_1(\vartheta^{[1]}_2)'\rangle
+\frac{1}{2}\left\{\frac{p'_1}{p_1}+\frac{p'_2}{p_2}\right\}y_1y_2\langle\vartheta^{[y]}_1\vartheta^{[y]}_2\rangle
+y_1y_2\langle(\vartheta^{[y]}_1)'\vartheta^{[y]}_2\rangle+y_1y_2\langle\vartheta^{[y]}_1(\vartheta^{[y]}_2)'\rangle\\
&\hspace{1cm}
-\left(\partial_{x_1}+\partial_{x_2}\right)\left\{\frac{c}{32}f_{12}^2\1
+\frac{1}{4}f_{12}\left\{\langle\vartheta^{[1]}_1\rangle+\langle\vartheta^{[1]}_2\rangle\right\}\right\}
\Big]\\
\=\frac{1}{2}\lim_{x_1\rechts x_2}
\Big[2\langle\vartheta^{[1]}_1(\vartheta^{[1]}_2)'\rangle
+p'_2\langle\vartheta^{[y]}_1\vartheta^{[y]}_2\rangle
+2p_2\langle\vartheta^{[y]}_1(\vartheta^{[y]}_2)'\rangle
-\left(\partial_{x_1}+\partial_{x_2}\right)\left\{\frac{c}{32}f_{12}^2\1
+\frac{1}{4}f_{12}\left\{\langle\vartheta^{[1]}_1\rangle+\langle\vartheta^{[1]}_2\rangle\right\}\right\}
\Big]
\:.
\end{align*}
We conclude that
\begin{align*}
\alignedbox{\partial_x|_{X_s}\langle\psi^{[1]}_2\rangle} 
{=\langle\vartheta_{X_s}(\vartheta^{[1]})'_{X_s}\rangle_r
+\frac{1}{2}p'_{X_s}\langle\vartheta^{[y]}_{X_s}\vartheta^{[y]}_{X_s}\rangle_r}\:.
\end{align*}
Likewise,
\begin{align*}
\partial_1\langle\psi^{[y]}_1\rangle+\partial_2\langle\psi^{[y]}_2\rangle
\=\partial_{x_2}\langle\vartheta^{[1]}_1\vartheta^{[y]}_2\rangle+\partial_{x_1}\langle\vartheta^{[y]}_1\vartheta^{[1]}_2\rangle
-\left(\partial_{x_1}+\partial_{x_2}\right)
\left\{\frac{1}{4}f_{12}\left\{\langle\vartheta^{[y]}_1\rangle+\langle\vartheta^{[y]}_2\rangle\right\}\right\}
+O(x_1-x_2)\:,
\end{align*}
so
\begin{align*}
\partial_2\langle\psi^{[y]}_2\rangle
\=\frac{1}{2}\lim_{x_1\rechts x_2}
\left[
2\langle\vartheta^{[1]}_1(\vartheta^{[y]}_2)'\rangle
-\left(\partial_{x_1}+\partial_{x_2}\right)
\left\{\frac{1}{4}f_{12}\left\{\langle\vartheta^{[y]}_1\rangle+\langle\vartheta^{[y]}_2\rangle\right\}\right\}\right]\:,
\end{align*}
whence
\begin{align*}
\alignedbox{\partial_x|_{X_s}\langle\psi^{[y]}_x\rangle} 
{=\langle\vartheta_{X_s}(\vartheta^{[y]})'_{X_s}\rangle_r}\:,
\end{align*}
as required.

\section{Proof of Claim \ref{claim: regular part of <vartheta Xs vartheta Xs>}}
\label{proof: claim: no pole part of <vartheta Xs vartheta x>, to leading order}

We compute the expressions given by eqs (\ref{eq: regular part of <vartheta(Xs) 1-part of vartheta>})
and (\ref{eq: regular part of <vartheta(Xs) y-part of vartheta>}), 
at $x=X_s$, up to order $(x-X_s)^3$ terms.
\begin{enumerate}
\item 
We first address eq.\ (\ref{eq: regular part of <vartheta(Xs) 1-part of vartheta>}),
\begin{displaymath}
\Big[\langle\vartheta_{X_s}\vartheta^{[1]}_x\rangle\Big]_{\reg}
=\left[\frac{c}{32}f_{X_sx}^2\1+\frac{1}{4}f_{X_sx}\left\{\langle\vartheta_{X_s}\rangle+\langle\vartheta^{[1]}_x\rangle\right\}\right]_{\reg}
+\langle\vartheta_{X_s}\vartheta^{[1]}_x\rangle_r\:, 
\end{displaymath}
We have
\begin{align*}
\frac{2}{p'_{X_s}}\frac{c}{32}f_{xX_s}^2
\=\frac{c}{16}
\left(
\frac{p'_{X_s}}{(x-X_s)^2}+\frac{p''_{X_s}}{x-X_s}+\frac{1}{4}\frac{[p''_{X_s}]^2}{p'_{X_s}}+\frac{1}{3}p^{(3)}_{X_s}
\right)\nn\\
\+\frac{c}{16}
\left(
\frac{1}{12}p^{(4)}_{X_s}+\frac{1}{6}\frac{p''_{X_s}}{p'_{X_s}}p^{(3)}_{X_s}
\right)(x-X_s)\nn\\
\+\frac{c}{16}
\left(
\frac{1}{24}\frac{p''_{X_s}}{p'_{X_s}}p^{(4)}_{X_s}+\frac{1}{60}p^{(5)}(X_s)+\frac{1}{36}\frac{[p^{(3)}_{X_s}]^2}{p'_{X_s}}
\right)(x-X_s)^2
+O((x-X_s)^3)\:.
\end{align*}
Thus to leading order,
 \begin{align}
\frac{2}{p'_{X_s}}\left[\frac{c}{32}f_{xX_s}^2\right]_{\reg}
\cong&\:\frac{c}{64}\frac{[p''_{X_s}]^2}{p'_{X_s}}
+\frac{c}{96}\frac{p''_{X_s}}{p'_{X_s}}p^{(3)}_{X_s}(x-X_s)\nn\\
\+\frac{c}{192}
\left(
\frac{1}{2}\frac{p''_{X_s}}{p'_{X_s}}p^{(4)}_{X_s}+\frac{1}{3}\frac{[p^{(3)}_{X_s}]^2}{p'_{X_s}}
\right)(x-X_s)^2
+O((x-X_s)^3)\:.\label{eq. for f(x,Xs) squared over p'(Xs) up to terms of order (x-Xs) to the cube, to leading order}
\end{align}
Now we address $\frac{2}{p'_{X_s}}\frac{1}{4}f_{xX_s}\left\{\vartheta_x+\vartheta_{X_s}\right\}$.
To simplify notations, set
\begin{displaymath}
\vartheta=\vartheta^{[1]}\:. 
\end{displaymath}
Now
\begin{align}
\frac{2}{p'_{X_s}}&\frac{1}{4}f_{xX_s}\left\{\vartheta_x+\vartheta_{X_s}\right\}\nn\\
\=\frac{\vartheta_{X_s}}{x-X_s}
+\frac{1}{2}\frac{p''_{X_s}}{p'_{X_s}}\vartheta_{X_s}
+\frac{1}{6}(x-X_s)\frac{p^{(3)}_{X_s}}{p'_{X_s}}\vartheta_{X_s}
+\frac{1}{24}(x-X_s)^2\frac{p^{(4)}_{X_s}}{p'_{X_s}}\vartheta_{X_s}\nn\\
\+\frac{1}{2}\vartheta'_{X_s}
+\frac{1}{4}(x-X_s)\frac{p''_{X_s}}{p'_{X_s}}\vartheta'_{X_s}
+\frac{1}{12}(x-X_s)^2\frac{p^{(3)}_{X_s}}{p'_{X_s}}\vartheta'_{X_s}\nn\\
\+\frac{1}{4}(x-X_s)\vartheta''_{X_s}
+\frac{1}{8}(x-X_s)^2\frac{p''_{X_s}}{p'_{X_s}}\vartheta''_{X_s}\nn\\
\+\frac{1}{12}(x-X_s)^2\vartheta^{(3)}_{X_s}
+O((x-X_s)^3)\:.\label{eq. for f(x,Xs) times (1-part of vartheta(x)+vartheta(Xs)) over p'(Xs) up to terms of order (x-Xs) to the cube}
\end{align}
or
\begin{align*}
\frac{2}{p'_{X_s}}\frac{1}{4}f_{xX_s}\left\{\vartheta_x+\vartheta_{X_s}\right\}
\=\frac{\vartheta_{X_s}}{x-X_s}\\
\+\frac{1}{2}\left(\frac{p''_{X_s}}{p'_{X_s}}\vartheta_{X_s}+\vartheta'_{X_s}\right)\\
\+\frac{1}{6}\left(\frac{p^{(3)}_{X_s}}{p'_{X_s}}\vartheta_{X_s}
+\frac{3}{2}\frac{p''_{X_s}}{p'_{X_s}}\vartheta'_{X_s}
+\frac{3}{2}\vartheta''_{X_s}\right)(x-X_s)\\
\+\frac{1}{24}\left(\frac{p^{(4)}_{X_s}}{p'_{X_s}}\vartheta_{X_s}
+2\frac{p^{(3)}_{X_s}}{p'_{X_s}}\vartheta'_{X_s}
+3\frac{p''_{X_s}}{p'_{X_s}}\vartheta''_{X_s}
+2\vartheta^{(3)}_{X_s}\right)(x-X_s)^2\\
\+O((x-X_s)^3)\:.
\end{align*}
\begin{remark}
Suppose $\left[\frac{1}{4}f_{xX_s}\vartheta_x\right]_{\reg}$  with $\vartheta=\vartheta^{[1]},\vartheta^{[y]}$ 
is known up to terms in $O((x-X_s)^2)$. This defines a system
\begin{align*}
(p'\vartheta)'_{X_s}=p''_{X_s}\vartheta_{X_s}+p'_{X_s}\vartheta'_{X_s}\hspace{1.6cm}\=*\\
p^{(3)}_{X_s}\vartheta_{X_s}+\frac{3}{2}\left(p''_{X_s}\vartheta'_{X_s}+p'_{X_s}\vartheta''_{X_s}\right)\=*
\end{align*}which is solvable for $\vartheta_{X_s}$ and $\vartheta'_{X_s}$ as functions of $\vartheta''_{X_s}$ iff $S(p)|_{X_s}\not=0$.
(This follows from eq.\ (\ref{eq. for f(x,Xs) times (1-part of vartheta(x)+vartheta(Xs)) over p'(Xs) up to terms of order (x-Xs) to the cube}),
using that $p'_{X_s}\not=0$.)
For instance, for $g=1$, $\langle\vartheta''\rangle$ is constant in position.
\end{remark}
It follows that
\begin{align}
\frac{2}{p'_{X_s}}\left[\frac{1}{4}f_{xX_s}\left\{\vartheta_x+\vartheta_{X_s}\right\}\right]_{\reg}
\cong&\:\frac{1}{2}\frac{p''_{X_s}}{p'_{X_s}}\vartheta_{X_s}\nn\\
\+\left(\frac{1}{6}\frac{p^{(3)}_{X_s}}{p'_{X_s}}\vartheta_{X_s}
+\frac{1}{4}\frac{p''_{X_s}}{p'_{X_s}}\vartheta'_{X_s}\right)(x-X_s)\nn\\
\+\frac{1}{4}\left(\frac{1}{6}\frac{p^{(4)}_{X_s}}{p'_{X_s}}\vartheta_{X_s}
+\frac{1}{3}\frac{p^{(3)}_{X_s}}{p'_{X_s}}\vartheta'_{X_s}
+\frac{1}{2}\frac{p''_{X_s}}{p'_{X_s}}\vartheta''_{X_s}
\right)(x-X_s)^2\nn\\
\+O((x-X_s)^3)\:.\label{eq. for f(x,Xs) times (vartheta(x)+vartheta(Xs)) over p'(Xs) up to terms of order (x-Xs) to the cube, to leading order}
\end{align}
Moreover, for the $(2,5)$ minimal model, $\langle\vartheta_{X_s}\vartheta_{X_s}\rangle_r$ is given by eq.\ (\ref{eq: definition of psi}).
Applying $\langle\:\rangle$ to the previous formulae, multiplying by $\frac{p'_{X_s}}{2}$ and summing up yields the claim with $\vartheta=\vartheta^{[1]}$.
\item
It remains to consider
eq.\ (\ref{eq: regular part of <vartheta(Xs) y-part of vartheta>}), 
\begin{displaymath}
\Big[\langle\vartheta_{X_s}\vartheta^{[y]}_x\rangle\Big]_{\reg}
=\left[\frac{1}{4}f_{X_sx}\langle\vartheta^{[y]}_x\rangle\right]_{\reg}
+\langle\vartheta_{X_s}\vartheta^{[y]}_x\rangle_r\:, 
\end{displaymath}
Here $\left[\frac{1}{4}f_{X_sx}\langle\vartheta^{[y]}_x\rangle\right]_{\reg}$
equals $\frac{1}{2}p'_{X_s}$ times the regular part 
in eq.\ (\ref{eq. for f(x,Xs) times (1-part of vartheta(x)+vartheta(Xs)) over p'(Xs) up to terms of order (x-Xs) to the cube}) 
for $\vartheta^{[y]}$ in place of $\vartheta$,
und $\langle\vartheta_{X_s}\vartheta^{[y]}_{X_s}\rangle_r=\frac{1}{2}\langle\psi^{[y]}_{X_s}\rangle$ is known for the $(2,5)$ minimal model.
\end{enumerate}

\section{Proof of the fourth differential equation when $n=5$}\label{proof: fourth differential equation when $n=5$}

We follow the arguments given by eqs (\ref{eq: general formulation of differential eq. for <k-th derivative of vartheta at Xs>}), 
(\ref{eq: differentiation can be pulled out of the state}) and (\ref{eq: regular part of <vartheta(Xs) 1-part of vartheta>}).
From eq.\ (\ref{eq: first derivative of the regular content of the singular part of the graphical rep of <vartheta(Xs) vartheta(x)>}) follows
\begin{align*}
\frac{2}{p'_{X_s}}&\frac{\partial^2}{\partial x_2^2}
\left[\frac{c}{32}f(X_s,x_2)^2\1+\frac{1}{4}f(X_s,x_2)\left\{\langle\vartheta_{X_s}\rangle+\langle\vartheta_x\rangle\right\}\right]_{\reg}\\
\cong&\:\frac{c}{96}\1
\left(
\frac{1}{2}\frac{p''_{X_s}}{p'_{X_s}}p^{(4)}_{X_s}+\frac{1}{3}\frac{[p^{(3)}_{X_s}]^2}{p'_{X_s}}
\right)\\
\+\frac{1}{12}p^{(4)}_{X_s}\frac{\langle\vartheta_{X_s}\rangle}{p'_{X_s}}
+\frac{1}{6}p^{(3)}_{X_s}\frac{\langle\vartheta'_{X_s}\rangle}{p'_{X_s}}
+\frac{1}{4}p''_{X_s}\frac{\langle\vartheta''_{X_s}\rangle}{p'_{X_s}}\nn\\
\+O(x_2-X_s)\:.
\end{align*}
From eqs (\ref{eq: first derivative of <vartheta(Xs) vartheta>r}) 
and (\ref{eq: first derivative of Psi(x)}) follows
\begin{align}\label{eq: second derivative of <vartheta(Xs) vartheta>r}
\frac{\partial^2}{\partial x^2}|_{X_s}\langle\vartheta_{X_s}\vartheta_x\rangle_r 
=\langle\vartheta_{X_s}\vartheta_{X_s}''\rangle_r
\=\frac{1}{2}\Psi_{X_s}''
-\langle\vartheta_{X_s}'\vartheta_{X_s}'\rangle_r\:,
\end{align}
where
\begin{align}\label{eq: second derivative of Psi(x)}
\langle\psi''_x\rangle
\=\frac{c}{240}[p^{(3)}_x]^2\1
+\frac{c}{480}p''_xp^{(4)}_x\1
-\frac{c}{480}p'_xp^{(5)}_x\1\nn\\
\+\frac{1}{5}p^{(4)}_x\langle\vartheta_x\rangle
+\frac{3}{10}p^{(3)}_x\langle\vartheta'_x\rangle
-\frac{1}{5}p''_x\langle\vartheta''_x\rangle
-\frac{1}{2}p'_x\langle\vartheta^{(3)}_x\rangle
-\frac{1}{5}p_x\langle\vartheta^{(4)}_x\rangle\:,
\end{align}
respectively. (Note that for $g=2$, $\langle\vartheta^{(4)}(x)\rangle=0$.)
Thus according to eq.\ (\ref{eq: regular part of <vartheta(Xs) 1-part of vartheta>}),
\begin{align*}
\frac{2}{p'_{X_s}}\left[\langle\vartheta_{X_s}\vartheta_{X_s}''\rangle\right]_{\reg}
\cong&\:
\left(
\frac{7c}{960}\frac{p''_{X_s}}{p'_{X_s}}p^{(4)}_{X_s}
+\frac{11c}{1440}\frac{[p^{(3)}_{X_s}]^2}{p'_{X_s}}
\right)\1\\
\+\frac{17}{60}p^{(4)}_{X_s}\frac{\langle\vartheta_{X_s}\rangle}{p'_{X_s}}
+\frac{7}{15}p^{(3)}_{X_s}\frac{\langle\vartheta'_{X_s}\rangle}{p'_{X_s}}
+\frac{1}{20}p''_{X_s}\frac{\langle\vartheta''_{X_s}\rangle}{p'_{X_s}}\nn\\
\-\frac{2}{p'_{X_s}}\langle\vartheta'_{X_s}\vartheta'_{X_s}\rangle_r\:.
\end{align*}
Multiplication by $\xi_s$ 
and using eqs (\ref{eq: general formulation of differential eq. for <k-th derivative of vartheta at Xs>}) yields the claim.

\section{Alternative proof of the fourth differential equation when $n=5$}\label{proof: alternative proof of fourth differential equation when $n=5$}

If $X_1=X_2$, then $f(X_s,x_2)=\tp_2$ is regular, 
and
\begin{align*}
\frac{\partial^2}{\partial x_2^2}&
\left[\frac{c}{32}\tp_2^2\1
+\frac{1}{4}\tp_2\left\{\langle\vartheta_{X_s}\rangle+\langle\vartheta_x\rangle\right\}\right]_{\reg}\\
\=\frac{c}{16}\left([\tp_2']^2+\tp_2\tp_2''\right)\1
+\frac{1}{4}\tp_2''\left\{\langle\vartheta_{X_s}\rangle+\langle\vartheta_x\rangle\right\}
+\frac{1}{2}\tp_2'\langle\vartheta'_x\rangle
+\frac{1}{4}\tp_2\langle\vartheta_2''\rangle
\:.
\end{align*}
In addition,
\begin{align*}
\frac{\partial^2}{\partial x^2}|_{X_s}\langle\vartheta_{X_s}\vartheta_x\rangle_r 
=\langle\vartheta_{X_s}\vartheta_{X_s}''\rangle_r
\=\frac{1}{2}\Psi_{X_s}''
-\langle\vartheta_{X_s}'\vartheta_{X_s}'\rangle_r\:,
\end{align*}
where
\begin{align*}
\psi_{X_s}''
\=\frac{3c}{20}[\tp_{X_s}']^2\1
+\frac{c}{20}\tp_{X_s}\tp_{X_s}''\1\\
\+\frac{12}{5}\tp_{X_s}''\langle\vartheta_{X_s}\rangle
+\frac{9}{5}\tp_{X_s}'\langle\vartheta_{X_s}'\rangle
-\frac{2}{5}\tp_{X_s}\langle\vartheta_{X_s}''\rangle\:. 
\end{align*}
So
\begin{align*}
\frac{2}{p'_{X_s}}\left[\langle\vartheta_{X_s}\vartheta_{X_s}''\rangle\right]_{\reg}
=\frac{2}{p'_{X_s}}
\Big\{&\:\frac{c}{16}\left([\tp_{X_s}']^2+\tp_{X_s}\tp_{X_s}''\right)\1\\
\+\frac{1}{2}\tp_{X_s}''\langle\vartheta_{X_s}\rangle
+\frac{1}{2}\tp_{X_s}'\langle\vartheta_{X_s}'\rangle
+\frac{1}{4}\tp_{X_s}\langle\vartheta_{X_s}''\rangle\\
\+\frac{3c}{40}[\tp_{X_s}']^2\1
+\frac{c}{40}\tp_{X_s}\tp_{X_s}''\1\\
\+\frac{6}{5}\tp_{X_s}''\langle\vartheta_{X_s}\rangle
+\frac{9}{10}\tp_{X_s}'\langle\vartheta_{X_s}'\rangle
-\frac{1}{5}\tp_{X_s}\langle\vartheta_{X_s}''\rangle\\
\-\langle\vartheta_{X_s}'\vartheta_{X_s}'\rangle_r\Big\}\\
=\frac{2}{p'_{X_s}}
\Big\{&\:\left(\frac{11c}{80}[\tp_{X_s}']^2+\frac{7c}{80}\tp_{X_s}\tp_{X_s}''\right)\1\\
\+\frac{17}{10}\tp_{X_s}''\langle\vartheta_{X_s}\rangle
+\frac{7}{5}\tp_{X_s}'\langle\vartheta_{X_s}'\rangle
+\frac{1}{20}\tp_{X_s}\langle\vartheta_{X_s}''\rangle\\
\-\langle\vartheta_{X_s}'\vartheta_{X_s}'\rangle_r\Big\}
\end{align*}
Translate back into untwiddled:
\begin{align*}
\frac{2}{p'_{X_s}}\left[\langle\vartheta_{X_s}\vartheta_{X_s}''\rangle\right]_{\reg}
=\frac{2}{p'_{X_s}}
\Big\{&\:\left(\frac{11c}{36\cdot80}[p_{X_s}^{(3)}]^2
+\frac{7c}{12\cdot160}p_{X_s}''p_{X_s}^{(4)}\right)\1\\
\+\frac{17}{120}p_{X_s}^{(4)}\langle\vartheta_{X_s}\rangle
+\frac{7}{30}p_{X_s}^{(3)}\langle\vartheta_{X_s}'\rangle
+\frac{1}{40}p_{X_s}''\langle\vartheta_{X_s}''\rangle\\
\-\langle\vartheta_{X_s}'\vartheta_{X_s}'\rangle_r\Big\}
\end{align*}

\section{Proof for the second derivative of the $3$-point function}
\label{proof: second derivative of the $3$-point function}

By the OPE for $\vartheta$ and the fact that $\langle\:\rangle$ is compatible with it, 
\begin{align}\label{eq. for <psi vartheta>} 
\langle\psi_2\vartheta_3\rangle
+O(x_1-x_2)
\=\langle\vartheta_1\vartheta_2\vartheta_3\rangle
-\frac{c}{32}f_{12}^2\langle\vartheta_3\rangle
-\frac{1}{4}f_{12}\left\{\langle\vartheta_1\vartheta_3\rangle+\langle\vartheta_2\vartheta_3\rangle\right\}\nn\\
\=\frac{c}{32}
\left\{
f_{23}^2\langle\vartheta_1\rangle
+f_{13}^2\langle\vartheta_2\rangle
\right\}\nn\\
\+\frac{1}{4}
f_{12}
\left\{\langle\vartheta_1\vartheta_3\rangle_r-\langle\vartheta_1\vartheta_3\rangle
+\langle\vartheta_2\vartheta_3\rangle_r-\langle\vartheta_2\vartheta_3\rangle\right\}\nn\\
\+\frac{1}{4}
f_{23}\left\{\langle\vartheta_1\vartheta_2\rangle_r+\langle\vartheta_3\vartheta_1\rangle_r\right\}\nn\\
\+\frac{1}{4}
f_{13}\left\{\langle\vartheta_1\vartheta_2\rangle_r+\langle\vartheta_2\vartheta_3\rangle_r\right\}\nn\\
\+\frac{1}{16}
f_{12}f_{23}
\left\{\:
\langle\vartheta_1\rangle
+\langle\vartheta_3\rangle\right\}\nn\\
\+\frac{1}{16}
f_{23}f_{31}
\left\{\:
\langle\vartheta_1\rangle
+\langle\vartheta_2\rangle\right\}\nn\\
\+\frac{1}{16}
f_{12}f_{31}
\left\{\:
 \langle\vartheta_2\rangle
+\langle\vartheta_3\rangle\right\}\nn\\
\+\frac{c}{4^3}
f_{12}f_{23}f_{31}\1\nn\\
\+\langle\vartheta_1\vartheta_2\vartheta_3\rangle_r
\end{align}
On the other hand,
\begin{align*}
\langle\psi_2\vartheta_3\rangle
\=-\frac{c}{480}\left(p_2'p_2^{(3)}-\frac{3}{2}[p_2'']^2\right)\langle\vartheta_3\rangle 
+\frac{1}{5}p_2''\langle\vartheta_2\vartheta_3\rangle 
-\frac{1}{10}p_2'\langle\vartheta_2'\vartheta_3\rangle 
-\frac{1}{5}p_2\langle\vartheta_2''\vartheta_3\rangle 
\end{align*}
We consider and $X_1=X_2$ ($s=1$) and  
\begin{displaymath}
p_x=(x-X_2)^2\tp_x\:,
\quad  
f_{xX_s}=\tp_x\:.
\end{displaymath}
We denote by $\tp_{X_s}'=\frac{d}{dx_3}|_{x_3=X_s}\tp_3$, etc.
Solving eq.\ (\ref{eq. for <psi vartheta>}),
evaluated at $x_1=x_2=X_s\:(=X_1=X_2)$, for $\langle\vartheta_{X_s}\vartheta_{X_s}\vartheta_3\rangle_r$,
\begin{align*}
\langle\vartheta_{X_s}\vartheta_{X_s}\vartheta_3\rangle_r
\=-\frac{c}{4^3}
\tp_{X_s}\tp_3^2\1\\ 
\+\frac{c}{80}\tp_{X_s}^2\langle\vartheta_3\rangle
+\frac{3}{20}\tp_3^2\langle\vartheta_{X_s}\rangle
-\frac{1}{8}\tp_{X_s}\tp_3
\left\{\:
\langle\vartheta_{X_s}\rangle
+\langle\vartheta_3\rangle\right\}\\
\+\frac{9}{10}\tp_{X_s}\langle\vartheta_{X_s}\vartheta_3\rangle
-\frac{1}{2}\tp_{X_s}\langle\vartheta_{X_s}\vartheta_3\rangle_r
-\frac{1}{2}
\tp_3\left\{\langle\vartheta_{X_s}\vartheta_{X_s}\rangle_r+\langle\vartheta_3\vartheta_{X_s}\rangle_r\right\}
\:, 
\end{align*}
\begin{align*}
\frac{d}{dx_3}\langle\vartheta_{X_s}\vartheta_{X_s}\vartheta_3\rangle_r
\=-\frac{c}{32}
\tp_{X_s}\tp_3\tp_3'\1
\\ 
\+\frac{c}{80}\tp_{X_s}^2\langle\vartheta_3'\rangle
+\frac{3}{10}\tp_3\tp_3'\langle\vartheta_{X_s}\rangle
-\frac{1}{8}\tp_{X_s}\tp_3'
\left\{\:
\langle\vartheta_{X_s}\rangle
+\langle\vartheta_3\rangle\right\}
-\frac{1}{8}\tp_{X_s}\tp_3\langle\vartheta_3'\rangle\\
\+\frac{9}{10}\tp_{X_s}\langle\vartheta_{X_s}\vartheta_3'\rangle
-\frac{1}{2}\tp_{X_s}\langle\vartheta_{X_s}\vartheta_3'\rangle_r
-\frac{1}{2}
\tp_3'\left\{\langle\vartheta_{X_s}\vartheta_{X_s}\rangle_r+\langle\vartheta_3\vartheta_{X_s}\rangle_r\right\}
-\frac{1}{2}
\tp_3\langle\vartheta_3'\vartheta_{X_s}\rangle_r
\:, 
\end{align*}
and
\begin{align*}
\frac{d^2}{dx_3^2}&\langle\vartheta_{X_s}\vartheta_{X_s}\vartheta_3\rangle_r\\
\=-\frac{c}{32}\tp_{X_s}[\tp_3']^2\1
-\frac{c}{32}\tp_{X_s}\tp_3\tp_3''\1\\ 
\+\frac{c}{80}\tp_{X_s}^2\langle\vartheta_3''\rangle
+\frac{3}{10}[\tp_3']^2\langle\vartheta_{X_s}\rangle
+\frac{3}{10}\tp_3\tp_3''\langle\vartheta_{X_s}\rangle
-\frac{1}{8}\tp_{X_s}\tp_3''
\left\{\:
\langle\vartheta_{X_s}\rangle
+\langle\vartheta_3\rangle\right\}
-\frac{1}{4}\tp_{X_s}\tp_3'\langle\vartheta_3'\rangle
-\frac{1}{8}\tp_{X_s}\tp_3\langle\vartheta_3''\rangle\\
\+\frac{9}{10}\tp_{X_s}\langle\vartheta_{X_s}\vartheta_3''\rangle
-\frac{1}{2}\tp_{X_s}\langle\vartheta_{X_s}\vartheta_3''\rangle_r
-\frac{1}{2}\tp_3''\left\{\langle\vartheta_{X_s}\vartheta_{X_s}\rangle_r+\langle\vartheta_3\vartheta_{X_s}\rangle_r\right\}
-\tp_3'\langle\vartheta_3'\vartheta_{X_s}\rangle_r
-\frac{1}{2}\tp_3\langle\vartheta_3''\vartheta_{X_s}\rangle_r
\:. 
\end{align*}
Thus
\begin{align*}
\frac{d^2}{dx_3^2}|_{x_3=X_s}&\langle\vartheta_{X_s}\vartheta_{X_s}\vartheta_3\rangle_r\\
\=-\frac{c}{32}\tp_{X_s}[\tp_{X_s}']^2\1
-\frac{c}{32}\tp_{X_s}^2\tp_{X_s}''\1\\ 
\-\frac{9}{50}\tp_{X_s}^2\langle\vartheta_{X_s}''\rangle
+\frac{3}{10}[\tp_{X_s}']^2\langle\vartheta_{X_s}\rangle
+\frac{1}{20}\tp_{X_s}\tp_{X_s}''\langle\vartheta_{X_s}\rangle
-\frac{1}{4}\tp_{X_s}\tp_{X_s}'\langle\vartheta_{X_s}'\rangle\\
\+\frac{9}{10}\tp_{X_s}\langle\vartheta_{X_s}\vartheta_{X_s}''\rangle
-\tp_{X_s}\langle\vartheta_{X_s}\vartheta_{X_s}''\rangle_r
-\tp_{X_s}''\Psi_{X_s}
-\tp_{X_s}'\langle\vartheta_{X_s}'\vartheta_{X_s}\rangle_r
\:. 
\end{align*}
Now
\begin{align*}
\langle\vartheta_{X_s}\vartheta_{X_s}''\rangle
\=\frac{c}{16}\1([\tp_{X_s}']^2+\tp_{X_s}\tp_{X_s}'')
+\frac{1}{2}\tp_{X_s}''\langle\vartheta_{X_s}\rangle
+\frac{1}{2}\tp_{X_s}'\langle\vartheta_{X_s}'\rangle
+\frac{1}{4}\tp_{X_s}\langle\vartheta_{X_s}''\rangle
+\langle\vartheta_{X_s}\vartheta_{X_s}''\rangle_r
\end{align*}
and
\begin{align*}
\psi_{X_s}
\=\frac{c}{80}\tp_{X_s}^2\1
+\frac{2}{5}\tp_{X_s}\langle\vartheta_{X_s}\rangle\:,
\end{align*}
so
\begin{align*}
\frac{d^2}{dx_3^2}|_{x_3=X_s}\langle\vartheta_{X_s}\vartheta_{X_s}\vartheta_3\rangle_r
\=\frac{c}{80}\left(\tp_{X_s}^2\tp_{X_s}''+2\tp_{X_s}[\tp_{X_s}']^2\right)\1\\ 
\+\frac{9}{200}\tp_{X_s}^2\langle\vartheta_{X_s}''\rangle
+\frac{1}{10}\left(\tp_{X_s}\tp_{X_s}''+3[\tp_{X_s}']^2\right)\langle\vartheta_{X_s}\rangle
+\frac{1}{5}\tp_{X_s}\tp_{X_s}'\langle\vartheta_{X_s}'\rangle\\
\-\frac{1}{10}\tp_{X_s}\langle\vartheta_{X_s}\vartheta_{X_s}''\rangle_r
-\tp_{X_s}'\langle\vartheta_{X_s}'\vartheta_{X_s}\rangle_r
\:. 
\end{align*}

\section{Proof of Claim \ref{claim: effect of the lnear fractiona trsf which maps the xi k to 0,1,infty, respectively, on X k, up to order eps to the power of 6}}
\label{proof: claim: effect of the lnear fractiona trsf which maps the xi k to 0,1,infty, respectively, on X k, up to order eps to the power of 6}

Let
\begin{align*}
X_0
\=\xi_0
\:\left(1+a_1\eps^4\left(\xi_0^2-2\ta_1\right)+\left(a_2\xi_0^3-5\ta_1a_2\xi_0-3a_2\ta_2\right)\eps^6+O(\eps^8)\right)^{-1}\\
\=\xi_0\sum_{k=0}^{\infty}\left(-a_1\eps^4\left(\xi_0^2-2\ta_1\right)-\left(a_2\xi_0^3-5\ta_1a_2\xi_0-3a_2\ta_2\right)\eps^6+O(\eps^8)\right)^k
\:,
\end{align*}
and so for $k=0,1,2$,
\begin{align*}
X_k
\=\xi_k
\:\left(1-a_1\eps^4\left(\xi_k^2-2\ta_1\right)-a_2\eps^6\left(\xi_k^3-5\ta_1\xi_k-3\ta_2\right)+O(\eps^8)\right)
\:.
\end{align*}
We have 
\begin{align*}
X_1-X_2
\=(\xi_1-\xi_2)
-a_1\eps^4\left(\xi_1^3-\xi_2^3-2\ta_1(\xi_1-\xi_2)\right)
-\left(a_2(\xi_1^4-\xi_2^4)-5\ta_1a_2(\xi_1^2-\xi_2^2)-3a_2\ta_2(\xi_1-\xi_2)\right)\eps^6+O(\eps^8)\\
X_1-X_0
\=(\xi_1-\xi_0)
-a_1\eps^4\left(\xi_1^3-\xi_0^3-2\ta_1(\xi_1-\xi_0)\right)
-\left(a_2(\xi_1^4-\xi_0^4)-5\ta_1a_2(\xi_1^2-\xi_0^2)-3a_2\ta_2(\xi_1-\xi_0)\right)\eps^6+O(\eps^8)
\:.
\end{align*}
So
\begin{align*}
\frac{X_1-X_2}{X_1-X_0}
\=(\xi_1-\xi_2)
\Big(
1
-a_1\eps^4\left(\frac{\xi_1^3-\xi_2^3}{\xi_1-\xi_2}-2\ta_1\right)
-a_2\eps^6\left(\frac{\xi_1^4-\xi_2^4}{\xi_1-\xi_2}-5\ta_1\frac{\xi_1^2-\xi_2^2}{\xi_1-\xi_2}-3\ta_2\right)+O(\eps^8)
\Big)\times\\
&\times\frac{1}{\xi_1-\xi_0}
\sum_{k=0}^{\infty}
\Big(a_1\eps^4\left(\frac{\xi_1^3-\xi_0^3}{\xi_1-\xi_0}-2\ta_1\right)
+a_2\eps^6\left(\frac{\xi_1^4-\xi_0^4}{\xi_1-\xi_0}-5\ta_1\frac{\xi_1^2-\xi_0^2}{\xi_1-\xi_0}-3\ta_2\right)+O(\eps^8)
\Big)^k\\
\=\frac{\xi_1-\xi_2}{\xi_1-\xi_0}
\Big(
1
-a_1\eps^4\left(\frac{\xi_1^3-\xi_2^3}{\xi_1-\xi_2}-2\ta_1\right)
-a_2\eps^6\left(\frac{\xi_1^4-\xi_2^4}{\xi_1-\xi_2}-5\ta_1\frac{\xi_1^2-\xi_2^2}{\xi_1-\xi_2}-3\ta_2\right)+O(\eps^8)
\Big)\times\\
&\times
\Big(1+a_1\eps^4\left(\frac{\xi_1^3-\xi_0^3}{\xi_1-\xi_0}-2\ta_1\right)
+a_2\eps^6\left(\frac{\xi_1^4-\xi_0^4}{\xi_1-\xi_0}-5\ta_1\frac{\xi_1^2-\xi_0^2}{\xi_1-\xi_0}-3\ta_2\right)+O(\eps^8)
\Big)\\
\=\frac{\xi_1-\xi_2}{\xi_1-\xi_0}
\Big(
1
+a_1\eps^4
\Big(\frac{\xi_1^3-\xi_0^3}{\xi_1-\xi_0}-\frac{\xi_1^3-\xi_2^3}{\xi_1-\xi_2}\Big)
+a_2\eps^6
\Big(
\frac{\xi_1^4-\xi_0^4}{\xi_1-\xi_0}-\frac{\xi_1^4-\xi_2^4}{\xi_1-\xi_2}
-5\ta_1(\xi_0-\xi_2)
\Big)
+O(\eps^8)
\Big)
\:. 
\end{align*}
Here
\begin{align*}
\frac{\xi_1^3-\xi_0^3}{\xi_1-\xi_0}-\frac{\xi_1^3-\xi_2^3}{\xi_1-\xi_2}
\=\xi_1^2+\xi_0\xi_1+\xi_0^2-(\xi_1^2+\xi_1\xi_2+\xi_2^2)\\
\=\xi_0\xi_1+\xi_0^2-\xi_1\xi_2-\xi_2^2\\
\=(\xi_0-\xi_2)(\xi_0+\xi_1+\xi_2)
=0\:,
\end{align*}
by eq.\ (\ref{eq: sum of half-periods is zero}) (nicely, this vanishing also works for the other combinations of $\frac{X_i-X_j}{X_i-X_k}$.)
We also note that $\xi_0-\xi_2=\frac{1}{4}\vartheta_4^4$.
Moreover,
\begin{align*}
\frac{x-X_0}{x-X_2} 
\=(x-X_0)\Big(x-
\xi_2\Big(
1
-a_1\eps^4\left(\xi_2^2-2\ta_1\right)
-a_2\eps^6\left(\xi_2^3-5\ta_1\xi_2-3\ta_2\right)+O(\eps^8)\Big)\Big)^{-1}\\
\=
\Big(
1
-\frac{\xi_0}{x}
\left(1-a_1\eps^4\left(\xi_0^2-2\ta_1\right)+O(\eps^6)\right)
\Big)\times\\
&\times\Big(
1
+\sum_{k\geq 1}
\left(\frac{\xi_2}{x}\right)^k
\Big(
1
-ka_1\eps^4\left(\xi_2^2-2\ta_1\right)
+O(\eps^6)\Big)
\Big)\\
\=1-\frac{\xi_0}{x}
\left(1-a_1\eps^4\left(\xi_0^2-2\ta_1\right)+O(\eps^6)\right)
+\sum_{k\geq 1}
\left(\frac{\xi_2}{x}\right)^k
\Big(
1
-ka_1\eps^4\left(\xi_2^2-2\ta_1\right)
+O(\eps^6)
\Big)\\
\-\frac{\xi_0}{x}
\left(1-a_1\eps^4\left(\xi_0^2-2\ta_1\right)+O(\eps^6)\right)\sum_{k\geq 1}
\left(\frac{\xi_2}{x}\right)^k
\Big(
1
-ka_1\eps^4\left(\xi_2^2-2\ta_1\right)
+O(\eps^6)
\Big)\:.
\end{align*}
So
\begin{align*}
\frac{(\eps^2\hat{X}_k)^{-1}-X_0}{(\eps^2\hat{X}_k)^{-1}-X_2} 
\=1-\xi_0\eps^2\hat{X}_k
\left(1-a_1\eps^4\left(\xi_0^2-2\ta_1\right)+O(\eps^6)\right)
+\sum_{m\geq 1}
\left(\xi_2\eps^2\hat{X}_k\right)^m
\Big(
1
-ma_1\eps^4\left(\xi_2^2-2\ta_1\right)
+O(\eps^6)
\Big)\\
\-\xi_0\eps^2\hat{X}_k
\left(1-a_1\eps^4\left(\xi_0^2-2\ta_1\right)+O(\eps^6)\right)
\sum_{m\geq 1}
\left(\xi_2\eps^2\hat{X}_k\right)^m
\Big(
1
-ma_1\eps^4\left(\xi_2^2-2\ta_1\right)
+O(\eps^6)
\Big)
\\
\=1
+\eps^2\hat{X}_k(\xi_2-\xi_0)
+\eps^4\hat{X}_k^2\xi_2(\xi_2-\xi_0)
+O(\eps^6)\:.
\end{align*}
So the linear fractional transformation 
\begin{align*}
x\mapsto
f(x)
=\frac{X_1-X_2}{X_1-X_0}\frac{x-X_0}{x-X_2}
\: 
\end{align*}
maps $X_0,X_1,X_2$ to $0,1,\infty$, respectively, and $X_{k+3}$ ($k=0,1,2$) to
\begin{align*}
f\left(\frac{1}{\eps^2\hat{X}_k}\right)
\=\frac{\xi_1-\xi_2}{\xi_1-\xi_0}
\Big(
1
+O(\eps^6)
\Big)
\Big(
1
+\eps^2\hat{X}_k\Big(\xi_2-\xi_0\Big)
+\xi_2^2\eps^4\hat{X}_k^2
+\eps^4\hat{X}_k^2\xi_2(\xi_2-\xi_0)
+O(\eps^6)\Big)
\Big)\\
\=\frac{\vartheta_3^4}{\vartheta_2^4}
\Big(
1
+O(\eps^6)
\Big)
\Big(
1
-\frac{\vartheta_4^4}{4}\eps^2\hat{X}_k
+\eps^4\xi_2^2\hat{X}_k^2
+O(\eps^6)
\Big)\:.
\end{align*}
On the other hand, the linear fractional transformation
\begin{align*}
x\mapsto
f(x)
=\frac{\xi_1-\xi_2}{\xi_1-\xi_0}\frac{x-\xi_0}{x-\xi_2}
\: 
\end{align*}
maps $\xi_0,\xi_1,\xi_2$ to $0,1,\infty$, respectively, and maps $\xi_{k+3}$ ($k=0,1,2$) to 
\begin{align*}
f\left(\frac{1}{\eps^2\hat{\xi_k}}\right)
\=\frac{\xi_1-\xi_2}{\xi_1-\xi_0}(1-\eps^2\xi_0\hat{\xi}_k)\sum_{m=0}^{\infty}(\eps^2\xi_2\hat{\xi}_k)^m\nn\\
\=\frac{\xi_1-\xi_2}{\xi_1-\xi_0}
\Big(1
-\eps^2\hat{\xi}_k(\xi_0-\xi_2)
-\eps^4\hat{\xi}_k^2\xi_2(\xi_0-\xi_2))
+O(\eps^6)\:.
\end{align*}
Here $\frac{\xi_1-\xi_2}{\xi_1-\xi_0}=\frac{\vartheta_3^4}{\vartheta_2^4}$ and $\xi_0-\xi_2=\frac{1}{4}\vartheta_4^4$.

\section{Proof of Claim \ref{claim: the quotient of X3-X4 and X3-X5}}
\label{appendix: Proof of Claim: the quotient of X3-X4 and X3-X5}

We have
\begin{align*}
X^{3,j,3,\ell}_{2,j,2,\ell}
-X^{3,j',3,\ell'}_{2,j',2,\ell'}
=\frac{\vartheta^4_{3,\Omega_{11}}}{\vartheta^4_{2,\Omega_{11}}}
\Big(R^{3,j,3,\ell}_{2,j,2,\ell}-R^{3,j',3,\ell'}_{2,j',2,\ell'}\Big)
\:,
\end{align*}
where either $j=j'=3$ (the case $X_3-X_5$) or $\ell=j'=2$ (the case $X_3-X_4$) or $\ell=\ell'=4$ (the case $X_4-X_5$).
The case $X_3-X_4$: Here $\ell=j'=2$, and
\begin{align*}
R^{3,3,3,2}_{2,3,2,2}-R^{3,2,3,4}_{2,2,2,4}
\=4\nu^2\:(R^{(1)}_{3,3}-R^{(1)}_{2,3}-R^{(1)}_{3,4}+R^{(1)}_{2,4})\\
\+4\nu^4\left(4R^{(1)}_{3,3}R^{(1)}_{3,2}+4R^{(1)}_{2,3}R^{(1)}_{2,2}
+\left[R^{(1)}_{3,3}\right]^2+3\left[R^{(1)}_{2,3}\right]^2\right)
-4\nu^4\left(4R^{(1)}_{3,2}R^{(1)}_{3,4}+4R^{(1)}_{2,2}R^{(1)}_{2,4}
+\left[R^{(1)}_{3,2}\right]^2+3\left[R^{(1)}_{2,2}\right]^2\right)\\
\-16\nu^4\left(R^{(1)}_{3,3}\left(R^{(1)}_{2,3}+R^{(1)}_{2,2}\right)+R^{(1)}_{3,2}R^{(1)}_{2,3}\right)
+16\nu^4\left(R^{(1)}_{3,2}R^{(1)}_{2,4}+R^{(1)}_{3,4}\left(R^{(1)}_{2,2}+R^{(1)}_{2,4}\right)\right)\\
\+\frac{4}{3}\nu^4
\left(
R^{(2)}_{3,3}
-R^{(2)}_{2,3}
-R^{(2)}_{3,4}
+R^{(2)}_{2,4}
\right)\\
\+O(\nu^6)\:.
\end{align*}
The case $X_3-X_5$: Here $j=j'=3$, and
\begin{align*}
R^{3,3,3,2}_{2,3,2,2}-R^{3,3,3,4}_{2,3,2,4}
\=
4\nu^2\:(R^{(1)}_{3,2}-R^{(1)}_{3,4}+R^{(1)}_{2,4}-R^{(1)}_{2,2})\\
\+4\nu^4\left(4R^{(1)}_{3,3}R^{(1)}_{3,2}+4R^{(1)}_{2,3}R^{(1)}_{2,2}
+\left[R^{(1)}_{3,2}\right]^2+3\left[R^{(1)}_{2,2}\right]^2\right)
-4\nu^4\left(4R^{(1)}_{3,3}R^{(1)}_{3,4}+4R^{(1)}_{2,3}R^{(1)}_{2,4}
+\left[R^{(1)}_{3,4}\right]^2+3\left[R^{(1)}_{2,4}\right]^2\right)\\
\-16\nu^4\left(R^{(1)}_{3,3}R^{(1)}_{2,2}+R^{(1)}_{3,2}R^{(1)}_{2,3}\right)
+16\nu^4\left(R^{(1)}_{3,3}R^{(1)}_{2,4}+R^{(1)}_{3,4}R^{(1)}_{2,3}\right)\\
\+\frac{4}{3}\nu^4
\left(R^{(2)}_{3,2}-R^{(2)}_{2,2}
\right)
-\frac{4}{3}\nu^4
\left(R^{(2)}_{3,4}-R^{(2)}_{2,4}
\right)\\
\+O(\nu^6)\:.
\end{align*}
Now we have
\begin{align*}
\frac{X_3-X_4}{X_3-X_5}
=
\frac{R^{3,3,3,2}_{2,3,2,2}-R^{3,2,3,4}_{2,2,2,4}}{R^{3,3,3,2}_{2,3,2,2}-R^{3,3,3,4}_{2,3,2,4}}
\=
\frac{R^{(1)}_{3,3}-R^{(1)}_{2,3}-R^{(1)}_{3,4}+R^{(1)}_{2,4}+O(\nu^2)}
{R^{(1)}_{3,2}-R^{(1)}_{2,2}-R^{(1)}_{3,4}+R^{(1)}_{2,4}}
\left(1+O(\nu^2)\right)\\
\=
\frac{\Big(
\frac{\vartheta'_{3,\Omega_{11}}}{\vartheta_{3,\Omega_{11}}}
-\frac{\vartheta'_{2,\Omega_{11}}}{\vartheta_{2,\Omega_{11}}}\Big)
\Big(
\frac{\vartheta'_{3,\Omega_{22}}}{\vartheta_{3,\Omega_{22}}}
-\frac{\vartheta'_{4,\Omega_{22}}}{\vartheta_{4,\Omega_{22}}}
\Big)}
{\Big(
\frac{\vartheta'_{3,\Omega_{11}}}{\vartheta_{3,\Omega_{11}}}
-\frac{\vartheta'_{2,\Omega_{11}}}{\vartheta_{2,\Omega_{11}}}
\Big)
\Big(
\frac{\vartheta'_{2,\Omega_{22}}}{\vartheta_{2,\Omega_{22}}}
-\frac{\vartheta'_{4,\Omega_{22}}}{\vartheta_{4,\Omega_{22}}}
\Big)}
\left(1+O(\nu^2)\right)\\
\=
\frac{
\frac{\vartheta'_{3,\Omega_{22}}}{\vartheta_{3,\Omega_{22}}}
-\frac{\vartheta'_{4,\Omega_{22}}}{\vartheta_{4,\Omega_{22}}}
}
{\frac{\vartheta'_{2,\Omega_{22}}}{\vartheta_{2,\Omega_{22}}}
-\frac{\vartheta'_{4,\Omega_{22}}}{\vartheta_{4,\Omega_{22}}}}
(1+O(\nu^2))\\
\=\frac{\vartheta^4_{2,\Omega_{22}}}{\vartheta^4_{3,\Omega_{22}}}(1+O(\nu^2))
\end{align*}
Addendum: We have
\begin{displaymath}
\frac{X_3-X_4}{X_3-X_5}
=\frac{\vartheta^4_{2,\Omega_{22}}}{\vartheta^4_{3,\Omega_{22}}}(1+O(\nu^2))
=(16\rho_2^{1/2}+O(\rho_2))(1+O(\nu^2))
\:,
\end{displaymath}
so when $\rho_2$ is small, so is the distance between $X_3$ and $X_4$.

\pagebreak

\end{document}